%% file: 1.main.tex
\documentclass[11pt]{article}
\input{0.includes}
\usepackage{graphicx}  

\title{\bfseries Introspectively Envy-Free and Efficient Allocation of \\ Indivisible Mixed Manna}

\author{Siddharth Barman\thanks{Indian Institute of Science. {\tt barman@iisc.ac.in}} \and Paritosh Verma\thanks{University of Toronto. {\tt paritoshverma97@gmail.com}}}

\date{}

\begin{document}

\maketitle

\begin{abstract}
The existence of allocations that are simultaneously fair and efficient is a central inquiry in fair division. Prominent results show that these complementary desiderata can be achieved together for indivisible items with nonnegative values. In particular, for indivisible goods and additive valuations, there always exists an allocation that is both envy-free up to one item (\(\EFo\)) and fractionally Pareto efficient (\(\fPO\)) \cite{BKV2018}. An analogous \(\EFo\) and \(\fPO\) guarantee is now known for indivisible chores \cite{mahara2025existence}. For indivisible mixed manna---items whose values can be positive, negative, or zero---the compatibility of fairness and efficiency is more delicate. Indeed, recent work shows that \(\EFo\) and \(\fPO\) are incompatible in the additive mixed-manna setting \cite{mackenzie2026one}. This work establishes positive results at this fairness-efficiency frontier: through principled relaxations of \(\EFo\), we show that compatibility with \(\fPO\) can be restored even for indivisible mixed manna under additive valuations. As a principled relaxation of $\EFo$, our proposed fairness notion allows agents to resolve envy through single-item additions to or removals from their own bundles.

For indivisible mixed manna and additive valuations, we establish the existence of allocations that are \(\fPO\) and \emph{introspectively envy-free up to one item} (\(\IEFo\)). In an \(\IEFo\) allocation, each agent can eliminate its envy toward every other agent by either adding an item to or removing an item from its own bundle. The notion of \(\IEFo\) coincides with \(\EFo\) for indivisible chores, and hence our result generalizes the aforementioned existence guarantee for chores \cite{mahara2025existence}. Our techniques also enable us to prove the existence of allocations that are \(\fPO\) and satisfy a chiral version of \(\IEFo\), which we call \emph{extrospectively envy-free up to one item} (\(\XEFo\)). This result generalizes the existence guarantee of \(\EFo\) and \(\fPO\) allocations for goods \cite{BKV2018}, since \(\XEFo\) is equivalent to \(\EFo\) for goods. Thus, our results provide a unified approach that recovers and extends the \(\EFo\) and \(\fPO\) guarantees for both indivisible goods and indivisible chores.

Finally, we utilize our $\IEFo$ and $\fPO$ result for indivisible items, together with limit arguments, to provide a distinct proof of the known \(\EF\) and \(\fPO\) guarantee for divisible mixed manna \cite{Bogo2017}. Our results rely on novel structural insights and algorithmic ideas, and they underscore the power of the Knaster--Kuratowski--Mazurkiewicz (KKM) Theorem for proving the existence of fair and efficient allocations of indivisible and divisible items.
\end{abstract}

\thispagestyle{empty}
\newpage

\tableofcontents
\thispagestyle{empty}
\newpage
\setcounter{page}{1}

\input{2.intro}
\input{3.prelims}
\input{4.PhaseI}

\input{5.PhaseII}

\input{6.IEFo-fPO}
\input{7.EF-PO-div}
\section{Conclusion and Future Work}
The current work contributes to the fairness-efficiency frontier for indivisible mixed manna with additive valuations. In particular, we establish the existence of allocations that are both $\IEFo$ and satisfy the strong efficiency criterion $\fPO$, despite the incompatibility of $\EFo$ and $\fPO$ in this setting. In addition to the $\IEFo$ and $\fPO$ result for mixed manna, we provide a unified framework for establishing the $\EFo$ and $\PO$ guarantee for indivisible goods and indivisible chores, respectively. Our techniques also extend to the divisible-items setting, where we develop an alternative proof of the established $\EF$ and $\fPO$ guarantee for divisible mixed manna.

Resolving the $\EFo$ and $\PO$ question exactly remains an interesting direction for future work. Developing pseudo-polynomial time algorithms for finding fair and efficient allocations of indivisible mixed manna (or even indivisible chores) is another interesting direction.

\bibliographystyle{alpha}
\bibliography{references}

\appendix
\input{8.appendix}

\end{document}

%% file: 0.includes.tex
\usepackage[margin=0.8in]{geometry}
\usepackage{charter}

\usepackage{amsfonts,amssymb}
\usepackage{bbm}
\usepackage{mathrsfs}

\usepackage{amsthm,thmtools,thm-restate}
\usepackage{mathtools} 
\usepackage{amsmath} 
\usepackage[ruled,vlined]{algorithm2e}
\LinesNumbered
\usepackage{tabularx}
\usepackage{tabularx}

\usepackage[dvipsnames]{xcolor} 
\usepackage{tikz}
\usetikzlibrary{arrows.meta,positioning,shapes.geometric,calc}
 
\definecolor{ptblue}{RGB}{15,76,129} 
\definecolor{ptemerald}{HTML}{009473} 
\definecolor{bluegray}{rgb}{0.4, 0.6, 0.8}
\definecolor{ptilluminating}{HTML}{F5DF4D} 
\definecolor{ptgray}{HTML}{939597}
\usepackage{cprotect}
\definecolor{cobalt}{rgb}{0.0, 0.28, 0.67}
\usepackage{hyperref}
\hypersetup{
linktocpage,
colorlinks=true,
citecolor=cobalt, 
urlcolor=ptblue, 
linkcolor=cobalt, 
}

\usepackage{cleveref}

\usepackage{thm-restate}
\usepackage{enumitem}
\usepackage{verbatim}

\newcommand{\angleb}[1]{\langle #1 \rangle}
\crefname{ALC@line}{line}{lines} 

\DeclareMathOperator*{\argmax}{arg\,max}
\DeclareMathOperator*{\argmin}{arg\,min}

\newtheorem{theorem}{Theorem}

\newtheorem{corollary}[theorem]{Corollary}
\newtheorem{definition}{Definition}
\newtheorem{lemma}[theorem]{Lemma}
\newtheorem{proposition}{Proposition}
\newtheorem{claim}{Claim}

\usepackage[square,semicolon,numbers]{natbib}
\newcommand\numberthis{\addtocounter{equation}{1}\tag{\theequation}}

\newcommand{\alloc}{\mathcal{A}}
\newcommand{\allocb}{\mathcal{B}}

\newcommand{\EFEo}{\mathrm{IEF}1}
\newcommand{\IEFo}{\mathrm{IEF}1}

\newcommand{\XEFo}{\mathrm{XEF}1}

\newcommand{\EF}{\mathrm{EF}}
\newcommand{\PO}{\mathrm{PO}}
\newcommand{\fPO}{\mathrm{fPO}}
\newcommand{\EFo}{\mathrm{EF}1}
\newcommand{\OPT}{\mathrm{OPT}}
\newcommand{\SW}{\mathrm{SW}}
\newcommand{\Propo}{\mathrm{Prop}1}

\newcommand{\LP}{\textbf{\textrm{LP}}}
\newcommand{\Dual}{\textbf{\textrm{Dual-LP}}}

\newcommand{\vbar}{\widehat{v}}
\newcommand{\mbar}{\widehat{m}}
\newcommand{\barI}{\widehat{\mathcal{I}}}
\newcommand{\supp}{\mathrm{supp}}
\newcommand{\pr}{\mathrm{par}}
\newcommand{\ch}{\mathrm{ch}}
\newcommand{\cho}{\overline{\mathrm{ch}}}

\newcommand{\calI}{\mathcal{I}}

\newcommand{\calA}{\mathcal{A}}
\newcommand{\calD}{\mathcal{D}}

\newcommand{\calB}{\mathcal{B}}

\newcommand{\calX}{\mathcal{X}}
\newcommand{\calY}{\mathcal{Y}}
\newcommand{\calZ}{\mathcal{Z}}
\newcommand{\calF}{\mathcal{F}}
\newcommand{\calS}{\mathcal{S}}
\newcommand{\calQ}{\mathcal{Q}}
\newcommand{\calO}{\mathcal{O}}
\newcommand{\calM}{\mathcal{M}}
\newcommand{\calL}{\mathcal{L}}

%% file: 2.intro.tex
\section{Introduction}
The existence of allocations that simultaneously achieve fairness and efficiency lies at the core of the fair division literature \cite{BT96fair, moulinHandbookComputational16, moulinFairDivision03}. This question has been extensively studied over the past several decades, particularly in the context of goods---items with nonnegative values. Indeed, the classic result of Varian \cite{varianEquityenvy74} establishes that, for divisible goods, a fair and efficient allocation always exists. This result provides a normative foundation for competitive markets with equal incomes: for divisible goods, a competitive equilibrium allocation (under equal incomes) is both envy-free---i.e., no agent values another's bundle over their own \cite{foleyResourceallocation67}---and Pareto efficient.\footnote{In the context of divisible goods, the notion of Pareto optimality is the same as fractional Pareto optimality.} In fact, such an allocation can be found by maximizing the Nash social welfare (the geometric mean of agents' valuations), equivalently by solving the Eisenberg--Gale convex program \cite{eisenberg1959consensus}.

While the compatibility of fairness and efficiency is well studied for goods, it remains less resolved in the equally well-motivated settings of chores (negatively valued items) and, more generally, mixed manna (items whose valuations may vary in sign across agents). Notably, in contrast to the classic result for divisible goods \cite{varianEquityenvy74}, an analogous guarantee of fairness and efficiency for divisible mixed manna was established only relatively recently \cite{Bogo2017}. The results in \cite{Bogo2017} underscore that the mixed-manna setting is significantly more complex than goods or even chores alone: for divisible mixed manna, unlike in the case of divisible goods, the set of competitive equilibria may be non-convex and even disconnected. Furthermore, although a competitive equilibrium allocation for divisible goods can be efficiently computed via a convex program, finding such an allocation in the divisible mixed-manna setting is ${\rm PPAD}$-hard \cite{chaudhury2021competitive}.

The current work studies the compatibility of fairness and efficiency for indivisible items, i.e., for resources and chores that cannot be fractionally assigned among the agents. For indivisible items as well, the distinction between goods, chores, and mixed manna is fundamental. For indivisible goods, a prominent result of Caragiannis et al.~\cite{caragiannis2019unreasonable} establishes that, under additive valuations, there always exists an allocation that is both Pareto efficient ($\PO$) and fair. Since an envy-free allocation is not guaranteed to exist for indivisible items, fairness in discrete fair division is captured through meaningful relaxations, most notably, envy-freeness up to one item ($\EFo$) \cite{budishCombinatorialAssignment11}. An allocation of indivisible mixed manna is said to be $\EFo$ if, whenever an agent $i$ envies another agent $j$, the envy can be eliminated by either the (hypothetical) removal of a single good from $j$'s bundle or the (hypothetical) removal of a chore from $i$'s bundle.\footnote{When all the items are goods, as in \cite{caragiannis2019unreasonable}, this definition reduces to the requirement that any existing envy can be eliminated by removing a good from the other agent's bundle.} In particular, Caragiannis et al.~\cite{caragiannis2019unreasonable} showed that, under additive valuations, an allocation of indivisible goods that maximizes the Nash social welfare is both $\EFo$ and $\PO$. The stronger existential guarantee of $\EFo$ together with fractional Pareto efficiency ($\fPO$) was subsequently established in \cite{BKV2018}.

The $\fPO$ criterion is particularly natural and well studied in this context: $\fPO$ strengthens $\PO$ by ruling out not only Pareto improvements through reassignments of indivisible items, but also Pareto improvements that become possible when items can be fractionally allocated. Furthermore, under additive valuations, $\fPO$ can be verified efficiently, in contrast to $\PO$, whose verification is $\rm{co\mbox{-}NP}$-hard \cite{de2009complexity}. It is also worth noting that, beyond welfare-based approaches confined to goods~\cite{caragiannis2019unreasonable,eckart2024fairness}, almost all existence results for $\PO$ allocations in indivisible-item settings have been established by instead achieving the stronger $\fPO$ guarantee. 

The existence of $\EFo$ and $\fPO$ allocations was known for indivisible goods \cite{BKV2018}, but, in contrast, for indivisible chores the existence of $\EFo$ and $\fPO$ (or $\PO$) allocations remained open for nearly a decade. In fact, until the recent resolution by Mahara~\cite{mahara2025existence} for indivisible chores, the existence question was open even among four agents with additive valuations.

Notably, however, the corresponding existence guarantee does not hold for indivisible mixed manna---items whose values can be positive, negative, or zero: under additive mixed-manna valuations, $\EFo$ and $\fPO$ are incompatible \cite{mackenzie2026one}. Thus, mixed manna marks a sharp departure from both goods and chores: the efficiency benchmark $\fPO$ is not attainable together with $\EFo$.

To overcome this incompatibility, we establish positive results showing that, even for indivisible mixed manna, compatibility with $\fPO$ can be restored through principled relaxations of $\EFo$ that allow agents to eliminate envy using single-item additions to or removals from their own bundles. In particular, we establish that, under additive valuations, there always exists an allocation of indivisible mixed manna that is $\fPO$ and \emph{introspectively envy-free up to one item} ($\IEFo$). In an $\IEFo$ allocation, each agent can eliminate its envy towards every other agent by either adding an item to or removing an item from its own bundle. Thus, $\IEFo$ is introspective: rather than modifying the envied agent's bundle, each agent hypothetically modifies its own bundle---adding a good or removing a chore---to eliminate its envy.

The notion of $\IEFo$ coincides with $\EFo$ when all the indivisible items are chores, and hence our result generalizes the existence guarantee for chores \cite{mahara2025existence}. We also apply our techniques to prove the existence of allocations that are $\fPO$ and satisfy a chiral version of the $\IEFo$ fairness guarantee, which we call \emph{extrospectively envy-free up to one item} ($\XEFo$). In an $\XEFo$ allocation, each agent $i$ can eliminate its envy towards any other agent $j$ by either removing a good from $j$'s bundle or adding a chore to $j$'s bundle. This result generalizes the existence guarantee of $\EFo$ and $\fPO$ allocations for goods, since $\XEFo$ and $\EFo$ are equivalent for goods. Thus, in addition to novel results for mixed manna, our techniques provide a unified approach that recovers the $\EFo$ and $\fPO$ guarantees for indivisible goods and indivisible chores, respectively.

Fairness and efficiency for indivisible mixed manna among $n$ agents is also considered in \cite{IM25} and \cite{barman2025fair}. The use of the Knaster--Kuratowski--Mazurkiewicz (KKM) theorem for the indivisible-items setting first appeared in \cite{IM25}, where the focus was on allocating mixed manna under $k$ category constraints. Their work shows that there exists a subset of $\min\{k+1,n\}(n-1)$ items and an allocation such that, by reallocating items only from this subset, one can make every agent envy-free while maintaining Pareto efficiency.  The independent and concurrent work \cite{barman2025fair} establishes an envy-freeness and Pareto-efficiency guarantee while requiring the reallocation of at most $(n-1)$ items. The current work strengthens this guarantee.\footnote{An $\IEFo$ and $\PO$ allocation is also $\mathrm{EFR}$-$(n-1)$ and $\PO$, where $\mathrm{EFR}$-$(n-1)$ is the notion studied in \cite{barman2025fair}.} Subsequent to \cite{IM25}, \cite{mahara2025existence} combined the KKM approach with additional primal-dual arguments to obtain the above-mentioned result for chores.

The current work advances the KKM-based framework and develops novel ideas; in addition to discrete fair division, the scope of our results also extends to divisible fair division. \Cref{subsection:techniques} highlights our technical contributions, and \Cref{subsection:novel_ideas} discusses the key ideas underlying our approach.

\subsection{Our Results}
\label{subsection:techniques}

Our results are summarized below.

\begin{itemize}

\item We first establish the existence of $\IEFo$ and $\PO$ allocations for indivisible mixed manna with additive valuations (\Cref{theorem:main_fair_efficient}). This result is of independent interest, and establishes a fairness-efficiency guarantee for mixed manna under the standard Pareto-efficiency notion.

\item {\bf Main Result.} We strengthen the efficiency guarantee from $\PO$ to $\fPO$: under additive valuations, there always exists an $\IEFo$ and $\fPO$ allocation of indivisible mixed manna (\Cref{theorem:addingFtoPO}). Recall that an (integral) allocation is fractionally Pareto efficient if it is not Pareto dominated even by a fractional allocation of the items; thus, $\fPO$ is a strictly stronger efficiency notion than $\PO$.

\item We also adapt our framework to establish a fairness guarantee based on a chiral variant of $\IEFo$. In particular, we consider allocations that are \emph{extrospectively envy-free up to one item} ($\XEFo$): each agent $i$ can eliminate its envy towards any other agent $j$ by either removing a good from $j$'s bundle or adding a chore to $j$'s bundle. We prove that $\XEFo$ and $\PO$ allocations always exist (\Cref{theorem:OutPO} in \Cref{section:argument-works-for-goods}), and, analogously to the $\IEFo$ case, establish the existence of $\XEFo$ and $\fPO$ allocations (\Cref{theorem:OutfPO}).

\item Our $\IEFo$ and $\fPO$ result (\Cref{theorem:addingFtoPO}), together with a limit argument, implies the existence of envy-free ($\EF$) and $\fPO$ allocations for divisible mixed manna (\Cref{theomem:divisibleEFnPO}). This provides an alternative proof, for additive valuations, of the existential guarantee established in \cite{Bogo2017}.
 
\item We also note that our $\IEFo$ and $\XEFo$ results naturally extend to settings wherein the agents have different entitlements (weights). The formal statements are presented in \Cref{thm:wief1-fpo,thm:wxef1-fpo}; we omit the proofs, as they closely follow those for the equal-entitlement setting.
\end{itemize}

A summary of our results and their implications is presented in \Cref{fig:results-overview}.
\input{results-picture}

\paragraph{Significance of $\IEFo$ and $\XEFo$.}
While the existence of $\EFo$ and $\PO$ allocations for indivisible mixed manna under additive valuations remains open, our $\IEFo$ and $\XEFo$ guarantees make meaningful progress on the fairness--efficiency frontier. In particular, $\IEFo$ and $\XEFo$ provide strong fairness guarantees while retaining $\fPO$. As noted above, $\EFo$ and $\fPO$ are incompatible for mixed manna~\cite[Theorem 4.3]{mackenzie2026one}. Thus, our results identify natural fairness guarantees that remain attainable at the stronger $\fPO$ efficiency benchmark, despite the impossibility of retaining $\EFo$.

Moreover, the incompatibility established in \cite{mackenzie2026one} highlights that, for the $\EFo$ and $\PO$ existence question in the mixed-manna setting, one cannot go via an $\fPO$ guarantee. At present, the only known techniques that achieve $\PO$ without necessarily achieving $\fPO$ rely on welfare maximization, such as Nash welfare maximization~\cite{caragiannis2019unreasonable}. However, prior work has shown that welfare maximization cannot provide any (even approximate) fairness guarantees for chores, and hence for mixed manna~\cite{eckart2024fairness}. Thus, while obtaining $\EFo$ and $\PO$ allocations for mixed manna remains an important open direction, our results provide meaningful fairness guarantees that can be achieved together with the strong efficiency criterion $\fPO$. On the conceptual front, $\IEFo$ and $\XEFo$ are supported by simple and clean interpretations and close connections to prominent fairness notions in the literature (\Cref{fig:results-overview}).

\subsection{Our Techniques}
\label{subsection:novel_ideas}

The novel technical ideas and algorithmic techniques developed in this work are outlined next.

For indivisible goods, $\EFo$ and efficiency guarantees were obtained either by maximizing the Nash social welfare \cite{caragiannis2019unreasonable} or via Fisher-market-based approaches \cite{BKV2018,mahara2024,garg2023computing}. It was shown in \cite{eckart2024fairness} that there does not exist a welfare function analogous to the Nash social welfare that guarantees $\EFo$ for chores (and, hence, for mixed manna). In addition, the market-based approach does not directly extend to the mixed-manna setting.

A successful alternative approach, which extends beyond the goods setting, was identified recently in \cite{IM25}, \cite{barman2025fair}, and \cite{mahara2025existence}. These works apply the Knaster--Kuratowski--Mazurkiewicz (KKM) theorem in the indivisible-items setting and build on this framework with additional linear programming and combinatorial insights. However, the $\EFo$ and $\fPO$ guarantee established in \cite{mahara2025existence} applies specifically to indivisible chores and does not extend to mixed manna.

Our first result establishes the existence of an $\IEFo$ and $\PO$ allocation of mixed manna under additive valuations (\Cref{theorem:main_fair_efficient}). Establishing this guarantee requires new ideas both in the application of the KKM theorem and in the subsequent combinatorial arguments. Below, we highlight these ideas in the order in which they arise in the proofs.

We then strengthen the efficiency guarantee from $\PO$ to $\fPO$ (\Cref{theorem:addingFtoPO}). This strengthening is obtained through a limiting argument applied to a sequence of suitably perturbed, non-degenerate instances; see \Cref{section:addingFtoPO} for details. The remainder of this subsection is devoted to the $\PO$ result.

\medskip 

\noindent
{\bf Invoking the KKM Theorem.} A key difficulty that arises in the mixed-manna setting is the presence of items that are valued at zero by some agents and positively or negatively by others. In the setting with only chores~\cite{mahara2025existence}, such items can be preprocessed: one can a priori remove such an item from consideration, identify a fair and efficient allocation for the remaining instance, and subsequently assign the item to any agent that values it at zero. This preprocessing preserves the fairness and efficiency guarantees in the all-chores setting. The KKM theorem in~\cite{mahara2025existence} is essentially used to find a component-wise positive weight vector for the agents, after which the final fair and efficient allocation is a weighted-welfare-maximizing allocation with respect to these weights. In the presence of zero-valued items, however, the covering condition required for the KKM theorem breaks down.\footnote{This is also one of the obstacles that prevents the result of~\cite{mahara2025existence} from being applied to the setting of goods.} This is a subtle yet crucial technical issue. To illustrate this, consider the following example, which shows that, in the absence of an auxiliary fix, no positive weights can result in an allocation that satisfies a desired fairness requirement:  Consider an instance with two agents and $m$ items. Agent $1$ values each item at $1$ and agent $2$ values each item at $0$. Under any positive weights, $w_1, w_2>0$, any weighted welfare-maximizing allocation assigns all the items to agent $1$. That is, under any such efficient allocation, agent $2$ receives an empty bundle. Hence, we get that agent $2$ can never be envy-free when considering prices (dual variables), i.e., agent $2$ can never be price envy-free (contrary to the desired property in \Cref{theorem:KKM-application}). 

To address the issue caused by zero-valued items, we construct a non-degenerate instance and introduce an auxiliary item with an identical, positive value for every agent. The auxiliary item ensures that, under suitable weights, every agent can receive at least one item, in particular the auxiliary good. This is a subtle, but key, step towards ensuring that the covering condition for the KKM theorem (\Cref{lemma:KKM-covering}) holds, even in the mixed manna setting. That is, even for mixed manna, we are able to invoke the KKM Theorem (\Cref{theorem:KKM}) for establishing the existence of a welfare-maximizing (efficient) and fair allocation (\Cref{theorem:KKM-application}). 

Obtaining our $\PO$ and $\XEFo$ result (\Cref{theorem:OutPO}), however, requires further changes to the KKM framework. In particular, the primal and dual linear programs, as well as the KKM sets ($\overline{C}_i$s) used in the $\XEFo$ argument, are novel. In the $\IEFo$ result, the KKM sets require agents to be price envy-free, which is a natural choice and also underlies the approach of~\cite{mahara2025existence}. In contrast, rather surprisingly, the KKM sets needed to obtain the $\XEFo$ guarantee (\Cref{definition:KKM-sets}) require each agent to have a bundle with the worst price; that is, each agent must weakly price-envy every other agent. While price envy-freeness is a standard notion in the literature, we are not aware of an analogue of this counterintuitive weak price-envy requirement or of its use to obtain downstream fairness guarantees.

Furthermore, in contrast to the present work and~\cite{mahara2025existence}, the KKM sets in~\cite{barman2025fair} and \cite{IM25} do not involve prices (dual variables) or price envy-freeness. Consequently, our invocation of the KKM theorem differs from that in~\cite{barman2025fair} and \cite{IM25}.

Finally, we additively perturb the non-zero values in the given fair-division instance to construct a non-degenerate instance. These perturbations provide crucial non-degeneracy properties (Lemmas \ref{lemma:three-types} and \ref{lemma:non-degeneracy}). The construction relies on a collection of interdependent parameters controlling the magnitude of the perturbations ($\varepsilon$), the value of the auxiliary item ($\lambda$), the weight shifts for the agents ($\eta$), and other quantities (e.g., $\omega$). These parameters must be carefully chosen, as the subsequent arguments depend delicately on their values and interrelationships. Thus, a key technical contribution is the construction of the non-degenerate instance with an auxiliary item and carefully selected parameters, which simultaneously addresses the challenges caused by zero-valued items and enables the application of the KKM theorem. \\

\noindent
{\bf Combinatorial Insights and the Augmenting Tree Algorithm.}
Having identified, via the KKM theorem, a weight vector $w^*$ and an associated collection of weighted-welfare-maximizing allocations that are price envy-free for the $n$ agents, we next constructively prove the existence of the desired fair and efficient allocation. \cite{mahara2025existence} follows a similar high-level approach, albeit without the additional complexities introduced by mixed manna. However, the subsequent argument and proofs in \cite{mahara2025existence} rely on all items being chores,\footnote{In particular, the proofs in Section 3.1 of~\cite{mahara2025existence} crucially rely on positive prices.} and therefore do not extend to mixed manna. To address this obstacle, we establish new structural insights and develop a novel method for consolidating the collection of $n$ price-envy-free allocations into a single fair allocation.

This consolidation is achieved through our \emph{augmenting tree algorithm}, which iteratively reassigns items by traversing a tree whose vertices consist of agents and items (the price graph). The augmenting tree algorithm in our work is notably different from those in prior works. It is more general, encompassing the chores setting as a special case, while its analysis is comparatively simpler.

The guarantee that these tree augmentations can be successfully performed rests on a key structural lemma (\Cref{lemma:forest-augmentation-existence}). We quantify progress using a potential defined in terms of a \emph{min--max price threshold}, which plays a central role in both the design and analysis of the algorithm. Despite the additional complexities introduced by mixed manna, both the design and analysis of our augmenting tree algorithm---in particular, the proof of termination---are simpler than those in~\cite{mahara2025existence}. This reflects a substantially different combinatorial strategy. For example, even when all items are chores, our augmenting tree algorithm requires at most $n$ augmentation steps, with each step making one agent introspectively envy-free up to one item. In contrast, the algorithm in~\cite{mahara2025existence} essentially requires $n^3$ iterations~\cite[Lemma 3.14]{mahara2025existence}, with each iteration aimed at decreasing a non-trivial potential function. Hence, when specialized to the all-chores setting, our approach provides a more direct proof of the existence of $\EFo$ and $\PO$ allocations.

In settings with only goods or only chores, adding an item to a bundle has a predictable effect on its price: it either increases the price of the bundle or decreases it, respectively. The presence of mixed manna breaks this monotonicity and makes the design of the augmentation procedure more challenging. We address this by treating the addition of a good as the analogue of deleting a chore. This principle is reflected throughout the design and analysis of our tree augmentations and simplifies the presentation by allowing us to use the unified $\triangle$ notation.

\subsection{Additional Related Work}
\label{section:add-related-work}

\noindent {\bf $\EFo$ for Indivisible Goods, Chores, and Mixed Manna.} 
The following prior works address $\EFo$ division of indivisible items, without the $\PO$ desideratum. For goods and under monotone valuations, an $\EFo$ allocation is guaranteed to exist and can be computed efficiently via the envy-cycle elimination algorithm~\cite{LMM04}. Furthermore, under additive valuations, the round-robin algorithm finds $\EFo$ allocations for both goods-only and chores-only settings. 

\cite{azizetalFairallocation22} established the existence and efficient computation of $\EFo$ allocations of mixed manna with additive valuations. \cite{BSV21} showed that $\EFo$ allocation for chores with monotone (decreasing) valuations can be efficiently computed via the elimination of top-trading envy cycles.  \\

\noindent
{\bf $\EFo$ and $\PO$ for Mixed Manna.} 
For the specific case of two agents, the existence of $\EFo$ and $\PO$ allocations of mixed manna was established in~\cite{azizetalFairallocation22}. This guarantee is also known particularly for identical or ternary valuations~\cite{aleksandrovwalshTwoAlgorithms20}. Also, for two agents and under category constraints, \cite{SHS23} considered the relaxation $\EF{[1,1]}$ and showed that mixed-manna allocations that are $\EF{[1,1]}$ and $\PO$ exist and can be computed efficiently.  For restricted additive valuations, in particular, \cite{livanos2022almost} obtained fairness and efficiency guarantees for mixed manna. Overall, prior to the current work, efficiency and up-to-one-item envy-freeness guarantees for mixed manna were not known in the general setting of arbitrarily many agents with additive valuations. \\

\noindent
{\bf $\EFo$ and $\PO$ for Chores.} Before \cite{mahara2025existence}, $\EFo$ and $\PO$ results for indivisible chores were known only in special cases. $\EFo$ and $\PO$ chore allocations, with additive valuations, were known to exist for two agents~\cite{azizetalFairallocation22}, for three agents~\cite{gargetalNewAlgorithms23}, and for three types of agents \cite{gargetalWeightedEF124}.  

Previous works have also established fairness and efficiency guarantees for restricted classes of valuations~\cite{ebadianetalHowfairly22,gargetalFairEfficientAAAI22,wuetalWeightedEF123,GargConstantFactorSTOC25,azizetalFairallocation23}. \\

\noindent 
{\bf $\Propo$ and $\PO$ for Mixed Manna.} Share-based fairness criteria have also been studied for indivisible items \cite{amanatidis2023fair}. In particular, an allocation of mixed manna is said to be proportional up to one item ($\Propo$) if each agent can, by adding a good or removing a chore from its bundle, ensure a value at least as much as its proportional share, i.e., at least $1/n$ times its value for the grand bundle \cite{azizetalFairallocation22}. Note that, in $\Propo$ (as in $\IEFo$), items are added or removed from an agent's own bundle. Moreover, under additive values, an $\IEFo$ allocation of mixed manna is also $\Propo$. Therefore, our $\IEFo$ and $\PO$ result recovers the existential guarantee for $\Propo$ and $\PO$ mixed-manna allocations, established in \cite{aziz2020polynomial}; see also \cite{branzei2024algorithms,barman2019proximity}. 

%% file: results-picture.tex
\begin{figure}[htbp]
  \centering
  \begin{tikzpicture}[
    scale=0.78,
    transform shape,
    >={Implies[length=2.5mm,width=3mm]},
    box/.style={
      draw, rounded corners=2pt, line width=0.6pt,
      align=center, inner sep=6pt,
      minimum width=38mm, minimum height=14mm, font=\small},
    main/.style={
      box, line width=1.1pt, double, double distance=1.6pt,
      minimum width=38mm, minimum height=18mm, fill=black!4},
    blob/.style={
      draw, ellipse, line width=0.9pt, align=center,
      inner sep=4pt, font=\small\scshape},
    imp/.style={->, double, double distance=2pt, line width=0.7pt,
                shorten >=2pt, shorten <=2pt},
    link/.style={dash pattern=on 1pt off 2.5pt, line width=0.7pt},
  ]
  \node[main] (X) at (0,0)
    {$\mathrm{IEF}1 + \mathrm{fPO}$\\[1pt] Mixed Manna\\[3pt]
     \Cref{theorem:addingFtoPO}};

  \node[main] (Y) at (10,0)
    {$\mathrm{XEF}1 + \mathrm{fPO}$\\[1pt] Mixed Manna\\[3pt]
     \Cref{theorem:OutfPO}};

  \node[blob] (our) at ($(X)!0.5!(Y)$) {Our\\ Results};

  \draw[link] (X) -- (our);
  \draw[link] (our) -- (Y);

  \node[box, anchor=north] (efdiv) at (5,4)
    {$\mathrm{EF} + \mathrm{fPO}$\\[1pt] Divisible Mixed Manna\\[3pt]
     \Cref{theomem:divisibleEFnPO}};

  \node[box, anchor=south] (chores) at (0,-4)
    {$\mathrm{EF}1 + \mathrm{fPO}$\\[1pt] Chores \\[1pt]~\cite{mahara2025existence}};

  \node[box, anchor=south] (prop1) at (5,-4)
    {$\mathrm{PROP}1 + \mathrm{fPO}$\\[1pt] Mixed Manna \\[1pt]~\cite{aziz2020polynomial}};

  \node[box, anchor=south] (goods) at (10,-4)
    {$\mathrm{EF}1 + \mathrm{fPO}$\\[1pt] Goods \\[1pt]~\cite{BKV2018}};

  \draw[imp] (X.north east) -- (efdiv.south west);
  \draw[imp] (Y.north west) -- (efdiv.south east);
  \draw[imp] (X.south)      -- (chores.north);
  \draw[imp] (X.south east) -- (prop1.north west);
  \draw[imp] (Y.south)      -- (goods.north);


  \end{tikzpicture}
  \caption{Overview of our results and their implications.}
  \label{fig:results-overview}
\end{figure}

%% file: 3.prelims.tex
\section{Preliminaries}\label{section:preliminaries}

\noindent
{\bf Fair Division Instance.} A fair division instance $\calI$ is denoted by a triple $\calI = \angleb{[n], [m], \{v_i\}_{i=1}^n}$; here, $m \in \mathbb{Z}_+$ indivisible items need to be partitioned among $n \in \mathbb{Z}_+$ agents with individual valuations $v_i$. The set of agents and items will be denoted by $[n] \coloneqq \{1, 2, \ldots, n\}$ and $[m] \coloneqq \{1, 2, \ldots, m\}$, respectively.  We will, throughout, write $v_i(t) \in \mathbb{R}$ to denote the value of each item $t \in [m]$ for agents $i \in [n]$. Furthermore, this work addresses agents with additive valuations, i.e., for each agent $i \in [n]$ and for any set of items $S \subseteq [m]$, the value $v_i(S) \coloneqq \sum_{t \in S} v_i(t)$. By convention, we have $v_i(\emptyset) = 0$ for each agent $i \in [n]$. We obtain results for fair division instances with mixed manna, i.e., the values, $v_i(t)$, can be positive, negative, or zero.   

For any subset $X \subseteq [m]$ and item $t \in [m]$, we will write $X + t$, as a shorthand, to denote the set $X \cup \{t\}$. Also, write $X - t$ to denote $X \setminus \{t\}$. For any two subsets $X,Y \subseteq [m]$, we will use $X\triangle Y$ to denote the symmetric set difference between them, $X \triangle Y \coloneqq  (X\setminus Y) \cup (Y \setminus X)$. Note that, here, if $Y$ is a singleton, then $\triangle$ toggles items in and out of $X$. Specifically, if $t \in X$, then it holds that $X\triangle \{ t \} = X - t$. Otherwise, for $t \notin X$, we have $X \triangle \{t\} = X + t$. \\

\noindent
{\bf Allocation.} An allocation $\alloc = (A_1, A_2, \ldots, A_n)$ is a partition of the $m$ items among the $n$ agents; specifically, for the item subsets $A_i \subseteq [m]$ we have $A_i \cap A_j = \emptyset$, for all $i \neq j$, and $\cup_{i=1}^n A_i = [m]$. Here, the subset $A_i$ is assigned to agent $i$ and is referred to as agent $i$'s bundle. Furthermore, a partial allocation  $(A'_1, A'_2, \ldots, A'_n)$ corresponds to a collection of pairwise-disjoint subsets ($A'_i \cap A'_j = \emptyset$ for all $i \neq j$) wherein some of the items may remain unassigned, $\cup_{i=1}^n A'_i \subsetneq [m]$. Complementarily and to emphasize that every item is assigned in an allocation, we will use the term complete allocation. The set of all complete allocations of the items, $[m]$, among the $n$ agents will be denoted by $\Pi_n([m])$. 
 
A {fractional} allocation $\calX = (X_1, \ldots, X_n)$ corresponds to a fractional assignment of the items; here,  $X_i = (X_{i,1}, X_{i,2}, \ldots, X_{i,m})$ with each $X_{i,t} \in [0,1]$ denoting the fraction of item $t$ assigned to agent $i$. In a fractional allocation $\calX = (X_1, \ldots, X_n)$ we have $\sum_{i=1}^n X_{i,t} = 1$, for each item $t \in [m]$, and each agent $i$ receives a value $v_i(X_i) \coloneqq \sum_{t=1}^m X_{i,t} \ v_i(t)$. Indeed, if $X_{i,t} \in \{0,1\}$ for each $i$ and $t$, then $\calX$ is an (integral) allocation. \\

\noindent
{\bf Fairness.} Next, we define the fairness notions addressed in this work. A central fairness criterion in fair division is that of envy-freeness. Specifically, an allocation $\calS = (S_1, \ldots, S_n)$ is said to be envy-free ($\EF$) if each agent $i \in [n]$ values its bundle at least as much as that of any other agent, $v_i(S_i) \geq v_i(S_j)$ for all $j \in [n]$. For indivisible items, an envy-free allocation is not guaranteed to exist, motivating the study of relaxations of envy-freeness in discrete fair division. A well-studied relaxation is envy-freeness up to one item ($\EFo$). Formally, an allocation $\calB = (B_1, \ldots, B_n)$ of mixed manna is said to be $\EFo$ if, for every pair of agents $i, j \in [n]$, either $i$ does not envy $j$, or there exists an item $t \in B_i \cup B_j$ such that $v_i(B_i -t) \geq v_i(B_j - t)$. That is, existing envy from agent $i$ towards agent $j$ can be eliminated by either removing a chore from $i$'s bundle or removing a good from $j$'s bundle. 

This work resolves the fairness and efficiency question for mixed manna, considering a close variant of the standard $\EFo$ criterion. In particular, we deem an allocation to be fair---introspectively envy-free up to one item ($\IEFo$)---if existing envy from agent $i$ towards agent $j$ can be resolved by removing a chore or adding a good to $i$'s own bundle. Hence, unlike $\EFo$, in $\IEFo$ the hypothetical inclusion or exclusion of an item $t$, aimed at eliminating $i$'s envy, only concerns $i$'s bundle, i.e., the notion is introspective.\footnote{Recall that $X\triangle \{ t \} = X - t$, if $t \in X$. Otherwise, for $t \notin X$, we have $X \triangle \{t\} = X + t$.}

\begin{definition}[$\IEFo$]
\label{definition:iefo}
In an instance $\calI = \angleb{[n], [m], \{v_i\}_{i=1}^n}$, an allocation $\alloc = (A_1, \ldots, A_n) \in \Pi_n([m])$ is said to be \emph{introspectively envy-free up to one item} ($\IEFo$) if for each agent $i \in [n]$ there exists a subset $S \subseteq [m]$ of size at most one ($|S| \leq 1$) such that $v_i(A_i \triangle S) \geq \max\limits_{1 \leq j \leq n} \  v_i(A_j)$. 
\end{definition}
Note that in fair division instances with only chores (i.e., with $v_i(t) < 0$ for all agents $i$ and items $t$) an allocation is $\IEFo$ iff it is $\EFo$. In the case of indivisible goods (i.e., when $v_i(t) \geq 0$ for all agents $i$ and items $t$), $\EFo$ implies $\IEFo$: Starting with an $\EFo$ allocation $(A_1, \ldots, A_n)$ of goods, to satisfy \Cref{definition:iefo} for an envious agent $i$, we can set $S=\{t^*\}$, where good $t^* \in \argmax_{t \in [m] \setminus A_i} \ v_i(t)$.  \\

\noindent
{\bf Efficiency and Social Welfare.} In a fair division instance $\calI = \angleb{[n], [m], \{v_i\}_{i=1}^n}$, a (possibly fractional) allocation $\calX = (X_1, \ldots, X_n)$ is said to {Pareto dominate} another allocation $\mathcal{Y} = (Y_1, \ldots, Y_n)$ if $v_i(X_i) \geq v_i(Y_i)$, for all $i \in [n]$, and this inequality is strict for at least one agent. 

\begin{definition}[$\PO$]
In an instance $\calI = \angleb{[n], [m], \{v_i\}_{i=1}^n}$, an allocation $\alloc \in \Pi_n([m])$ is said to be \emph{Pareto efficient} ($\PO$) if no other allocation $\calB \in \Pi_n([m])$ Pareto dominates it.    
\end{definition}

\begin{definition}[$\fPO$]
 In an instance $\calI = \angleb{[n], [m], \{v_i\}_{i=1}^n}$, an allocation $\alloc \in \Pi_n([m])$ is said to be \emph{fractionally Pareto efficient} ($\fPO$) if no fractional allocation $\calX \in [0,1]^{n \times m}$ Pareto dominates it.
\end{definition}

The social welfare of any allocation $\alloc=(A_1,\ldots, A_n)$ is the sum of agents' values under it, $\sum_{i=1}^n v_i(A_i)$. \\

We will use $\Delta_{n-1}$ to denote the $(n-1)$-dimensional simplex, $\Delta_{n-1} \coloneqq \{w \in \mathbb{R}^n \mid \sum_{i=1}^n w_i = 1 \text{ and } w_i \geq 0 \text{ for all } i \in [n]\}$. For any vector $w \in \mathbb{R}_{\geq 0}^n$, write $\supp(w)$ to denote the support (positive components) of vector $w$, i.e., $\supp(w) \coloneqq \{j \in [n] \mid w_j > 0 \}$.  

%% file: 4.PhaseI.tex
\section{$\IEFo$ and $\PO$ allocations for Mixed Manna}

The following theorem provides our main result, the guaranteed existence of a fair and efficient allocation of indivisible mixed manna.

\begin{theorem}\label{theorem:main_fair_efficient}
Every fair division instance $\calI = \angleb{[n],[m],\{v_i\}_{i=1}^n}$, with indivisible mixed manna and additive valuations, admits an allocation that is both introspectively envy-free up to one item ($\IEFo$) and Pareto efficient ($\PO$).
\end{theorem}

\paragraph{Proof Structure.} The proof of \Cref{theorem:main_fair_efficient} is developed in the following three sections. 

First, Section \ref{section:phaseI} establishes that it suffices to focus on {\it non-degenerate} instances. In particular, starting from the given fair division instance $\calI = \angleb{[n],[m],\{v_i\}_{i=1}^n}$, we construct (in Section \ref{subsection:nonDegInstanceDef}) a related instance $\barI$. Lemmas \ref{lemma:three-types} and \ref{lemma:non-degeneracy} (in \Cref{subsection:nonDegInstanceDef}) provide the non-degeneracy properties of the constructed instance $\barI$. Subsequently, Lemmas \ref{lemma:EFE-non-degenerate} and \ref{lemma:PO-non-degenerate} (in Section \ref{subsection:instace-reduction}) show that existence of a fair and efficient (in particular, a welfare-maximizing) allocation in $\barI$ implies the desired existence of an $\IEFo$ and $\PO$ allocation in the given instance $\calI$. Hence, Theorem \ref{theorem:main_fair_efficient} reduces to establishing fairness and efficiency in $\barI$ (see Theorem \ref{theorem:nondegenerate-fair-efficient}).  

Subsequently, Section \ref{section:setup-KKM} sets up the Knaster-Kuratowski-Mazurkiewicz (KKM) Theorem towards establishing the desired fairness and efficiency guarantee in $\barI$. Here, we prove (via the KKM Theorem) the existence of a weight vector $w^* \in \Delta_{n-1}$ with the property that, for each agent $i \in [n]$, we have a weighted welfare-maximizing allocation $\alloc^{(i)}$ that is fair (in a price sense) for $i$; see Theorem \ref{theorem:KKM-application}. 

Finally, in Section \ref{section:phaseII}, we start with the weight vector $w^*$ and the corresponding welfare-maximizing allocations $\alloc^{(1)}, \alloc^{(2)}, \ldots, \alloc^{(n)}$ that are price envy-free for the $n$ agents, respectively. We establish that the existence of such a collection of allocations implies the existence of a single allocation $\alloc^*$ that is $\IEFo$ and weighted welfare-maximizing in $\barI$; this implication is obtained via a novel Augmenting Tree Method (\Cref{theorem:leveled-allocation-is-fair-n-efficient}). Therefore, using the above-mentioned sufficiency of fairness and efficiency in $\barI$, we overall obtain our main result, \Cref{theorem:main_fair_efficient}.

\section{Sufficiency of Non-Degenerate Instances} \label{section:phaseI}

\paragraph{Item Types.} Without loss of generality, we will solely focus on instances $\calI = \angleb{[n],[m], \{v_i\}_{i=1}^n}$ wherein all the items in $[m]$ are of three types: goods, chores, or non-negatively valued. An item $t \in [m]$ is said to be a \emph{good} if $v_i(t) > 0$ for all agents $i \in [n]$. An item $t \in [m]$ is a \emph{chore} if $v_i(t) < 0$ for all agents $i \in [n]$. Further, an item is \emph{non-negatively valued} if $v_i(t) \geq 0$ for all agents $i \in [n]$ and we have $v_a(t) > 0$ for some agent $a \in [n]$ along with $v_b(t) = 0$ for some $b \in [n]$. This assumption that items are exclusively of these three types holds without loss of generality, given that we identify Pareto efficient (and fair) allocations. In particular, if in a given instance there exists an item $t$ that is positively valued by some agent and negatively valued by another agent, then we can set the value of $t$ to be zero for all agents that have a negative (or zero) value for it, and keep the values unchanged for agents that positively value $t$, i.e., transform $t$ into a non-negatively valued item. Note that, under the transformed valuation, in any Pareto efficient allocation, item $t$ must be assigned to an agent that positively valued it. This observation implies that any $\PO$ and $\IEFo$ allocation $\alloc$ under the transformed valuation is efficient and fair with respect to the given valuations as well. Furthermore, we can set aside any item $t'$ with the property that $v_i(t') \leq 0$, for all agents $i$, and $v_a(t') = 0$, for some agent $a$. Here, if $\alloc$ is a $\PO$ and $\IEFo$ allocation with $t'$ removed from consideration, then subsequently assigning $t'$ to any agent $a$ with $v_a(t') = 0$ maintains fairness and efficiency.  \\

Starting from the given instance $\calI = \angleb{[n],[m], \{v_i\}_{i=1}^n}$, we will first construct a \emph{non-degenerate} instance $\barI = \angleb{[n], [\mbar], \{\vbar_i\}_{i=1}^n}$, which will have certain desirable properties. We will show that the existence of a fair and efficient allocation in $\barI$ implies the existence of an $\IEFo$ and $\PO$ allocation in $\calI$. We begin by describing the construction of the instance $\barI$ (in Section \ref{subsection:nonDegInstanceDef}). We will then detail the (non-degeneracy) properties satisfied by $\barI$ (in Lemmas \ref{lemma:three-types} and \ref{lemma:non-degeneracy}), and show that, to prove \Cref{theorem:main_fair_efficient}, it suffices to establish a fairness and efficiency guarantee in $\barI$ (Lemmas \ref{lemma:EFE-non-degenerate} and \ref{lemma:PO-non-degenerate}). 

\subsection{Constructing Non-degenerate Instance $\barI$}\label{subsection:nonDegInstanceDef}
To construct $\barI = \angleb{[n], [\mbar], \{\vbar_i\}_{i=1}^n}$ from the given instance $\calI = \angleb{[n], [m], \{v_i\}_{i=1}^n}$, we first define the following parameters 

\[\lambda \coloneqq \min_{i\in [n], S,T \subseteq [m]} \big\{v_i(S) - v_i(T) \mid v_i(S) - v_i(T) > 0 \big\}, \]
\[\omega \coloneqq \min_{\alloc,\allocb \in \Pi_n([m])} \left\{ \sum_{i=1}^n v_i(A_i) - \sum_{i=1}^n v_i(B_i) \mid \sum_{i=1}^n v_i(A_i) - \sum_{i=1}^n v_i(B_i) > 0 \right\}. \]

Note that $\lambda$ represents the minimum possible envy in the given instance $\calI$ and $\omega$ is the minimum gap between social welfares. Both $\lambda$ and $\omega$ are positive. We further fix parameter $\varepsilon >0$ to satisfy  
\begin{equation}\label{definition:varepsilon}
\varepsilon < \frac{\lambda \omega}{16 m^2 n \ \max\limits_{i \in [n], t \in [m]} |v_i(t)|}
\end{equation}

All the terms in the right-hand-side of Equation (\ref{definition:varepsilon}) are positive and, hence, we can fix an $\varepsilon >0$ that satisfies this strict inequality. In addition, note that $\omega \leq 2 m \max\limits_{i \in [n], t \in [m]} |v_i(t)|$ -- this bound gives us 
\begin{align}
\varepsilon < \frac{\lambda}{8m} \label{ineq:ub-on-eps}
\end{align}

For the instance $\barI = \angleb{[n], [\mbar], \{\vbar_i\}_{i=1}^n}$, we set $\mbar \coloneqq m+1$, i.e., we add a new item to the existing set $[m]$. For the new item, $(m+1)$, in the instance $\barI$ we set the value $\vbar_i(m+1) = \lambda / 2$ for each agent $i \in [n]$. Furthermore, for agents $i \in [n]$ and items $t \in [m]$ with nonzero $v_i(t)$, we additively perturb this given value to obtain $\vbar_i(t)$ in the instance $\barI$. Specifically, for each $i$ and $t$, with $v_i(t) \neq 0$, draw independently and uniformly at random $\varepsilon_{i,t} \sim \mathrm{Uniform}([0,\varepsilon])$ and set $\vbar_i(t) = v_i(t) - \varepsilon_{i,t}$. We preserve the zero values: if for any agent $a$ and item $s$, we have $v_a(s) = 0$, then set $\vbar_a(s) = 0$.  As in the given instance $\calI$, the agents' valuations, $\vbar_i(\cdot)$, in $\barI$ are additive. This completes the construction of $\barI = \angleb{[n],[\mbar], \{\vbar_i\}_{i=1}^n}$. 

The proposition below notes that the sign of each value $v_i(t)$ is preserved by $\vbar_i(t)$ and, in the two instances, the maximum absolute values across the items are comparable. 
\begin{proposition}
\label{proposition:signs}
For any agent $i \in [n]$ and item $t$, if the given value $v_i(t) > 0$, then $\vbar_i(t)$ is positive as well. Similarly, if $v_a(s) = 0$, for an agent $a$ and item $s$, then $\vbar_a(s) = 0$. In addition, for each agent $j$ and item $t'$ with $v_j(t') <0$, it holds that $\vbar_j(t') <0$. Furthermore, \[\max_{\substack{i \in [n] \\ t\in [m]}} \  |v_i(t)| \ \geq \ \frac{1}{2}\max_{\substack{i \in [n] \\ t\in [\mbar]}} \  |\vbar_i(t)|.\] 
\end{proposition}
\begin{proof}
If $v_a(s) = 0$, then, by construction, we have $\vbar_a(s) = 0$. Next, note that if $v_j(t') <0$, then we set $\vbar_j(t') = v_j(t') - \varepsilon_{j,t'}$ with $\varepsilon_{j,t'} \geq 0$. Hence, $v_j(t') <0$ implies $\vbar_j(t') < 0$. 

In addition, for each value $v_i(t) > 0$, it holds that $v_i(t) = v_i(t) - v_i(\emptyset) \geq \lambda$; the last inequality follows from the definition of $\lambda$. Now, using the fact $\varepsilon_{i,t} \leq \varepsilon <  \frac{\lambda}{8m}$ (see Equation (\ref{ineq:ub-on-eps})), we obtain $\vbar_i(t) = v_i(t) - \varepsilon_{i,t}  > \lambda - \frac{\lambda}{8m} > 0$. Therefore, each positive value, $v_i(t)>0$, continues to be positive under the constructed valuation, $\vbar_i(t) >0$. 

We next relate the maximum absolute values across the items in the two instances $\calI$ and $\barI$. In the definition of $\lambda$, the minimization includes considering one of the sets to be a singleton and the other to be empty. Hence, for any nonzero value $v_i(t)$ it holds that $|v_i(t)| \geq \lambda$. Therefore, for any $i$ and $t$ with nonzero $v_i(t)$, we have 
\begin{align}
|\vbar_i(t)| = |v_i(t) - \varepsilon_{i,t}| \leq |v_i(t)| + \varepsilon \leq |v_i(t)| + \lambda \leq 2|v_i(t)| \label{ineq:two-factor}
\end{align}  
Also, for the auxiliary item $(m+1)$, the assigned value is sufficiently small: $\vbar_i(m+1) = \lambda/2$ for each $i \in [n]$. This observation and Equation (\ref{ineq:two-factor}) gives us the desired bound $\max_{{i \in [n], \ t\in [m]}} \  |v_i(t)| \ \geq \ \frac{1}{2}\max_{{i \in [n], \ t\in [\mbar]}} \  |\vbar_i(t)|$. 
\end{proof}

The following two lemmas provide relevant properties of the constructed instance $\barI$. Specifically, Lemma \ref{lemma:three-types} shows that all the items in $\barI$ can be categorized into three types: goods, chores, or non-negatively valued. Then, Lemma \ref{lemma:non-degeneracy} establishes that valuations $\vbar_i(\cdot)$---obtained by perturbing the given values---satisfy a useful non-degeneracy property. 

\begin{lemma}\label{lemma:three-types}
In the constructed instance $\barI = \angleb{[n], [\mbar], \{\vbar_i\}_{i=1}^n}$, the set of items $[\mbar]$ can be partitioned as $[\mbar] = G \cup C \cup G_0$, where $G$ is the subset of goods, $C$ is the subset of chores, and $G_0$ is the subset of non-negatively valued items; in particular, goods $G = \{t \in [\mbar] \mid \vbar_i(t) > 0 \text{ for all } i \in [n] \}$, chores $C = \{t  \in [\mbar] \mid \vbar_i(t) < 0 \text{ for all } i \in [n]\}$, and non-negatively valued items $G_0 = \{t \in [\mbar] \mid \vbar_i(t) \geq 0 \text{ for all } i \in [n], \text{ and } \vbar_a(t) > 0  \text{ along with } \vbar_b(t) = 0 \text{ for some } a, b \in [n] \}$. 
\end{lemma}
\begin{proof}
Recall the assumption, which holds without loss of generality, that the set of items $[m]$ in the given instance $\calI$ can be partitioned into: goods $G' = \{t \in [m] \mid v_i(t) > 0 \text{ for all } i \in [n] \}$, chores $C' = \{t \in [m] \mid v_i(t) < 0 \text{ for all } i \in [n] \}$, and non-negatively valued items $G_0'=\{t \in [m] \mid v_i(t) \geq 0 \text{ for all } i \in [n] \text{ and } v_a(t) > 0 \text{ along with } v_b(t) = 0 \text{ for some } a, b \in [n] \}$. 

We will show that the following containments hold between $G', C', G_0'$ and the subsets $G, C, G_0$ (as defined in the lemma): $G = G' \cup \{ (m+1) \}$ along with $C = C'$ and $G_0 = G'_0$. Since $[\mbar] = [m] \cup \{(m+1)\}$, these containments establish the lemma. 

The sign preservations stated in \Cref{proposition:signs} lead to the stated containments: For each item $t \in G'$, we have $v_i(t) > 0$ for all the agents $i \in [n]$. Therefore, $\vbar_i(t) > 0$ for all agents $i$, i.e., as desired, $t \in G$. Also, for each item $t' \in C'$, it holds that $v_j(t') < 0$, for all the agents $j \in [n]$. This implies $\vbar_j(t') < 0$ for each $j \in [n]$ and, hence, $t' \in C$. Similarly, using sign preservation, we obtain that each non-negatively valued item in $\calI$ continues to be non-negatively valued in $\barI$; $s \in G'_0$ implies $s \in G_0$.  

Finally, note that for item $(m+1)$, we have $\vbar_i(m+1) = \lambda/2 >0$, for each $i \in [n]$, hence, the stated containment, $(m+1) \in G$, holds. 

Therefore, the set of items $[\mbar]$ partitions into goods $G$, chores $C$, and non-negatively valued items $G_0$. The lemma stands proved. 
\end{proof}

\begin{lemma}\label{lemma:non-degeneracy}
In the instance $\barI = \angleb{[n],[\mbar], \{\vbar_i\}_{i=1}^n}$ and for any simple cycle $\mathcal{C} = (i_1, t_1, i_2, t_2, \ldots, i_k, t_k, i_1)$ in the complete bipartite graph $K_{n,\mbar}$,\footnote{Note that, since the cycles $\mathcal{C}$ are simple, any agent $i_\ell \in [n]$ and any item $t_\ell \in [\mbar]$ can appear at most once in $\mathcal{C}$.} with $\vbar_{i_\ell}(t_\ell), \vbar_{i_{\ell+1}} (t_{\ell}) \neq 0$, for each index $\ell \in [k]$, it holds that $\prod_{\ell = 1}^{k} \left( \frac{\vbar_{i_{\ell+1}}(t_\ell)}{\vbar_{i_\ell}(t_{\ell})} \right) \neq 1$; here, the indices are cyclically mapped, i.e., $i_{k+1} = i_{1}$.
\end{lemma}
\begin{proof}
Fix any simple cycle $\mathcal{C} = (i_1, t_1, i_2, t_2, \ldots, i_k, t_k, i_1)$ in $K_{n,\mbar}$ with $\vbar_{i_\ell}(t_\ell), \vbar_{i_{\ell+1}} (t_{\ell}) \neq 0$, for each $\ell \in [k]$. Since cycle $\mathcal{C}$ contains at least two items and item $(m+1)$ can appear at most once in the simple cycle, there exists an item $t_x$ in $\mathcal{C}$ such that $t_x \neq (m+1)$. For such an item $t_x$, the values $\vbar_{i_x}(t_x)$ and $\vbar_{i_{x+1}} (t_x)$ are nonzero and satisfy $\vbar_{i_x}(t_x) = v_{i_x}(t_x) - \varepsilon_{i_x, t_x}$ along with $\vbar_{i_{x+1}}(t_x) = v_{i_{x+1}}(t_x) - \varepsilon_{i_{x+1}, t_x}$. Here,  $\varepsilon_{i_x, t_x}$ and $\varepsilon_{i_{x+1}, t_x}$ are continuous random variables independently drawn from the uniform distribution over $[0, \varepsilon]$. Hence, the product $\prod_{\ell = 1}^{k} \left( \frac{\vbar_{i_{\ell+1}}(t_\ell)}{\vbar_{i_\ell}(t_{\ell})} \right)$ is obtained by multiplying independent\footnote{Recall that $\varepsilon_{i,t}$s are drawn independently.} continuous random variables. Therefore, the product itself is a continuous random variable. That is, $\prod_{\ell = 1}^{k} \left( \frac{\vbar_{i_{\ell+1}}(t_\ell)}{\vbar_{i_\ell}(t_{\ell})} \right) \neq 1$ with probability one. 

Furthermore, using the fact that $K_{n,\mbar}$ has finitely many---in particular, $O(n^n (\mbar)^{\mbar})$---cycles, we get that, with probability one, for all cycles $\mathcal{C}$ (whose edges have nonzero values associated with them) the product $\prod_{\ell = 1}^{k} \left( \frac{\vbar_{i_{\ell+1}}(t_\ell)}{\vbar_{i_\ell}(t_{\ell})} \right) \neq 1$. This completes the proof of the lemma.   
\end{proof}

\subsection{Reduction to Non-degenerate Instance $\barI$}
\label{subsection:instace-reduction}

We will show that, to prove \Cref{theorem:main_fair_efficient} it suffices to focus on the non-degenerate instance $\barI$. Towards this, we first set parameter $\eta > 0$ as 
\begin{equation}\label{definition:value-of-eta}
\eta \coloneqq \frac{\lambda}{2mn \max\limits_{i\in [n], \ t \in [\mbar]}|\vbar_i(t)|}.
\end{equation}

The definition of $\eta$ leads to the following relation between this parameter and $\varepsilon$:
\begin{align}
\frac{\eta \omega}{4m} & = \frac{\omega}{4m} \left( \frac{\lambda}{2mn \max\limits_{i\in [n], \ t \in [\mbar]}|\vbar_i(t)|} \right) \nonumber \\
& \geq \frac{\omega}{4m} \left( \frac{\lambda}{4 mn \max\limits_{i\in [n], \ t \in [m]}|v_i(t)|} \right) \tag{via \Cref{proposition:signs}} \\
& = \frac{\lambda \omega}{16 m^2 n \max\limits_{i\in [n], \ t \in [m]}|v_i(t)|} \nonumber \\ 
& > \varepsilon \label{ineq:to-contradict}
\end{align}
Here, the last (strict) inequality follows from the definition of $\varepsilon$ (\Cref{definition:varepsilon}).

For instance $\barI$ and constant $\eta >0$, we define the following primal and dual linear programs parameterized by weight vectors $w \in \Delta_{n-1}$. In particular, the linear program $\textbf{\textrm{LP}}(w,\eta)$ maximizes the weighted social welfare across all fractional allocations. 

\[
\begin{array}{rlll}
\textbf{\textrm{LP}}(w,\eta): & & \qquad \textbf{\textrm{Dual-LP}}(w,\eta): \\  \ \ 
\max & \displaystyle \sum_{i \in [n]}\sum_{t \in [\mbar]} \bigl(w_i + \eta \bigr)\vbar_{i}(t) x_{i,t} & \ \ \ \qquad \min \displaystyle \sum_{t \in [\mbar]} p_t \\  
\text{s.t.} & \displaystyle \sum_{i \in [n]} x_{i, t} = 1 \ \text{ for all } t \in [\mbar] & \ \ \ \qquad \text{s.t.} \ \ \ p_t \geq \bigl(w_i + \eta \bigr)\vbar_{i}(t) \text{ for all } i \in [n] \text{ and } t \in [\mbar]. \\ 
& x_{i,t} \ge 0 \ \ \text{for all } i \in [n] \text{ and } t \in [\mbar]. & &
\end{array}
\]

Note that, for any weight vector $w \in \Delta_{n-1}$, the linear program $\textbf{\textrm{LP}}(w,\eta)$ always admits an integral optimal solution. Such an optimum corresponds to an integral allocation $\alloc =(A_1, \ldots, A_n) \in \Pi_n([\mbar])$ that can be obtained by assigning each item $t$ to an agent $i \in \argmax_{j \in [n]} \ (w_j + \eta) \vbar_{j}(t)$. We will write $\alloc \in \OPT(\textbf{\textrm{LP}}(w,\eta))$ to denote the weighted welfare optimality of $\alloc$. In fact, there can be multiple (integral) allocations in the set of optimal solutions $\OPT(\textbf{LP}(w,\eta))$. By contrast, for each $w \in \Delta_{n-1}$, the dual linear program $\textbf{\textrm{Dual-LP}}(w,\eta)$ has a unique optimal solution given by $p_t = \max_{i \in [n]} \  \bigl(w_i + \eta \bigr)\vbar_{i}(t)$; we will denote this fact by writing $p = \OPT(\textbf{\textrm{Dual-LP}}(w,\eta))$ and refer to $p \in \mathbb{R}^{\mbar}$ as the vector of optimal {prices}. 

The proposition below shows that the optimal prices are always nonzero. 
\begin{proposition}\label{proposition:price-non-negativity}
For any weight vector $w \in \Delta_{n-1}$, if $p = \OPT(\textbf{\textrm{Dual-LP}}(w,\eta))$, then $p_t > 0$ for each item $t \in G \cup G_0$ and $p_t < 0$ for each $t \in C$.
\end{proposition}
\begin{proof}
The optimal prices $p = \OPT(\textbf{\textrm{Dual-LP}}(w,\eta))$ satisfy $p_t = \max_{i \in [n]} \  \bigl(w_i + \eta \bigr)\vbar_{i}(t)$ for each item $t \in [\mbar]$. Since $\eta >0$ and $w \in \Delta_{n-1}$, it holds that the weight $\bigl(w_i + \eta \bigr) > 0$ for each $i \in [n]$. 

Now, note that, for each item $t \in G \cup G_0$ (by definition of the item subsets $G$ and $G_0$), we have $\max_{i \in [n]} \  \vbar_{i}(t) > 0$. Hence, for each $t \in G \cup G_0$, the optimal price $p_t = \max_{i \in [n]} \  \bigl(w_i + \eta \bigr)\vbar_{i}(t) > 0$. On the other hand, for each item $s \in C$, the value $\vbar_i(s) < 0$ for every agent $i \in [n]$.  Hence, $p_s = \max_{i \in [n]}  \bigl(w_i + \eta \bigr)\vbar_{i}(s) < 0$ for items $s \in C$. This completes the proof of the proposition. 
\end{proof}

The following two lemmas establish that, if an allocation $\alloc$ is fair and welfare-maximizing in the non-degenerate instance $\barI$, then $\alloc$ (with the auxiliary item $(m+1)$ removed) is fair and $\PO$ in the given instance $\calI$. That is, these two lemmas ensure that the existence of a fair and welfare-maximizing allocation in $\barI$ implies our main result (Theorem \ref{theorem:main_fair_efficient}) for $\calI$. 

\begin{restatable}{lemma}{nonDegEFE}\label{lemma:EFE-non-degenerate}
If an allocation $\alloc = (A_1, \ldots, A_n)$ is $\IEFo$ in the non-degenerate instance $\barI = \angleb{[n],[\mbar], \{\vbar_i\}_{i=1}^n}$, then the allocation $(A_1 \setminus \{(m+1)\}, A_2 \setminus \{(m+1)\}, \ldots, A_n \setminus \{(m+1)\})$ is $\IEFo$ in the given instance $\calI = \angleb{[n],[m], \{v_i\}_{i=1}^n}$.
\end{restatable}

\begin{restatable}{lemma}{nonDegPO}\label{lemma:PO-non-degenerate}
If an allocation $\alloc = (A_1, \ldots, A_n) \in \OPT(\textbf{\textrm{LP}}(w,\eta))$, for any $w \in \Delta_{n-1}$, then the allocation $(A_1 \setminus \{(m+1)\}, A_2 \setminus \{(m+1)\}, \ldots, A_n \setminus \{(m+1)\})$ is $\PO$ in the instance $\calI = \angleb{[n],[m], \{v_i\}_{i=1}^n}$.
\end{restatable}

Upon proving these lemmas, the subsequent goal of the paper will be to establish the existence of a $w^* \in \Delta_{n-1}$ and an allocation $\alloc^* \in \OPT(\textbf{LP}(w^*,\eta))$ that is $\IEFo$ in the non-degenerate instance $\barI$. As mentioned previously, Lemmas \ref{lemma:EFE-non-degenerate} and \ref{lemma:PO-non-degenerate} will then imply that $(A^*_1\setminus \{(m+1)\}, A^*_2\setminus \{(m+1)\}, \ldots, A^*_n \setminus \{(m+1)\})$ is $\IEFo$ and $\PO$ for the given instance $\calI$, thereby proving \Cref{theorem:main_fair_efficient}.

\begin{proof}[Proof of \Cref{lemma:EFE-non-degenerate}]
Allocation $\alloc = (A_1, \ldots, A_n)$ is $\IEFo$ in instance $\barI$ and, hence, for each agent $i\in [n]$ there exists a subset $S_i \subseteq [\mbar]$, with $|S_i| \leq 1$ and the property that $\vbar_i(A_i \triangle S_i) \geq \max_{j \in [n]} \ \vbar_i(A_j)$. We will show that the inequality $\vbar_i(A_i \triangle S_i) \geq \max_{j \in [n]} \vbar_i(A_j)$ implies the following bound (under the given valuation $v_i$s)
\begin{equation}\label{equation:efe-wlog-proof-ineq-1}
v_i \Big((A_i \setminus \{m+1\}) \triangle (S_i \setminus \{m+1\}) \Big) \allowbreak \geq \max_{j \in [n]} v_i \big(A_j \setminus \{m+1\} \big).
\end{equation}
Since $|S_i \setminus \{m+1\}| \leq 1$, this would imply that the allocation $(A_1 \setminus \{m+1\}, A_2 \setminus \{m+1\}, \ldots, A_n \setminus \{m+1\})$ is $\IEFo$ for the instance $\calI = \angleb{[n],[m], \{v_i\}_{i=1}^n}$, thereby establishing the lemma. To prove Equation (\ref{equation:efe-wlog-proof-ineq-1}), we start with inequality $\vbar_i(A_i \triangle S_i) \geq \max_{j \in [n]} \vbar_i(A_j)$, and, for each $j \in [n]$, obtain:  
\begin{align*}
0 & \leq \vbar_i(A_i \triangle S_i) - \vbar_i(A_j) \\
& \leq \vbar_i\big((A_i \triangle S_i ) \setminus \{m+1\}\big) - \vbar_i(A_j \setminus \{m+1\}) + \vbar_i(m+1) \tag{given that $\vbar_i(m+1) > 0$}\\
& \leq v_i \big((A_i \triangle S_i ) \setminus \{m+1\} \big) - v_i(A_j \setminus \{m+1\}) + \varepsilon |A_j \setminus \{m+1\}| + \vbar_i(m+1) \tag{since $v_i(t) - \varepsilon \leq \vbar_i(t) \leq v_i(t)$, for each $t$} \\
& \leq v_i\big((A_i \triangle S_i ) \setminus \{m+1\}\big) - v_i(A_j \setminus \{m+1\}) + \varepsilon m + \frac{\lambda}{2} \tag{$\vbar_i(m+1) = \lambda/2$} \\
& < v_i\big((A_i \triangle S_i ) \setminus \{m+1\}\big) - v_i(A_j \setminus \{m+1\}) + \frac{\lambda}{8m} \ m + \frac{\lambda}{2} \tag{$\varepsilon < \frac{\lambda}{8m}$; \Cref{ineq:ub-on-eps}}\\
& \leq v_i \Big((A_i \setminus \{m+1\}) \triangle (S_i \setminus \{m+1\}) \Big) - v_i(A_j \setminus \{m+1\}) + \lambda
\end{align*}
Hence, $v_i(A_j \setminus \{m+1\}) - v_i\Big((A_i \setminus \{m+1\}) \triangle (S_i \setminus \{m+1\}) \Big) < \lambda$. This strict inequality implies that, in fact, $v_i(A_j \setminus \{m+1\}) - v_i\Big((A_i \setminus \{m+1\}) \triangle (S_i \setminus \{m+1\}) \Big) \leq 0$; otherwise, the bound  $0 < v_i(A_j \setminus \{m+1\}) - v_i\Big((A_i \setminus \{m+1\}) \triangle (S_i \setminus \{m+1\}) \Big) < \lambda$
would contradict the definition of $\lambda$. Using $v_i(A_j \setminus \{m+1\}) - v_i\Big((A_i \setminus \{m+1\}) \triangle (S_i \setminus \{m+1\}) \Big) \leq 0$ for all agents $j \in [n]$, we obtain Equation (\ref{equation:efe-wlog-proof-ineq-1}). This completes the proof of the lemma. 
\end{proof}

Next, we establish \Cref{lemma:PO-non-degenerate}.
\begin{proof}[Proof of \Cref{lemma:PO-non-degenerate}]
We begin by defining an intermediate instance $\calI' = \angleb{[n], [m+1], \{v'_i\}_{i=1}^n}$ wherein $v'_i(t) = v_i(t)$, for all $i \in [n]$ and $t \in [m]$, and $v'_i(m+1) = \lambda/2$ for each $i \in [n]$. That is, we construct instance $\calI'$ by adding item $(m+1)$ to $\calI$ and keeping the valuations of the given $m$ items unchanged. Instance $\calI'$ is defined for the purposes of analysis: We will first show that $\alloc$ is $\PO$ in $\calI'$. Then, we will conclude the proof by showing that $(A_1 \setminus \{m+1\}, A_2 \setminus \{m+1\}, \ldots, A_n \setminus \{m+1\})$ is $\PO$ in $\calI$.

Assume, towards a contradiction, that there exists an allocation $\allocb = (B_1,\ldots, B_n)$ that Pareto dominates $\alloc$ in $\calI'$. That is, $v'_i(B_i) \geq v'_i(A_i)$, for all $i \in [n]$, and $v'_a(B_a) > v'_a(A_a)$ for some $a \in [n]$. These inequalities imply that
\begin{align*}
0 & < \sum_{i=1}^n v'_i(B_i) - \sum_{i=1}^n v'_i(A_i)\\
& =  \frac{\lambda}{2} + \sum_{i=1}^n v'_i(B_i \setminus \{m+1\}) -  \frac{\lambda}{2} - \sum_{i=1}^n v'_i(A_i\setminus \{m+1\}) \tag{$v'_j(m+1) = \lambda/2$ for all $j$}\\
& = \sum_{i=1}^n v_i(B_i \setminus \{m+1\}) - \sum_{i=1}^n v_i(A_i\setminus \{m+1\}) \label{equation:nonDegPOProof1}\numberthis
\end{align*}
That is, $\sum_{i=1}^n v_i(B_i \setminus \{m+1\}) - \sum_{i=1}^n v_i(A_i\setminus \{m+1\}) > 0$. Using this strict inequality and the definition of $\omega$, we obtain 
\[\omega \leq \sum_{i=1}^n v_i(B_i \setminus \{m+1\}) - \sum_{i=1}^n v_i(A_i\setminus \{m+1\}) = \sum_{i=1}^n v'_i(B_i) - \sum_{i=1}^n v'_i(A_i) \label{equation:nonDegPOProof2} \numberthis \]
For the last equality, we use the fact that $v'_j(m+1) = \lambda/2$ for all $j$.
 
 Next, we use the welfare optimality of $\alloc$ in $\barI$ (and under the weight vector $w$) to arrive at a contradiction. In particular, $\alloc \in \OPT(\textbf{\textrm{LP}}(w,\eta))$ implies that $\sum_{i=1}^n (w_i + \eta) \vbar_i(B_i) \leq \sum_{i=1}^n (w_i + \eta) \vbar_i(A_i)$. For each subset $S \subseteq [\mbar]$ and each agent $i$, it holds that $v'_i(S) - \varepsilon |S| \leq \vbar_i(S) \leq v'_i(S)$. Therefore, we have $\sum_{i=1}^n (w_i + \eta) (v'_i(B_i) - \varepsilon |B_i|) \leq \sum_{i=1}^n (w_i + \eta) v'_i(A_i)$. This inequality reduces to 
 \begin{align}
\sum_{i=1}^n (w_i + \eta) (v'_i(B_i) - v'_i(A_i)) & \leq \varepsilon \sum_{i=1}^n (w_i + \eta)  |B_i| \nonumber \\
& = \varepsilon \sum_{i=1}^n w_i |B_i|  \ + \  \varepsilon \eta \sum_{i=1}^n |B_i| \nonumber  \\ 
& \leq \varepsilon (m+1) \sum_{i=1}^n w_i    \ + \  \varepsilon \eta \sum_{i=1}^n |B_i| \tag{since $w_i \geq 0$ and $|B_i| \leq m +1 $} \\
 & = \varepsilon (m+1)  + \varepsilon \eta (m+1) \tag{since $\sum_i w_i = 1$ and $\sum_i |B_i| = m + 1$} \\ 
 & = \varepsilon (m+1) (1 + \eta) \nonumber  \\
& \leq 2 \varepsilon (m+1) \tag{since $1 + \eta < 2$} \\
 & \leq 4 \varepsilon m \label{ineq:extra-two}
\end{align}
We can lower bound the left-hand-side of inequality (\ref{ineq:extra-two}) as follows 
\begin{align*}
\sum_{i=1}^n (w_i + \eta) (v'_i(B_i) - v'_i(A_i)) & \geq \eta \sum_{i=1}^n  (v'_i(B_i) - v'_i(A_i)) \tag{since $v'_i(B_i) \geq v'_i(A_i)$ and $w_i \geq 0$, for each $i$}\\
& \geq \eta \omega \tag{via (\ref{equation:nonDegPOProof2})}
\end{align*}
The last inequality and Equation (\ref{ineq:extra-two}) imply $\eta \omega \leq 4 \varepsilon m$, i.e., $\frac{\eta \omega}{4m} \leq \varepsilon$. This bound, however, contradicts the strict inequality (\ref{ineq:to-contradict}). Therefore, by way of contradiction, we obtain that no allocation $\calB$ Pareto dominates $\alloc \in \OPT(\textbf{\textrm{LP}}(w,\eta))$ in the instance $\calI'$, i.e., $\alloc$ is $\PO$ in $\calI'$. 

Finally, we will use the Pareto optimality of $\alloc$ in $\calI'$ to show that the allocation $(A_1 \setminus \{m+1\},\ldots, A_n \setminus \{m+1\})$ is $\PO$ in the given instance $\calI = \angleb{[n], [m], \{v_i\}_{i=1}^n}$. Assume, towards a contradiction, that allocation $(D_1, \ldots, D_n)$ Pareto dominates $(A_1 \setminus \{m+1\},\ldots, A_n \setminus \{m+1\})$ in $\calI$. Write $a \in [n]$ to denote the agent that receives the auxiliary item $(m+1)$ in allocation $\alloc$, i.e., $m+1 \in A_a$. Note that the allocation $(A_1 \setminus \{m+1\},\ldots, A_a \setminus \{m+1\}, \ldots, A_n \setminus \{m+1\})$ is the same as $(A_1,\ldots, A_a \setminus \{m+1\}, \ldots, A_n)$. The fact that $(D_1, \ldots, D_n)$ Pareto dominates $(A_1,\ldots, A_a \setminus \{m+1\}, \ldots, A_n)$ implies that $(D_1, \ldots, D_a \cup \{ m+1 \}, \ldots, D_n)$ Pareto dominates $\alloc$ in instance $\calI'$. Since this contradicts the already established Pareto optimality of $\alloc$ in $\calI'$, we obtain that $(A_1 \setminus \{m+1\}, \ldots, A_n \setminus \{m+1\})$ is $\PO$ in the instance $\calI$. The lemma stands proved.  
\end{proof}

In light of Lemmas \ref{lemma:EFE-non-degenerate} and \ref{lemma:PO-non-degenerate}, our goal (in Sections \ref{section:setup-KKM} and \ref{section:phaseII}) will be to prove the following theorem.

\begin{restatable}{theorem}{nonDegFairEfficient}\label{theorem:nondegenerate-fair-efficient}
There exists a vector $w^* \in \Delta_{n-1}$ and an associated allocation $\alloc^* \in \OPT(\textbf{\textrm{LP}}(w^*,\eta))$ with the property that $\alloc^*$ is $\IEFo$ in the constructed instance $\barI$.
\end{restatable}

As stated previously, our main result, \Cref{theorem:main_fair_efficient}, follows directly from \Cref{theorem:nondegenerate-fair-efficient}. For completeness, we provide a proof here. 

\begin{proof}[Proof of \Cref{theorem:main_fair_efficient}]
Let $\alloc^* = (A^*_1, \ldots, A^*_n)$ be the allocation whose existence is guaranteed by \Cref{theorem:nondegenerate-fair-efficient}. By Lemmas \ref{lemma:EFE-non-degenerate} and \ref{lemma:PO-non-degenerate}, the allocation $(A^*_1 \setminus\{m+1\}, A^*_2\setminus\{m+1\}, \ldots, A^*_n\setminus\{m+1\})$ is $\IEFo$ and $\PO$ in the given instance $\calI$. Therefore, we obtain the stated fairness and efficiency guarantee for indivisible mixed manna. 
\end{proof}

\section{Applying the KKM Theorem}
\label{section:setup-KKM}
To prove \Cref{theorem:nondegenerate-fair-efficient}, we will invoke the Knaster-Kuratowski-Mazurkiewicz (KKM) Theorem. The theorem asserts that any $n$ closed subsets of the $(n-1)$-dimensional simplex $\Delta_{n-1}$ that satisfy a covering condition necessarily have a nonempty intersection. 

\begin{theorem}{\cite{knaster1929beweis}}\label{theorem:KKM}
Let $C_1, C_2, \ldots, C_n \subseteq \Delta_{n-1}$ be closed subsets of $\Delta_{n-1}$ with the property that for each vector $w \in \Delta_{n-1}$ there is an index $i \in \supp(w)$ satisfying $w \in C_i$. 
Then, there exists a vector $w^* \in \cap_{i=1}^n C_i$.
\end{theorem}
To apply \Cref{theorem:KKM}, we will appropriately define the subsets $C_i \subseteq \Delta_{n-1}$ and show that the vector $w^*$ guaranteed in the KKM Theorem satisfies \Cref{theorem:nondegenerate-fair-efficient}. 

We define $C_i$s considering the previously-mentioned primal and dual linear programs, $\LP(w,\eta)$ and $\Dual(w,\eta)$, which are parameterized by weight vectors $w \in \Delta_{n-1}$. Recall that 
$p = \OPT(\Dual(w, \eta))$ denotes the optimal price vector and $\OPT(\LP(w, \eta))$ contains the weighted welfare-maximizing (integral) allocations. For any subset of items $S \subseteq [\mbar]$, we will write $p(S)$ to denote the total price of the items in $S$, i.e., $p(S) = \sum_{t \in S} p_t$. 

A key construct here is \emph{price envy-freeness}. Specifically, we say that an agent $i$ is price envy-free in an allocation $\calB=(B_1, \ldots, B_n)$ and under prices $p' \in \mathbb{R}^{\mbar}$ if $i$'s bundle has a price at least as much as that of any other agent: $p'(B_i) \geq p'(B_j)$ for all $j \in [n]$. For each agent $i \in [n]$, we define $C_i$ to be the set of weights $w \in \Delta_{n-1}$ for which there is an optimal allocation $\alloc \in \OPT(\textbf{\textrm{LP}}(w,\eta))$ wherein $i$ is price envy-free, under the optimal prices $p = \OPT(\Dual(w, \eta))$. Formally, 

\begin{equation}\label{definition:KKM-sets-a}
  C_i \coloneqq \left\{w \in \Delta_{n-1} \mid \text{ there exists } \alloc \in \OPT(\textbf{\textrm{LP}}(w,\eta)) \text{ such that } i \in \argmax_{j \in [n]}{p(A_j)} \text{ for the optimal prices } p \right\}
\end{equation}

The following two lemmas show that the sets $C_1, \ldots, C_n$ (defined in Equation (\ref{definition:KKM-sets-a})) satisfy the conditions of the KKM Theorem.

\begin{lemma}[Closedness]\label{lemma:KKM-closedness}
For each $i \in [n]$, the set $C_i \subseteq \Delta_{n-1}$ is closed.
\end{lemma}
\begin{proof}
Fix any agent $i \in [n]$. Consider any convergent sequence $\{w^k \}_{k \in \mathbb{N}}$ in $C_i$ and $\overline{w} \in \Delta_{n-1}$ be its limit. We will show that $\overline{w} \in C_i$ and, hence, obtain that the $C_i$ is closed. 

Since the total number of allocations is finite, there must be an infinite subsequence $\{ \widetilde{w}^k \}_{k \in \mathbb{N}}$ and a particular allocation $\widetilde{\alloc} = (\widetilde{A}_1,\ldots, \widetilde{A}_n)$ such that,  for each $k \in \mathbb{N}$, allocation $\widetilde{\alloc} \in \OPT(\textbf{\textrm{LP}}(\widetilde{w}^k,\eta))$ and $i$ is price envy-free here with optimal prices $\widetilde{p}^k = \OPT(\textbf{\textrm{Dual-LP}}(\widetilde{w}^k,\eta))$.

We first note that $\widetilde{\alloc}$ is a welfare-maximizing allocation even for the limit point $\overline{w}$. 
\begin{claim}
\label{claim:opt-convergence}
$\widetilde{\alloc} \in \OPT(\LP(\overline{w},\eta))$.
\end{claim}
\begin{proof}
Write $\SW(\calB, w)$ to denote the weighted social welfare of allocation $\calB=(B_1, \ldots, B_n)$,\footnote{Recall that $\eta$ is a fixed parameter.} i.e., $\SW(\calB, w) \coloneqq \sum_{i=1}^n (w_i + \eta) \vbar_i(B_i)$. In addition, write $\SW^*(w)$ to denote the optimal value of the linear program $\LP(w,\eta)$. Note that the optimal value satisfies $\SW^*(w) = \max_{\calB \in \Pi_n([\mbar])} \ \SW(\calB, w)$. 

Furthermore, for any fixed allocation $\calB \in \Pi_n([m])$, the function $\SW(\calB, w)$ is continuous in $w$. Hence, $\SW^*(w)$---a maximum of continuous functions---is also continuous. Now, using the sequential criterion of continuity, we obtain that $\SW(\widetilde{\alloc}, \widetilde{w}^k) \to \SW(\widetilde{\alloc}, \overline{w})$ as $k \to \infty$. In addition, $\SW^*(\widetilde{w}^k) \to \SW^*(\overline{w})$ as $k \to \infty$. The two sequences here, $\{\SW(\widetilde{\alloc}, \widetilde{w}^k)\}_k$ and $\{\SW^*(\widetilde{w}^k)\}_k$, are the same, since $\widetilde{\alloc} \in \OPT(\textbf{\textrm{LP}}(\widetilde{w}^k,\eta))$ for each $k$. Therefore, $\SW(\widetilde{\alloc}, \overline{w})  = \SW^*(\overline{w})$, i.e., allocation $\widetilde{\alloc}$ is welfare-maximizing at $\overline{w}$ as well: $\widetilde{\alloc} \in \OPT(\LP(\overline{w},\eta))$. This completes the proof of the claim. 
\end{proof}

Write $\overline{p}= \OPT(\Dual(\overline{w},\eta))$ to denote the optimal price vector at the limit point $\overline{w}$. The claim below shows that the allocation $\widetilde{\alloc} = (\widetilde{A}_1,\ldots, \widetilde{A}_n)$ ensures price envy-freeness for the agent $i$ under $\overline{p}$. 
\begin{claim}
\label{claim:price-convergence}
Agent $i \in  \argmax_{j \in [n]} \ \overline{p} (\widetilde{A}_j)$. 
\end{claim}
\begin{proof}
Recall that for any weight vector $w'$, the optimal prices $p' = \OPT(\Dual(w', \eta))$ satisfy $p'_t = \max_{i \in [n]} \ (w'_i + \eta) \vbar_i(t)$ for each item $t \in [\mbar]$. That is, for each item $t$, the optimal price is a continuous function of the weight. Hence, using the sequential criterion of continuity again, we obtain that the optimal price vectors $\widetilde{p}^k = \OPT(\textbf{\textrm{Dual-LP}}(\widetilde{w}^k,\eta))$ converge (component-wise) to $\overline{p}$. Furthermore, for each bundle $\widetilde{A}_j$ in the allocation $\widetilde{\alloc} = (\widetilde{A}_1,\ldots, \widetilde{A}_n)$, we have $\widetilde{p}^k \left(\widetilde{A}_j \right) \to \overline{p} \left(\widetilde{A}_j \right)$ as $k \to \infty$.  

Next, note that, for each $k \in \mathbb{N}$, since agent $i$ is price envy-free in $\widetilde{\alloc}$ (under $\widetilde{p}^k$), we have $\widetilde{p}^k \left(\widetilde{A}_i \right) \geq \widetilde{p}^k \left(\widetilde{A}_j \right)$ for every $j \in [n]$. Therefore, the order-preserving property of limits gives us $\overline{p} \left(\widetilde{A}_i \right) \geq \overline{p} \left(\widetilde{A}_j \right)$ for each $j$. The claim stands proved. 
\end{proof}
Claims \ref{claim:opt-convergence} and \ref{claim:price-convergence} imply that, for the limit point $\overline{w}$, the allocation $\widetilde{\alloc} \in \OPT(\LP(\overline{w},\eta))$ yields price envy-freeness for agent $i$ under $\overline{p}$. Therefore, the limit point $\overline{w} \in C_i$ (see Equation (\ref{definition:KKM-sets-a})) and, as desired, we obtain that $C_i$ is closed. 
\end{proof}

\begin{lemma}[KKM Covering]\label{lemma:KKM-covering}
For each $w \in \Delta_{n-1}$ there exists an $i \in \supp(w)$ such that $w \in C_i$.
\end{lemma}
\begin{proof}
Fix any weight vector $w \in \Delta_{n-1}$. We consider two cases, based on whether $\supp(w) = [n]$ or $\supp(w) \subsetneq [n]$. \\

\noindent
{\it Case {\rm I}: $\supp(w) = [n]$.} Consider any allocation $\alloc \in \OPT(\textbf{\textrm{LP}}(w,\eta))$ and write $p = \OPT(\textbf{\textrm{Dual-LP}}(w,\eta))$. Also, write $i' \in \argmax_{j \in [n]} p(A_j)$. In the current case $\supp(w) = [n]$ and, hence, $i' \in \supp(w)$. Additionally, we have $w \in C_{i'}$, since $i' \in \argmax_{j \in [n]} p(A_j)$  (see Equation (\ref{definition:KKM-sets-a})). Hence, for vectors $w$ with full support, the KKM covering condition holds. \\

\noindent
{\it Case {\rm II}: $\supp(w) \subsetneq [n]$.} In this case $[n] \setminus \supp(w) \neq \emptyset$. Write $p = \OPT(\textbf{\textrm{Dual-LP}}(w,\eta))$ and consider any allocation $\alloc \in \OPT(\textbf{\textrm{LP}}(w,\eta))$. We will show that there exists an $i' \in \supp(w)$ with the property that $i' \in \argmax_{j \in [n]} \ p(A_j)$, i.e., $i'$ (with $w_{i'} >0$) is price envy-free. For such an agent $i'$, we have $w \in C_{i'}$, and, hence, we will obtain KKM covering in the current case as well. 

Write $H \coloneqq \{i \in [n] \mid w_i \geq 1/n\}$. This set $H \neq \emptyset$, since $\sum_{i=1}^n w_i = 1$. Also, the set of agents $[n]$ partitions into three disjoint subsets: $[n] \setminus \supp(w)$, $H$, and $\supp(w) \setminus H$. 

Furthermore, note that in the welfare-maximizing allocation $\alloc \in \OPT(\textbf{\textrm{LP}}(w,\eta))$, each item $t \in [\mbar]$ is assigned to an agent in $\argmax_{i \in [n]} \ (w_i + \eta) \vbar_i(t)$. In particular, the auxiliary item $(m+1)$ must be assigned (in $\alloc$) to an agent in $\argmax_{i \in [n]} w_i$; recall that, by construction, item $(m+1)$ has the same value $\vbar_i(m+1) = \lambda/2 > 0$ for all the agents $i$. Write $h \in \argmax_{i \in [n]} w_i$ to denote the agent that receives $(m+1)$ in $\alloc$. That is, $h$ has a highest-valued component in $w$ and item $(m+1) \in A_h$. Also, we have $h \in H \subseteq \supp(w)$. 

To establish the existence of price envy-free agent $i' \in \supp(w)$, we will show that $\max\limits_{j \in [n] \setminus \supp(w)} p(A_j) \leq p(A_h)$. Since $h \in \supp(w)$, the last inequality extends to  
\begin{align}
\max\limits_{j \in [n] \setminus \supp(w)} p(A_j) \leq p(A_h) \leq \max\limits_{i \in \supp(w)} p(A_i) \label{ineq:sandwich}
\end{align}
Note that inequality (\ref{ineq:sandwich}) implies that there exists an agent $i' \in \supp(w)$ with a maximum-price bundle, $i' \in \argmax_{i \in [n]} p(A_i)$, i.e., as desired, there exists a price envy-free agent $i' \in \supp(w)$. 

Hence, to complete the analysis of the current case, it remains to prove that $\max\limits_{j \in [n] \setminus \supp(w)} p(A_j) \leq p(A_h)$. Towards this, we will establish the following two inequalities
\begin{align}  
\max\limits_{j \in [n] \setminus \supp(w)} p(A_j) & \leq \frac{\lambda}{2n} \label{ineq:zero-upper} \\ 
\frac{\lambda}{2n} & \leq p(A_h) \label{ineq:h-lower}
\end{align}
Indeed, inequalities (\ref{ineq:zero-upper}) and (\ref{ineq:h-lower}) imply the stated bound: $\max\limits_{j \in [n] \setminus \supp(w)} p(A_j) \leq \frac{\lambda}{2n} \leq p(A_h)$.

\noindent
\emph{Establishing Inequality (\ref{ineq:zero-upper}).} As mentioned above, in the welfare-maximizing allocation $\alloc \in \OPT(\textbf{\textrm{LP}}(w,\eta))$, each item $t \in [\mbar]$ is assigned to an agent in $\argmax_{i \in [n]} \ (w_i + \eta) \vbar_i(t)$. In particular, for each agent $b \in [n] \setminus \supp(w)$ and each item $s \in A_b$, it holds that $p_s  = \max_{i \in [n]} (w_i + \eta)\vbar_i(t) = \eta \ \vbar_b(s)$; recall that $w_b = 0$, since $b \in [n] 
\setminus \supp(w)$. Since the auxiliary item $(m+1)$ is assigned to the agent $h \notin [n] \setminus \supp(w)$, for each $b \in [n] \setminus \supp(w)$, we have $|A_b| \leq m$. These observations give us 
\begin{align}
p(A_b) = \sum_{s \in A_b} p_s  = \sum_{s \in A_b} \eta \ \vbar_b(s) \leq  |A_b| \ \eta \ \max_{i \in [n], \ t \in [\mbar]} |\vbar_i(t)| \leq m \eta \ \max_{i \in [n], \ t \in [\mbar]} |\vbar_i(t)| = \frac{\lambda}{2n} 
\label{ineq:price-low}
\end{align}
The last equality follows from the definition of $\eta$ (Equation (\ref{definition:value-of-eta})). Since Equation (\ref{ineq:price-low}) holds for each $b \in [n] \setminus \supp(w)$, inequality (\ref{ineq:zero-upper}) follows. \\

\noindent
\emph{Establishing Inequality (\ref{ineq:h-lower}).} To lower bound $p(A_h)$, we first note that the price of item $(m+1) \in A_h$ satisfies $p_{m+1} = \frac{\lambda}{2}  \ \max\limits_{i \in [n]} (w_i + \eta) \geq \frac{\lambda}{2n}$; here, the last inequality follows from the fact that $\max_{i \in [n]} w_i \geq 1/n$. 

Now, recall that in the instance $\barI = \angleb{[n], [\mbar], \{\vbar_i \}_{i=1}^n}$ the set of items $[\mbar]$ can be partitioned into three disjoint sets: goods $G$, chores $C$, and non-negatively valued items $G_0$ (\Cref{lemma:three-types}). We will next show that, in allocation $\alloc$, no chore is assigned to agent $h$. Since only chores have negative prices (\Cref{proposition:price-non-negativity}), this will give us inequality (\ref{ineq:h-lower}): $p(A_h) \geq p_{m+1} \geq \frac{\lambda}{2n}$. 

Specifically, for any chore $c \in C$ and agent $b \in [n] \setminus \supp(w)$, we will prove 
\begin{align}
\eta \vbar_b(c) > (w_h + \eta) \vbar_h(c) \label{ineq:chore-h}
\end{align}
 Since, in allocation $\alloc$, chore $c$ must be assigned to an agent in $\argmax_{i \in [n]} (w_i + \eta) \vbar_i(c)$, inequality (\ref{ineq:chore-h}) implies that $c$ is not assigned to $h$ in $\alloc$. 
 
Here, since $c$ is a chore, $\vbar_h(c) < 0$ and $v_h(c) < 0$ (see \Cref{proposition:signs}). Moreover, $|\vbar_h(c)| = \left|v_h(c) - \varepsilon_{h,c} \right| \geq \left| v_h(c)\right|$; the last inequality holds since $v_h(c) < 0$ and $\varepsilon_{h,c} \geq 0$. Further, using the definition of $\lambda$ we get $|\vbar_h(c)|  \geq |v_h(c)| \geq \lambda$. This bound implies 
\begin{align}
(w_h + \eta) |\vbar_h(c)| & \geq (w_h + \eta) \lambda \nonumber \\
& \geq \frac{\lambda}{n} \tag{$w_h \geq 1/n$ and $\eta >0$} \\
& >  \eta \max_{i \in [n], \ t \in [\mbar]} | \vbar_i(t)| \tag{Equation (\ref{definition:value-of-eta})} \\
& \geq \eta |\vbar_b(c)| \label{ineq:abs-reverse}
\end{align}  
Since $c$ is a chore, both $\vbar_b(c) <0$ and $\vbar_h(c)<0$. Hence, Equation (\ref{ineq:abs-reverse}) reduces to the stated strict inequality (\ref{ineq:chore-h}). That is, no chore $c$ is assigned to agent $h$ and, as stated, we obtain inequality (\ref{ineq:h-lower}): $p(A_h) \geq p_{m+1} \geq \frac{\lambda}{2n}$. \\

In summary, we have inequalities (\ref{ineq:zero-upper}) and (\ref{ineq:h-lower}). As mentioned above, these bounds imply that there exists an agent $i' \in \supp(w)$ with a maximum-price bundle, $i' \in \argmax_{i \in [n]} p(A_i)$. Overall, for this agent $i' \in \supp(w)$, we have $w \in C_{i'}$ and, hence, the KKM covering condition holds in the current case as well.  This completes the proof of the lemma. 
\end{proof}

Lemmas \ref{lemma:KKM-closedness} and \ref{lemma:KKM-covering} enable us to invoke the KKM Theorem with sets $C_i$s as defined in \Cref{definition:KKM-sets-a}. This invocation proves that there necessarily exists a weight vector $w^*$ with the property that, for each agent $i \in [n]$, we have a weighted welfare-maximizing allocation $\alloc^{(i)} =(A^{(i)}_1, \ldots, A^{(i)}_n)$ wherein agent $i$ is price envy-free. Formally, 
\begin{theorem}\label{theorem:KKM-application}
There exists a vector $w^* \in \Delta_{n-1}$ with the property that, for each $i \in [n]$, we have an allocation $\alloc^{(i)} \in \OPT(\textbf{\textrm{LP}}(w^*,\eta))$ wherein agent $i \in \argmax\limits_{j \in [n]}  \ p^*\left(A^{(i)}_j \right)$; here $p^* = \OPT(\textbf{\textrm{Dual-LP}}(w^*,\eta))$.
\end{theorem}
\begin{proof}
Consider the sets $C_1, \ldots, C_n \subseteq \Delta_{n-1}$ defined in \Cref{definition:KKM-sets-a}. \Cref{lemma:KKM-closedness} ensures that $C_1, \ldots, C_n$ are closed. In addition, via \Cref{lemma:KKM-covering}, we have that these sets satisfy the KKM covering condition. Hence, invoking the KKM Theorem (\Cref{theorem:KKM}), we obtain that there exists a vector $w^* \in \cap_{i=1}^n C_i$. Write $p^* = \OPT(\textbf{\textrm{Dual-LP}}(w^*,\eta))$. For each $i \in [n]$, we have $w^* \in C_i$. Hence, for each $i \in [n]$, by the definition of $C_i$, we have an allocation $\alloc^{(i)} \in \OPT(\textbf{\textrm{LP}}(w^*,\eta))$ wherein agent $i \in \argmax_{j \in [n]}{p^* \left(A^{(i)}_j \right)}$. 
The theorem stands proved. 
\end{proof}

%% file: 5.PhaseII.tex
\section{The Augmenting Tree Algorithm}\label{section:phaseII}
This section proves the fairness and efficiency guarantee, as stated in \Cref{theorem:nondegenerate-fair-efficient}, for the non-degenerate instance $\barI$. Towards this, we start with weight vector $w^*$ and the corresponding allocations $\alloc^{(1)}, \ldots, \alloc^{(n)}$, whose existence was established in \Cref{theorem:KKM-application}. Then, we constructively prove that the existence of such a collection of allocations (which are price envy-free for the $n$ agents, respectively) implies the desired existence of a single allocation $\alloc^*$ that is $\IEFo$ and welfare-maximizing in $\barI$. Here, our constructive proof relies on a novel Augmenting Tree Algorithm.  

\subsection{Price Graph}
Our subsequent analysis is based on a bipartite subgraph $H$ which we define considering the weight vector $w^* \in \Delta_{n-1}$, identified in \Cref{theorem:KKM-application}, and the associated optimal price vector $p^* =  \OPT(\textbf{\textrm{Dual-LP}}(w^*,\eta))$. 

We first consider, for each agent $i \in [n]$, the items $s$ for which $i$ is the unique agent in $\argmax_{j \in [n]} \ (w^*_j + \eta) \vbar_j(s)$. In particular, write bundle $L_i \coloneqq \left\{ s \in [\mbar] \mid (w^*_i + \eta) \vbar_i(s) >  (w^*_j + \eta) \vbar_j(s) \text { for all } j \neq i \right\}$, for each $i \in [n]$.\footnote{The set $L_i$ can be empty.} Write $\mathcal{L} = (L_1, \ldots, L_n)$ to denote the partial allocation obtained via these bundles. 

Recall that in any welfare-maximizing allocation $\alloc \in \OPT(\LP(w^*, \eta))$, each item $t \in [m]$ is assigned to an agent in $\argmax_{j \in [n]} \ (w^*_j + \eta) \vbar_j(t)$. Hence, in any welfare-maximizing allocation $\alloc =(A_1, \ldots, A_n)$ the containment $L_i \subseteq A_i$ holds for each $i \in [n]$. 

We define the bipartite graph $H = ([n] \cup U, E)$ with the left part as the set of agents, $[n]$, and the right part $U \coloneqq [\mbar] \setminus \left( \cup_{i=1}^n L_i\right)$. The definition of $L_i$s implies that the right part $U$ consists of the items $t$ with $\left| \argmax_{j \in [n]} \ (w^*_j + \eta) \vbar_j(t)\right| \geq 2$. Moreover, we include $(i,t) \in [n] \times U$ in the edge set $E$ iff $i \in \argmax_{j \in [n]} (w^*_j + \eta) \vbar_j(t)$. 

We will use $\Gamma(x)$ to denote the vertices adjacent to each $x \in [n] \cup U$ in $H = ([n] \cup U, E)$. Note that the construction of the edge set $E$ and the definition of $U$ ensure that each item $t \in U$ has degree at least two in the bipartite graph, $|\Gamma(t)| \geq 2$. 

The lemma below provides a key property of $H$: it is acyclic.

\begin{lemma}\label{lemma:H-is-forest}
The bipartite graph $H = ([n] \cup U, E)$ is a forest.
\end{lemma}
\begin{proof}
First, we note that in the bipartite graph $H = ([n] \cup U, E)$, for every edge $(i,t) \in E$ we have $\vbar_i(t) \neq 0$: The optimal prices $p^*_t \neq 0$ for each item $t \in [\mbar]$ (\Cref{proposition:price-non-negativity}). In addition, by construction of the edge set $E$, for every $(i,t) \in E$, we have $p^*_t = (w^*_i + \eta) \vbar_i(t) \neq 0$. Since $(w^*_i + \eta) > 0$, it holds that $\vbar_i(t) \neq 0$.

 Assume, towards a contradiction, that $H$ admits a simple cycle $\mathcal{C} = (i_1, t_1, i_2, t_2, \ldots, i_k, t_k, i_1)$, with $i_1, \ldots, i_k \in [n]$ and $t_1, \ldots, t_k \in U$.
 Here, each $(i_\ell, t_\ell), (t_\ell, i_{\ell+1}) \in E$ and, hence, as noted above, the values $\vbar_{i_\ell}(t_\ell), \vbar_{i_{\ell +1}} (t_\ell) \neq 0$, for each $\ell \in [k]$.\footnote{By convention we set $i_{k+1} = i_{1}$.} Therefore, the product $\prod_{\ell = 1}^{k} \left( \frac{\vbar_{i_{\ell+1}}(t_\ell)}{\vbar_{i_\ell}(t_{\ell})} \right)$ is nonzero and satisfies 
 \begin{align}
\prod_{\ell = 1}^{k} \left( \frac{\vbar_{i_{\ell+1}}(t_\ell)}{\vbar_{i_\ell}(t_{\ell})} \right) = \prod_{\ell = 1}^{k} \left( \frac{(w^*_{i_{\ell+1}} + \eta) \vbar_{i_{\ell+1}}(t_\ell)}{ (w^*_{i_\ell} + \eta) \vbar_{i_\ell}(t_{\ell})}  \right) \label{eq:tele}
\end{align}
 The last equality is obtained by multiplying and dividing by the positive terms $\{ (w^*_{i_\ell} + \eta)\}_{\ell \in [k]}$  to the left-hand-side of the equation, and a telescoping shift. 
 
 Since both $(i_\ell, t_\ell), (t_\ell, i_{\ell+1}) \in E$, we have $p^*_{t_\ell} = (w^*_{i_{\ell+1}} + \eta) \vbar_{i_{\ell+1}}(t_\ell) = (w^*_{i_\ell} + \eta) \vbar_{i_\ell}(t_{\ell})$. 
That is, each ratio $ \left( \frac{(w^*_{i_{\ell+1}} + \eta) \vbar_{i_{\ell+1}}(t_\ell)}{ (w^*_{i_\ell} + \eta) \vbar_{i_\ell}(t_{\ell})}  \right) = 1$ and, hence, Equation (\ref{eq:tele}) reduces to $\prod_{\ell = 1}^{k} \left( \frac{\vbar_{i_{\ell+1}}(t_\ell)}{\vbar_{i_\ell}(t_{\ell})} \right) = 1$. 
 
This, however, contradicts \Cref{lemma:non-degeneracy}, since $\mathcal{C}$ is a cycle in the complete bipartite graph $K_{n, \mbar}$ as well. Hence, by way of contradiction, we obtain that $H$ is acyclic, i.e., a forest. The lemma stands proved. 
\end{proof}

The following lemma provides another useful property of the bipartite graph $H$: Any completion of the partial allocation $\mathcal{L} = (L_1, \ldots, L_n)$ by assigning the items in $U = [\mbar] \setminus \left(\cup_{i \in [n]} \  L_i \right) $ along the edges in $H$ provides a welfare-maximizing allocation.  

Specifically, define the collection of (complete) allocations 
\begin{align}
\calO^* \coloneqq \left\{ (B_1, \ldots, B_n) \in \Pi_n([\mbar]) \mid  L_i \subseteq B_i \subseteq L_i \cup \Gamma(i) \text{ for each } i \in [n] \right\} \label{eqn:defn-ostar}
\end{align}
We next show that $\calO^*$ is the same as $\OPT(\textbf{\textrm{LP}}(w^*,\eta))$, the set of welfare-maximizing allocations. 

\begin{lemma}\label{lemma:O-star-is-opt}
 Any allocation $\calB \in \calO^*$ if and only if $\calB \in \OPT(\textbf{\textrm{LP}}(w^*,\eta))$. 
\end{lemma}
\begin{proof}
If $\calB = (B_1, \ldots, B_n) \in \calO^*$, then, by definition, for all $i \in [n]$, we have $L_i \subseteq B_i \subseteq L_i \cup \Gamma(i)$. Hence, if item $t \in B_i$, then $(i,t) \in E$, i.e., $p^*_t = (w^*_i + \eta) \vbar_i(t)$. This implies that the objective value of $\textbf{\textrm{LP}}(w^*,\eta)$ for the allocation $\calB$ is $\sum_{i \in [n]} \sum_{t \in B_i} (w^*_i + \eta) \vbar_i(t) = \sum_{t \in [\mbar]} p^*_t$. Note that $\sum_{t \in [\mbar]} p^*_t$ is also the objective value of $\textbf{\textrm{Dual-LP}}(w^*,\eta)$. Hence, by strong duality, we get that $\calB \in \OPT(\textbf{\textrm{LP}}(w^*,\eta))$.

To prove the other direction, consider any $\calB' \in \OPT(\textbf{\textrm{LP}}(w^*,\eta))$. Optimality of $\calB' = (B'_1, \ldots, B'_n)$ implies that if $t \in B'_i$, then $(w^*_i + \eta) \vbar_i(t) = \max_{j \in [n]} \ (w^*_j + \eta) \vbar_j(t)$. 

For item $t$, if $i$ is the unique agent in $\argmax_{j \in [n]} \ (w^*_j + \eta) \vbar_j(t)$, then $t \in L_i$. Otherwise, by construction of the edge set $E$ in the bipartite graph $H$, we have $(i,t) \in E$. Hence, for each $i \in [n]$, we have $L_i \subseteq B'_i \subseteq L_i \cup \Gamma(i)$, implying that $\calB' \in \calO^*$. This completes the proof of the lemma. 
\end{proof}

\subsection{Min-Max Allocations}
\label{subsection:min-max-allocations}
A key construct in our analysis is the min-max price threshold $\tau \in \mathbb{R}$, defined as 
\begin{align}
\tau \coloneqq \min_{(B_1, \ldots, B_n) \in \calO^*} \ \ \max_{j \in [n]} \ p^*(B_j) \label{eqn:defn-tau}
\end{align}

In addition, write $\calM^*$ to denote all the min-max allocations 
\begin{align}
\calM^* = \left\{ (B_1, \ldots, B_n) \in \calO^* \mid  \max_{j \in [n]} \ p^*(B_j) = \tau \right\} \label{eqn:defn-mstar}
\end{align}
By definition, $\calM^* \subseteq \calO^*$. 

Additionally, for each agent $i \in [n]$ and any subset of items $S$, we define $p^+_i(S)$ as the maximum price that can be achieved by toggling (adding or removing) at most one item from $\Gamma(i)$ within $S$.\footnote{Recall that $\Gamma(i)$ denotes the set of neighbors of $i$ in the bipartite graph $H$. Also, if item $t \in S$, then it holds that $S\triangle \{ t \} = S \setminus \{ t \}$. Otherwise, for $t \notin S$, we have $S \triangle \{t\} =S \cup \{t\}$.}
\begin{align}
p^+_i(S) \coloneqq \max\limits_{\substack{T \subseteq \Gamma(i), \\ |T| \leq 1}} \ p^*(S \triangle T)
\end{align} 
Note that while the item prices, $p^*_t$s, are the same for all the agents, in defining $p^+_i(\cdot)$ we add or remove an item specifically from $\Gamma(i)$. Also, for any agent $i$ and subset $S$, it holds that $p^+_i(S) \geq p^*(S)$. 

We will work with a potential $\varphi^+(\cdot)$ that, for any allocation $\calB = (B_1, \ldots, B_n)$, counts the number of agents whose bundle prices can be toggled to at least $\tau$.
\begin{align}
\varphi^+(\calB) = \Big| \left\{ i \in [n] \mid p^+_i(B_i) \geq \tau \right\} \Big| \label{eqn:defn-potential}
\end{align}

For any min-max allocation $\calB \in \calM^*$ (since $\max_{j \in [n]} p^*(B_j) = \tau$), we have $\varphi^+(\calB) \geq 1$. 

The following theorem is a key result of this section, and it shows that there exists a min-max allocation $\calA^*$ wherein the bundle prices of all the agents can be toggled to $\tau$, or above. We provide a constructive proof of \Cref{theorem:leveled-allocation-is-fair-n-efficient} in Section \ref{subsection:toggle-theorem}. 

Before that, we restate and prove \Cref{theorem:nondegenerate-fair-efficient} using this result. This connection essentially follows from the fact that, for each agent $i$, a price guarantee up to one exchange (i.e., under $p^+_i(\cdot))$ translates into a valuation $\vbar_i(\cdot)$ guarantee up to one exchange. 

\begin{theorem}\label{theorem:leveled-allocation-is-fair-n-efficient}
There exists an allocation $\calA^* \in \calM^*$ such that $\varphi^+(\calA^*) = n$. 
\end{theorem}

\subsection{Proof of \Cref{theorem:nondegenerate-fair-efficient}}
\label{subsection:leveled-allocation-works}

\nonDegFairEfficient*
\begin{proof}
Let $w^* \in \Delta_{n-1}$ be the weight vector whose existence is guaranteed by \Cref{theorem:KKM-application}, and $p^* = \OPT(\Dual(w^*, \eta))$ be the associated optimal prices. Further, the collection of allocations $\calO^*$ and $\calM^*$ are as defined in Equations (\ref{eqn:defn-ostar}) and (\ref{eqn:defn-mstar}), respectively. We will consider the allocation $\alloc^* \in \calM^*$ provided in \Cref{theorem:leveled-allocation-is-fair-n-efficient}. 

Since allocation $\calA^* \in \calM^* \subseteq \calO^*$, \Cref{lemma:O-star-is-opt} implies that, as desired, allocation $\alloc^*$ is welfare maximizing, $\alloc^* \in \OPT(\textbf{\textrm{LP}}(w^*,\eta))$.

We will complete the proof by proving next that $\alloc^*=(A^*_1, \ldots, A^*_n)$ is $\EFEo$ in the instance $\barI$. For allocation $\alloc^*$ we have $\varphi^+(\alloc^*)= n$. Hence, for each agent $i \in [n]$, it holds that 
\begin{align}
p^+_i(A^*_i) \geq \tau = \max_{j \in [n]} p^*(A^*_j) \label{ineq:plus-max}
\end{align} 
The inequality follows from the fact that $\alloc^* \in \calM^*$. Furthermore, using the definition of $p^+_i(\cdot)$ we get that there exists a set $T \subseteq \Gamma(i)$---of size at most one---such that $p^+_i(A^*_i) = p^*(A^*_i \triangle T)$. Hence, inequality (\ref{ineq:plus-max}) reduces to the following bound for agent $i \in [n]$ and each $j \in [n]$: 
\begin{align}
p^*(A^*_i \triangle T) \geq p^*(A^*_j) \label{ineq:plus-max-a}
\end{align}
Since $\alloc^* \in \calO^*$, we have $L_i \subseteq A^*_i \subseteq L_i \cup \Gamma(i)$; see Equation (\ref{eqn:defn-ostar}). Also, $T \subseteq \Gamma(i)$. Hence, we have $A^*_i \triangle T \subseteq L_i \cup \Gamma(i)$. Note that for each item $t \in L_i \cup \Gamma(i)$ and, hence, for each $t \in A^*_i \triangle T$ it holds that $(w^*_i + \eta) \ \vbar_i(t) = p^*_t$. Using this equality for all the items in $A^*_i \triangle T$, we extend Equation (\ref{ineq:plus-max-a}) to 
\begin{align}
(w^*_i + \eta) \ \vbar_i(A^*_i \triangle T) = p^*(A^*_i \triangle T) \geq p^*(A^*_j) \geq (w^*_i + \eta) \ \vbar_i(A_j^*) \label{ineq:i-to-j}
\end{align} 
The last inequality follows from the fact that, for each item $s \in [\mbar]$, the optimal price $p^*_s = \max_{k \in [n]} \ (w^*_k + \eta) \vbar_k(s) \geq (w^*_i + \eta) \vbar_i(s)$. 

Since $(w^*_i + \eta) > 0$, Equation (\ref{ineq:i-to-j}) gives us $\vbar_i(A^*_i \triangle T) \geq \vbar_i(A^*_j)$, for each $j \in [n]$. This bound implies that $\alloc^*$ is $\EFEo$ in $\barI$, since for each agent $i \in [n]$, there exists a subset $T \subseteq [\mbar]$---of size at most one---such that $\vbar_i(A^*_i \triangle T) \geq \max_{j \in [n]} \vbar_i(A^*_j)$. The theorem stands proved. 
\end{proof}

\subsection{Proof of \Cref{theorem:leveled-allocation-is-fair-n-efficient}}
\label{subsection:toggle-theorem}

In this section, to establish \Cref{theorem:leveled-allocation-is-fair-n-efficient}, we develop an algorithm (\Cref{algo:augmenting-forest}) that takes as input any allocation $\calS \in \calM^*$, with $\varphi^+(\calS) < n$, and successfully returns another allocation $\widehat{\calS} \in \calM^*$ with a strictly higher $\varphi^+(\widehat{\calS})$. Since the integer-valued potential $\varphi^+(\cdot)$ is at most $n$ (see Equation (\ref{eqn:defn-potential})), starting with an arbitrary allocation $\alloc \in \calM^* \neq \emptyset$ and applying the algorithm at most $n$ times yields the desired allocation $\alloc^* \in \calM^*$, with $\varphi^+(\alloc^*) = n$.

We begin by describing the \textsc{AugmentTree Algorithm} (\Cref{algo:augmenting-forest}). \textsc{AugmentTree Algorithm} starts with input allocation $\calS = (S_1, \ldots, S_n) \in \calM^*$ for which $\varphi^+(\calS) < n$. Hence, there exists an agent $r \in [n]$ such that $p^+_r(S_r) < \tau$. In the bipartite $H=([n] \cup U, E)$, we consider the connected component that contains $r$. Given that $H$ is a forest (\Cref{lemma:H-is-forest}), the identified connected component is a tree, and we root this tree at $r$ itself. We will, throughout, write $R$ to denote this rooted tree. 

Also, for each node $x$ in the rooted tree $R$, write $\pr(x) \in \Gamma(x)$ to denote the (unique) parent of $x$. Since $H$ is bipartite between agents and items, for each agent $i \in [n]$ in the tree, the parent $\pr(i)$ is an item. Conversely, for any item $t$ in the tree, $\pr(t)$ is an agent. Also, for the root $r$ we set $\pr(r) = \emptyset$.

Furthermore, let $\ch(x) = \Gamma(x) \setminus \{ \pr(x) \}$ denote the children of node $x$ in the rooted tree $R$. Also, as mentioned previously, each item $t$ in the bipartite graph $H$ (and, hence, in $R$) has degree at least two. Therefore, for each item $t$ in the tree, we have $\ch(t) \neq \emptyset$. 

Finally, for each agent $i$ in the rooted tree, we will use $\cho(i)$ to denote the agents that are immediate descendants of $i$, i.e., $\cho(i) \coloneqq \bigcup_{t \in \ch(i)} \ \ch(t)$.\footnote{If, for any agent $i$, the set $\ch(i)$ is empty, then so is $\cho(i)$.}

The following proposition will be used in the analysis of the algorithm -- it directly follows from the fact that, in a rooted tree, no two distinct nodes can have a common child.\footnote{Indeed, two distinct nodes can have a common parent.}
\begin{proposition}
\label{prop:no-child}
Let $i \in [n]$ be an agent in the rooted tree $R$, and let item $t \in \ch(i)$. If $t$ is a neighbor of an agent $a \in [n] \setminus \{i\}$ (i.e., $t \in \Gamma(a)$ and $a \neq i$), then necessarily $\pr(a) = t$ and $a \in \cho(i)$. 
\end{proposition}

We start the \textsc{AugmentTree Algorithm} by initializing the bundles $\widehat{S}_j = S_j$ for the agents $j$. Then, in the main while-loop of the \textsc{AugmentTree Algorithm}, we execute a tree traversal using a set $\calQ$, which is initialized to $\{r\}$ (\Cref{line:root-node-in-queue}). In each iteration of the while-loop, we process an arbitrary agent $i \in \calQ$. At the beginning of the while-loop iteration we have $p^+_i(\widehat{S}_i) < \tau$. Though, by the end of the loop iteration (by Line \ref{line:i-bulk-update}) we ensure that, $\widehat{S}_i$, the updated bundle of agent $i$, satisfies $p^+_i(\widehat{S}_i) \geq \tau$. To achieve this increase, we exchange items between agent $i$ and the agents in $\cho(i)$, i.e., between $i$ and its  immediate descendants (see Line \ref{line:definition-of-X}). In our key lemma (\Cref{lemma:forest-augmentation-existence}), we prove that such a $p^+_i(\cdot)$-improving exchange of items between $i$ and agents in $\cho(i)$ is always possible.

If, after this exchange, an agent $a \in \cho(i)$ is left with a bundle satisfying $p^+_{a} (\widehat{S}_{a} ) < \tau$, then we add $a$ to $\calQ$ (Lines \ref{line:add-to-Q}), in order to process it subsequently and ensure that the final bundle of $a$ satisfies $p^+_{a} (\widehat{S}_{a} ) \geq \tau$. Proceeding in this manner, the algorithm iterates as long as $\calQ \neq \emptyset$. We prove that the algorithm terminates and always returns an allocation $\widehat{\calS}$ such that $\varphi^+(\widehat{\calS}) > \varphi^+(\calS)$. In addition, the algorithm maintains throughout that $\max_{j \in [n]} p^*(\widehat{S}_j) \leq \tau$, i.e., even the returned allocation $\widehat{\calS} \in \calM^*$. As mentioned previously, this guaranteed success of the algorithm proves \Cref{theorem:leveled-allocation-is-fair-n-efficient}.\\

\begin{algorithm}[h]
   \caption{\textsc{AugmentTree Algorithm}}\label{algo:augmenting-forest}
   
   \SetKwInOut{Input}{Input}
   \SetKwInOut{Output}{Output}
   \SetKwFor{While}{while}{do}{end}
   \SetKwFor{ForAll}{for all}{do}{end}
   \SetKwIF{If}{ElseIf}{Else}{if}{then}{else if}{else}{end}
   
   \SetAlgoNlRelativeSize{-1}  
   \DontPrintSemicolon          

   \Input{Allocation $\calS = (S_1, \dots, S_n) \in \calM^*$  (i.e., $\max_j p^*(S_j)=\tau$) with $\varphi^+(\calS) < n$.}
   \Output{Allocation $\widehat{\calS} = (\widehat{S}_1, \dots, \widehat{S}_n) \in \calM^*$ satisfying $\varphi^+(\widehat{\calS})  > \varphi^+(\calS)$.}
  
Let $r \in [n]$ be any agent with $p^+_r(S_r) < \tau$. Initialize set $\calQ = \{r\}$.\label{line:root-node-in-queue}\;

As defined above, $R$ is the tree rooted at $r$ and, for each node $x$ in the tree, $\pr(x)$ denotes the parent of $x$. In addition, $\ch(x) = \Gamma(x) \setminus \{ \pr(x) \}$ are the children of $x$. Also, for each agent $i$ in the tree, $\cho(i) = \cup_{t \in \ch(i)} \ \ch(t)$ denotes the agents that are immediate descendants of $i$.\label{line:rooting-the-tree}\; 
Initialize $\widehat{S}_j = S_j$ for all agents $j \in [n]$.\;
   \BlankLine
   \While{$\calQ \neq \emptyset$}{
       Let $i \in \calQ$ be any agent in $\calQ$. Update $\calQ \gets \calQ \setminus \{i\}$. \label{loopInvariant:i-needs-a-fix}\; 
       \BlankLine
       Set $X \subseteq \ch(i)$ to be an item subset that satisfies: (i) $p^+_i( \widehat{S}_i \triangle X) \geq \tau$,  (ii) $p^*(\widehat{S}_i \triangle X) \leq \tau$, and (iii) $p^*(\widehat{S}_i \triangle \{t\}) > p^*( \widehat{S}_i)$ for each $t \in X$. 
       {\color{teal}\{Existence of $X$ is guaranteed via \Cref{lemma:forest-augmentation-existence}.\}}\label{line:definition-of-X}\;
       \BlankLine
       \ForAll{items $t \in X$}{
        If item $t \in \widehat{S}_i$, then set $a$ to be any agent in $\ch(t) \neq \emptyset$. Otherwise, if $t \notin \widehat{S}_i$, then let $a \in \ch(t)$ be the agent such that $t \in \widehat{S}_a$. 
      {\color{teal}{\{We prove that such an agent $a \in \cho(i)$ exists.\}}}\label{line:agent-a}\;
      Update $\widehat{S}_{a} \gets \widehat{S}_{a} \triangle \{t\}$.\label{line:update-a}\; 
           \BlankLine
      If $p^+_{a}(\widehat{S}_{a}) < \tau$, then include $a$ for processing: $\calQ \gets \calQ \cup \{a\}$. {\color{teal}{\{Otherwise, if $p^+_{a}(\widehat{S}_{a}) \geq \tau$, then we do not include $a$ in $\calQ$.\}}}\label{line:add-to-Q}\;  
       } 
       \BlankLine
       Update $\widehat{S}_i \gets \widehat{S}_i \triangle X$\label{line:i-bulk-update}\;
   }
   \Return{Allocation $\widehat{\calS} = (\widehat{S}_1,\ldots, \widehat{S}_n)$.} \; 
   \end{algorithm}

Before analyzing the algorithm, we prove a key lemma that encapsulates the existential guarantee of \Cref{theorem:KKM-application}. 

 \begin{lemma}\label{lemma:forest-augmentation-existence}
 For any agent $i \in [n]$ in the rooted tree, let $S$ be an item subset such that $L_i \subseteq S \subseteq L_i \cup \Gamma(i)$ and $p^+_i(S) < \tau$. Then, there exists subset $X \subseteq \ch(i)$ satisfying the following properties
    \begin{enumerate}
       \item[(i)] $p^+_i(S \triangle X) \geq \tau$,
       \item[(ii)] $p^*(S \triangle X) \leq \tau$, and
       \item[(iii)] $p^*(S \triangle \{t\}) > p^*(S)$ for each $t \in X$.
    \end{enumerate}
 \end{lemma}
 \begin{proof}
For the agent $i \in [n]$, consider the allocation $\alloc^{(i)} \in \OPT(\textbf{\textrm{LP}}(w^*,\eta))$ from \Cref{theorem:KKM-application}, in which agent $i$ is price envy-free; in particular, $p^*(A^{(i)}_{i}) = \max_{j \in [n]} \ p^*(A^{(i)}_j)$. 
 
\Cref{lemma:O-star-is-opt} implies $\alloc^{(i)} \in \calO^*$ and, hence, we have $L_i \subseteq A^{(i)}_{i} \subseteq L_i \cup \Gamma(i)$; see Equation (\ref{eqn:defn-ostar}). Furthermore, the definition of $\tau$ (Equation (\ref{eqn:defn-tau})) gives us $p^*(A^{(i)}_{i}) = \max_{j \in [n]} \ p^*(A^{(i)}_j) \geq \tau$. 

Now, consider the set
\begin{align}
K \coloneqq \left\{t \in A^{(i)}_{i} \triangle S  \mid p^*(S \triangle \{t\}) > p^*(S) \right\} \label{eqn:defn-K}
\end{align}
The set $K$ consists of the chores in $S \setminus A^{(i)}_{i}$ and the goods (along with non-negatively valued items) in $A^{(i)}_{i} \setminus S$. 
Moreover, note that the bound $p^*\left(A^{(i)}_{i} \right) \geq \tau$ implies that for the set $K$ we have 
\begin{align}
p^*(S \triangle K) \geq \tau \label{ineq:imp-toggle}
\end{align}
Further, by the lemma assumption, the given set $S$ satisfies $p^+_i(S) < \tau$ and, hence, $p^*(S) < \tau$. This bound and inequality (\ref{ineq:imp-toggle}) 
ensure that $K \neq \emptyset$. We define $\widetilde{K} \subseteq K$ to be a minimal subset of $K$ with the property that $p^*\left(S \triangle \widetilde{K} \right) \geq \tau$. As in the case of $K$, it holds that $\widetilde{K} \neq \emptyset$. Finally, we obtain the desired subset $X$ as follows
 \begin{equation}\label{definition:set-X-existence}
       X \coloneqq \begin{cases}
          \widetilde{K} \setminus \{\pr(i)\}, & \quad \text{if } \pr(i) \in \widetilde{K},\\
          \widetilde{K} \setminus \{\kappa \}  \text{ for any } \kappa \in \widetilde{K} & \quad \text{otherwise, if } \pr(i) \notin \widetilde{K}.
       \end{cases}
    \end{equation}
We first show that $X \subseteq \ch(i)$. Towards this, note the containments $L_i \subseteq A^{(i)}_{i} \subseteq L_i \cup \Gamma(i)$ along with $L_i \subseteq S \subseteq L_i \cup \Gamma(i)$. Hence, $\Gamma(i) \supseteq A^{(i)}_{i} \triangle S \supseteq K \supseteq \widetilde{K}$. By construction of $X$, we have $\pr(i) \notin X$. Therefore, $X \subseteq \Gamma(i) \setminus \{ \pr(i) \} = \ch(i)$.
    
We will complete the proof by establishing that $X$ satisfies the three properties stated in the lemma. First, by the minimality of $\widetilde{K}$, we have $p^*(S \triangle X) < \tau$; note that, by definition of set $K$ and the fact that $X \subsetneq \widetilde{K} \subseteq K$, we have $p^*(S \triangle X) < p^*(S \triangle \widetilde{K})$. This establishes property (ii) from the lemma.

Next, note that $X$ is obtained by removing one item from $\widetilde{K} \subseteq \Gamma(i)$ and, hence, $p^+_i(S \triangle X) \geq p^*\left(S  \triangle \widetilde{K}\right) \geq \tau$. That is, property (i) holds as well. 

For property (iii), recall that, by definition, each item $t \in K$ satisfies $p^*(S \triangle \{t\}) > p^*(S)$. Since $X \subseteq \widetilde{K} \subseteq K$, the property follows. This completes the proof of the lemma. 
 \end{proof}

We analyze \textsc{AugmentTree Algorithm} by showing that each iteration of its while-loop executes successfully (\Cref{lemma:loop-succeeds}) and that the while-loop necessarily terminates (\Cref{corollary:termination}). Finally, we establish \Cref{theorem:leveled-allocation-is-fair-n-efficient} by proving that the algorithm always returns a complete allocation $\widehat{\calS} \in \calM^*$ with an increased $\varphi^+(\cdot)$ value. 

\begin{lemma}
\label{lemma:loop-succeeds}
Each iteration of the while-loop in the \textsc{AugmentTree Algorithm} executes successfully and maintains an allocation $\widehat{\calS} \in \calM^*$. Furthermore, at the end of the iteration (Line~\ref{line:i-bulk-update}), for the agent $i$ selected in that iteration, it holds that $p^+_i(\widehat{S}_i) \geq \tau$.
\end{lemma}
\begin{proof}
Fix any iteration of the while-loop, and let $i \in \calQ$ be the agent under consideration. At the beginning of this iteration, for the current allocation $\widehat{\calS}=(\widehat{S}_1, \ldots, \widehat{S}_n)$, we have $p^+_i(\widehat{S}_i) < \tau$. This follows from the fact that any agent $i$ (including the root $r$) is added to $\calQ$ only when $p^+_i(\widehat{S}_i) < \tau$. In addition, inductively, we have that the current allocation $\widehat{\calS} \in \calM^* \subseteq \calO^*$. Hence, $L_i \subseteq \widehat{S}_i \subseteq L_i \cup \Gamma(i)$. These observations allow us to apply \Cref{lemma:forest-augmentation-existence} with $S = \widehat{S}_i$ and obtain an item subset $X \subseteq \ch(i)$ that satisfies the requirements specified in Line~\ref{line:definition-of-X}.

Next, we observe that, for each item $t \in X \subseteq \ch(i)$, the desired agent $a$ (Line \ref{line:agent-a}) exists: Recall that $\ch(s) \neq \emptyset$ for any item $s$.  
Also, the current allocation $\widehat{\calS} \in \calM^* \subseteq \calO^*$ and, hence, item $t$ must be assigned (in $\widehat{\calS}$) to an agent in $\Gamma(t)$. 
Therefore, if $t \in \widehat{S}_i$, then (as stated in Line \ref{line:agent-a}) we can select any agent $a \in \ch(t)$. Otherwise, if $t \notin \widehat{S}_i$, it must be the case that $t \in \widehat{S}_a$ for an agent $a \in \ch(t)$ that satisfies $\pr(a) = t$ and $a \in \cho(i)$; see Proposition \ref{prop:no-child}. 

The existence of subset $X$ and agents $a$ ensure that the iteration executes successfully. Further, for the selected subset $X$ (\Cref{lemma:forest-augmentation-existence}) we have $p^+_i(\widehat{S}_i \triangle X) \geq \tau$. Hence, at the end of the iteration, i.e., after the update $\widehat{S}_i \gets \widehat{S}_i \triangle X$ (in Line \ref{line:i-bulk-update}), we have $p^+_i(\widehat{S}_i) \geq \tau$, as stated. 

It remains to show that, even after the updates in Lines \ref{line:update-a} and \ref{line:i-bulk-update}, the allocation $\widehat{\calS} \in \calM^*$. Note that the algorithm maintains the invariant that each item $t$ is assigned to an agent in $\Gamma(t)$. That is, the containments $L_j \subseteq \widehat{S}_j \subseteq L_j \cup \Gamma(j)$ continue to hold and, hence, the updated allocation satisfies $\widehat{\calS} \in \calO^*$. Here, for agent $i$, property (ii) of subset $X$ ensures that for the updated bundle $\widehat{S}_i$ we have $p^*(\widehat{S}_i) \leq \tau$. In addition, property (iii) ensures that for agents $a$ considered in the iteration $p^*(\widehat{S}_a)$ in fact decreases after the update in Line \ref{line:update-a}: $p^*(\widehat{S}_i) + p^*(\widehat{S}_a) =   p^*(\widehat{S}_i \triangle \{t\}) + p^*(\widehat{S}_a \triangle \{t \})$ and $p^*(\widehat{S}_i \triangle \{t\}) > p^*(\widehat{S}_i)$. Therefore, for the allocation $\widehat{\calS}$ maintained at the end of iteration, we have $p^*(\widehat{S}_j) \leq \tau$, for each agent $j$, i.e., as desired  $\widehat{\calS} \in \calM^*$. This completes the proof of the lemma. 
\end{proof}

\begin{lemma}
\label{lemma:each-agent-once}
Any agent $i \in [n]$ is included at most once in the set $\calQ$ throughout the execution of \textsc{AugmentTree Algorithm}.
\end{lemma}
\begin{proof}
The root agent $r$ is included in $\calQ$ at initialization. Any other agent $j \in [n]$ is included in $\calQ$ only when its immediate ancestor is being processed in the while-loop iteration. In particular, if $j$ is added to $\calQ$ in some iteration, then the agent $i$ being processed at that time must satisfy $j \in \cho(i)$ (see Lines~\ref{line:update-a} and~\ref{line:add-to-Q}). Since the root $r$ has no ancestors, it is never included in $\calQ$---and, hence, never processed---a second time.

By inductively applying this observation down the levels of the rooted tree, we obtain that no agent is added to $\calQ$ more than once. The lemma stands proved. 
 \end{proof}

\begin{corollary}
\label{corollary:termination}
The \textsc{AugmentTree Algorithm} (\Cref{algo:augmenting-forest}) always terminates and returns a complete allocation $\widehat{\calS} = (\widehat{S}_1, \ldots, \widehat{S}_n) \in \calM^*$.
\end{corollary}
\begin{proof}
Since each iteration of the while-loop executes successfully (\Cref{lemma:loop-succeeds}) and the while-loop can iterate at most $n$ times (\Cref{lemma:each-agent-once}), it follows that the algorithm necessarily terminates. Furthermore, given that the input allocation is contained in $\calM^*$ and the algorithm maintains this containment, $\widehat{\calS} \in \calM^*$ (\Cref{lemma:loop-succeeds}), we obtain that the returned allocation belongs to $\calM^*$ as well. This completes the proof of the corollary. 
\end{proof}

We now establish \Cref{theorem:leveled-allocation-is-fair-n-efficient}.

\begin{proof}[Proof of \Cref{theorem:leveled-allocation-is-fair-n-efficient}]
As mentioned previously, to prove the theorem it suffices to show that, for any input allocation $\calS \in \calM^*$, \textsc{AugmentTree Algorithm} returns another allocation $\widehat{\calS} \in \calM^*$ with a strictly higher $\varphi^+(\cdot)$ value. In particular, by repeatedly applying the algorithm, we obtain the desired allocation $\calA^* \in \calM^*$ with $\varphi^+(\calA^*) = n$, thereby establishing the theorem.

\Cref{corollary:termination} ensures that the algorithm successfully finds $\widehat{\calS} \in \calM^*$. We now complete the proof by showing that $\varphi^+(\widehat{\calS}) > \varphi^+(\calS)$.

Note that, at the end of the first iteration, the root agent $r$ receives a bundle with $p^+_r(\cdot)$ value at least $\tau$ (\Cref{lemma:loop-succeeds}). Since $r$’s bundle is not modified thereafter (\Cref{lemma:each-agent-once}), we have $p^+_r(\widehat{S}_r) \geq \tau$ in the returned allocation as well.
 
 Moreover, for any agent $j \in [n]$, if in the input allocation $\calS$ we have $p^+_j(S_j) \geq \tau$, then $p^+_j(\widehat{S}_j) \geq \tau$ in the returned allocation 
 $\widehat{\calS}$ as well: If $j$’s bundle was not updated during the algorithm, the assertion holds trivially, since $\widehat{S}_j = S_j$. Otherwise, if $j$’s bundle was updated and its $p^+_j(\cdot)$ value fell below $\tau$, then, by $j$’s inclusion in $\calQ$ and the subsequent updates, the algorithm ensures that $p^+_j(\widehat{S}_j) \geq \tau$ at termination; recall that $\calQ = \emptyset$ when the algorithm terminates. 
 
This observation and the inclusion of $r$ (as an agent which satisfies $p^+_r(\widehat{S}_r) \geq \tau$) imply that $\varphi^+(\widehat{\calS}) > \varphi^+(\calS)$. The theorem stands proved. 
\end{proof}

In summary, through the chain of implications illustrated in \Cref{figure:theorems}, we establish the existence of fair and efficient allocations for indivisible mixed manna under additive valuations.
\begin{figure}[h]
\begin{center}
\includegraphics[scale=0.72]{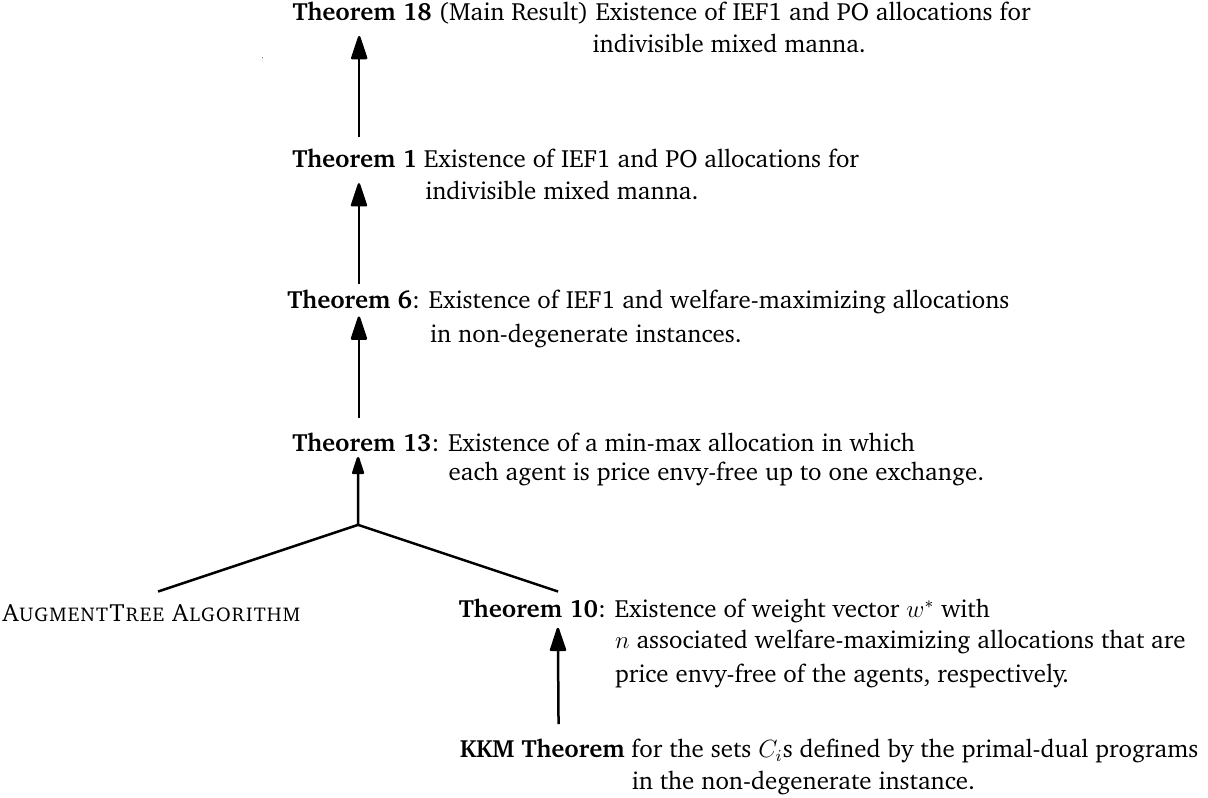}
	\end{center}
	\caption{Chain of implications established in Sections  \ref{section:phaseI} to \ref{section:addingFtoPO}.}
\label{figure:theorems}
\end{figure}

%% file: 6.IEFo-fPO.tex
\section{Main Result: $\EFEo$ and $\fPO$ for Mixed Manna}\label{section:addingFtoPO}
This section strengthens \Cref{theorem:main_fair_efficient} on the efficiency front. In particular, we show that, for indivisible mixed manna with additive valuations, there always exists an (integral) allocation that is both $\EFEo$ and fractionally Pareto efficient ($\fPO$). The proof for this result relies on LP duality and limit arguments.

\begin{theorem}\label{theorem:addingFtoPO}
Every fair division instance $\calI = \angleb{[n],[m],\{v_i\}_{i=1}^n}$, with indivisible mixed manna and additive valuations, admits an $\IEFo$ and $\fPO$ allocation. 
\end{theorem}
\begin{proof}
For the given instance $\calI$, consider the non-degenerate instance $\widehat{\calI}$ as constructed in \Cref{subsection:nonDegInstanceDef}. Recall that the construction of $\widehat{\calI}$ is applicable as long as the chosen positive parameter $\varepsilon > 0$ satisfies the strict inequality in  Equation (\ref{definition:varepsilon}). We will denote $\widehat{\calI}$ as $\widehat{\calI}_1$. 

Furthermore, we consider an infinite sequence of instances $\widehat{\calI}_{1}, \widehat{\calI}_{2}, \ldots, \widehat{\calI}_{k}, \ldots$ such that the instance $\widehat{\calI}_{k}$ is defined with parameter value equal to $\varepsilon/k$.  \Cref{theorem:nondegenerate-fair-efficient} ensures that, for each instance $\widehat{\calI}_{k}$, there exists a weight vector $w_k \in \Delta_{n-1}$ and an associated $\IEFo$ allocation $\alloc_k \in \OPT(\textbf{\textrm{LP}}(w_k,\eta, \widehat{\calI}_{k}))$; for clarity, in the proof here we will explicitly specify the instance $\widehat{\calI}_k$ in the LP. Since $\alloc_k$ lies in $\OPT(\textbf{\textrm{LP}}(w_k,\eta, \widehat{\calI}_{k}))$---i.e., it maximizes weighted social welfare among all fractional allocations---$\alloc_k$ is $\fPO$ in $\widehat{\calI}_{k}$. 

The claim below identifies a useful subsequence of the instances $\left\{ \widehat{\calI}_{k} \right\}_{k \in \mathbb{N}}$ and the corresponding weight vectors $\left\{ w_k \in \Delta_{n-1} \right\}_{k\in \mathbb{N}}$.    
   
\begin{claim}\label{claim:POtoFPOProof}
There exists a subsequence $\widehat{\calI}_{{\ell_1}}, \widehat{\calI}_{{\ell_2}}, \widehat{\calI}_{{\ell_3}}, \ldots$ of the infinite sequence $\left\{ \widehat{\calI}_{k} \right\}_{k \in \mathbb{N}}$ and an allocation $\alloc^* \in \Pi_n([\mbar])$ with the properties that 
 \begin{enumerate}
      \item[(1)] For all $k \in \mathbb{N}$, allocation $\alloc^*$ is $\EFEo$ in $\widehat{\calI}_{\ell_k}$ and $\alloc^* \in \OPT(\textbf{\textrm{LP}}(w_{\ell_k},\eta, \widehat{\calI}_{\ell_k}))$, where $w_{\ell_k}$ is the weight vector for $\widehat{\calI}_{\ell_k}$. 
      \item[(2)] The sequence $\left\{w_{\ell_k} \right\}_{k \in \mathbb{N}}$ converges to some $w^* \in \Delta_{n-1}$, and
      \item[(3)] The sequence of optimal price vectors $\left\{ p^{\ell_k} = \OPT({\textbf{\textrm{Dual-LP}}}(w_{{\ell_k}},\eta,\widehat{\calI}_{{{\ell}_k}})) \right\}_{k \in \mathbb{N}}$ also converges to a point $p^* \in \mathbb{R}^{\mbar}$.
   \end{enumerate}
   \end{claim}
   \begin{proof}
Consider the infinite sequence of instances $\left\{ \widehat{\calI}_{k} \right\}_{k \in \mathbb{N}}$ and the corresponding allocations $\left\{ \alloc_k \right\}_{k \in \mathbb{N}}$ that satisfy $\alloc_k \in \OPT(\textbf{\textrm{LP}}(w_k,\eta, \widehat{\calI}_{k}))$. Here, $\alloc_k$ is $\EFEo$ in $\widehat{\calI}_{k}$, for each $k$. Since the number of allocations is finite, there exists an allocation $\alloc^* \in \Pi_n([\mbar])$ and an infinite subsequence $\left\{ \widehat{\calI}_{a_k} \right\}_{k \in \mathbb{N}}$ such that  $\alloc_{a_k}  = \alloc^*$ for each $k$.  

Now, for the subsequence $\left\{ \widehat{\calI}_{a_k} \right\}_{k \in \mathbb{N}}$, consider the associated weights $\left\{ w_{a_k} \right\}_{k \in \mathbb{N}}$. Since these weights lie in the compact set $\Delta_{n-1}$, we obtain, via the Bolzano-Weierstrass theorem, a further subsequence $\{ w_{\ell_k} \}_{k \in \mathbb{N}}$ that converges to a point 
$w^* \in \Delta_{n-1}$; here, the weight vector $w_{\ell_k}$ is associated with instance $\widehat{\calI}_{\ell_k}$. This establishes property (2) in the lemma statement.  

Note that, for each $k \in \mathbb{N}$, we have $\alloc^* \in \OPT(\textbf{\textrm{LP}}(w_{\ell_k},\eta, \widehat{\calI}_{\ell_k}))$ and $\alloc^*$ is $\EFEo$ in $\widehat{\calI}_{\ell_k}$. This gives us property (1).

Finally, for property (3), we use the fact that the optimal prices $p^{\ell_k} = \OPT({\textbf{\textrm{Dual-LP}}}(w_{{\ell_k}},\eta,\widehat{\calI}_{{{\ell}_k}}))$ satisfy $p^{\ell_k}_t = \max\limits_{i\in [n]} \ \big( (w_{\ell_k})_i+ \eta \big) \vbar^{\ell_k}_i(t)$ for each item  $t \in [\mbar]$; here $\vbar^{\ell_k}_i(t)$ is the value of $t$ for agent $i$ in $\widehat{\calI}_{{{\ell}_k}}$. As $k$ tends to infinity, $w_{\ell_k}$ tends to $w^*$ and $\vbar^{\ell_k}_i(t)$ tends to $v_i(t)$, for each $i$ and $t \in [m]$. Also, $\vbar^{\ell_k}_i(m+1) = \lambda/2$ for each $i \in [n]$. Since the prices are continuous functions of the converging variables, by the sequential criterion of continuity, we obtain that $\lim_{k \to \infty} p^{\ell_k}_t = \max_{i \in [n]} \ (w^*_i + \eta)v_i(t)$, for each $t \in [m]$, and $\lim_{k \to \infty} p^{\ell_k}_{\mbar} = \max_{i \in [n]} \ (w^*_i + \eta)\lambda/2$, implying that $\left\{ p^{\ell_k} \right\}_{k \in \mathbb{N}}$ converges. Write the limit point as $p^* \in \mathbb{R}^{\mbar}$. This establishes property (3) and completes the proof of the claim.
\end{proof}

Using \Cref{claim:POtoFPOProof} and for the given instance $\calI$, we will next prove that $(A^*_1 \setminus \{m+1\}, \ldots, A^*_n \setminus \{m+1 \}) \in \OPT(\textbf{\textrm{LP}}(w^*,\eta, \calI))$ and $(A^*_1 \setminus \{m+1\}, \ldots, A^*_n \setminus \{m+1 \})$ is $\EFEo$ in $\calI$. This will imply that $\alloc^*$ is $\fPO$ and $\EFEo$ in the given instance $\calI$, completing the proof of the theorem. 

\Cref{claim:POtoFPOProof} gives us that $\alloc^*$ is $\EFEo$ in $\widehat{\calI}_{\ell_k}$, for each $k \in \mathbb{N}$. Using this observation (for any $k$) and \Cref{lemma:EFE-non-degenerate}, we obtain that $(A^*_1 \setminus \{m+1\}, \ldots, A^*_n \setminus \{m+1 \})$ is $\EFEo$ in the instance $\calI$ as well.

It remains to prove that $(A^*_1 \setminus \{m+1\}, \ldots, A^*_n \setminus \{m+1 \}) \in \OPT(\textbf{\textrm{LP}}(w^*,\eta, \calI))$. Towards this, consider the instance $\calI'$ obtained by including item $(m+1)$ into $\calI$ and keeping the values, $v_i(t)$s, of all the other items unchanged. We first show that the tuple $(\alloc^*, p^*)$ satisfies the complementary slackness condition for the primal-dual linear programs $\textbf{\textrm{LP}}(w^*,\eta, \calI')$ and $\textbf{\textrm{Dual-LP}}(w^*,\eta, \calI')$. Together with the fact that $\alloc^*$ and $p^*$ are feasible for the linear programs $\textbf{\textrm{LP}}(w^*,\eta, \calI')$ and $\textbf{\textrm{Dual-LP}}(w^*,\eta, \calI')$, respectively, we obtain that $\alloc^* \in \OPT(\textbf{\textrm{LP}}(w^*,\eta, \calI'))$. 

To prove that complementary slackness holds, consider any agent $i \in [n]$ and item $t \in [m]$. If we have $t \notin A^*_i$, then complementary slackness holds trivially. Otherwise, if $t \in A^*_i$, then using the fact that $(\alloc^*,p^{\ell_k})$ are an optimal primal-dual pair (\Cref{claim:POtoFPOProof}), we have $p^{\ell_k}_t = \big((w_{\ell_k})_i+\eta\big) \vbar^{\ell_k}_i(t)$, where $\vbar^{\ell_k}_i(t)$ is the value of item $t$ for agent $i$ in the instance $\widehat{\calI}_{\ell_k}$. Taking the limit on both sides, we get 
 $$p^*_t = \lim_{k \to \infty} p^{\ell_k}_t = \lim_{k \to \infty} \big((w_{\ell_k})_i+\eta\big) \vbar^{\ell_k}_i(t) = (w^*_i + \eta) v_i(t).$$
A similar argument shows that the complementary slackness holds for the item $(m+1)$ as well. Hence, $(\alloc^*,p^*)$ satisfies complementary slackness, implying that 
$\alloc^* \in \OPT(\textbf{\textrm{LP}}(w^*,\eta, \calI'))$. 

Finally, note that, for all agents $i \in [n]$ and all items $t \in [m]$, the valuations $v_i(t)$s are the same between $\calI$ and $\calI'$. In addition, the optimality of $\alloc^*$ in $\calI'$ implies that each item $t \in [m]$ is assigned to an agent in $\argmax_{j \in [n]} \ (w^*_j + \eta) v_j(t)$. This per-item guarantee implies that  $(A^*_1 \setminus \{m+1\}, \ldots, A^*_n \setminus \{m+1 \}) \in \OPT(\textbf{\textrm{LP}}(w^*,\eta, \calI))$. Hence, $(A^*_1 \setminus \{m+1\}, \ldots, A^*_n \setminus \{m+1 \})$ is $\fPO$ in $\calI$. 

Overall, we have that $(A^*_1 \setminus \{m+1\}, \ldots, A^*_n \setminus \{m+1 \})$ is both $\EFEo$ and $\fPO$ in the given instance $\calI$. The theorem stands proved. 
\end{proof}

%% file: 7.EF-PO-div.tex
\section{$\EF$ and $\fPO$ Allocations for Divisible Mixed Manna}
 This section shows that our result for indivisible mixed manna (specifically, \Cref{theorem:addingFtoPO}) implies the envy-free ($\EF$) and $\fPO$ guarantee for divisible mixed manna under additive valuations. As mentioned previously, this result for the divisible case provides an alternative proof of the existential guarantee from \cite{Bogo2017}. 

 To prove the fairness and efficiency guarantee for divisible mixed manna, we will use the following lemma, which shows that the set of $\fPO$ allocations form a closed set.\footnote{Although this lemma seems standard, we could not find a reference, and hence, provide a proof here for completeness.}

\begin{lemma}\label{lemma:fPOIsClosed}
   For any fair division instance $\calI = \angleb{[n], [m], \{v_i\}_{i=1}^n}$, the set of all $\fPO$ fractional allocations, $\mathscr{F} \coloneqq  \left\{\calX \in [0,1]^{n \times m} \mid \calX \text{ is } \allowbreak \fPO \text{ in } \calI \right\}$, is a closed set.
\end{lemma}
\begin{proof}
Write $\calF \subset \mathbb{R}^{nm}$ to denote the set of all fractional allocations in the given instance $\calI$. Note that $\calF$ is a polytope and has finitely many faces.   

Moreover, any fractional allocation $\calX \in \calF$ is $\fPO$ iff there exists a weight vector $\lambda \in \mathbb{R}^n_{>0}$ (i.e., a vector with $\lambda_i >0$ for each $i$) such that $\calX$ maximizes the weighted social welfare with respect to $\lambda$; see \cite[Lemma 2.3]{SS22Sharing} and \cite[Proposition 16.E.2]{mas1995microeconomic}. That is, $\calX \in \argmax_{ \calY \in \calF} \ \sum_{t=1}^m \sum_{i=1}^n \lambda_i v_i(t) Y_{i,t}$.  

Since the weighted welfare maximization problem is a linear program, for any weight vector $\lambda \in \mathbb{R}^n_{>0}$, the set of welfare-maximizing allocations is a face of the polytope $\calF$. Write $\calF_\lambda$ to denote this face, $\calF_\lambda = \argmax_{ \calY \in \calF} \ \sum_{t=1}^m \sum_{i=1}^n \lambda_i v_i(t) Y_{i,t}$. 

The above-mentioned characterization of $\fPO$ implies that the set of all $\fPO$ fractional allocations $\mathscr{F} = \bigcup\limits_{\lambda \in \mathbb{R}^n_{>0}} \ \calF_\lambda$. Note that this union is composed of finitely many faces; the polytope $\calF$ itself has a finite number of faces. Since every face of $\calF$ is closed, we get that $\mathscr{F}$ is a finite union of closed sets, and hence, is closed as well. The lemma stands proved. 
\end{proof}

The following theorem is the main result of this section. We employ our result for indivisible items, along with limit arguments, to establish this theorem for divisible mixed manna.

\begin{theorem}\label{theomem:divisibleEFnPO}
Every fair division instance $\calD = \angleb{[n], [m], \{v_i\}_{i=1}^n}$ with divisible mixed manna and additive valuations admits an $\EF$ and $\fPO$ allocation $\calX^* \in [0,1]^{n\times m}$.
\end{theorem}
\begin{proof}
Consider the following infinite sequence of instances with indivisible items: For each $k \in \mathbb{N}$, write $\calI_k = \angleb{[n],[mk],\{v^k_i\}_{i=1}^n}$ to denote an instance with $mk$ indivisible items and agents' additive valuations $v^k_i(\cdot)$. For each (divisible) item $t \in [m]$ in the given instance $\calD$, we include $k$ indivisible items in $\calI_k$. In particular, for each $t \in [m]$, the $k$ indivisible items $(t-1)k+1, (t-1)k+2, \ldots, tk$ in $\calI_k$ are obtained by dividing the $t$th item in $\calD$ equally. Here, for each agent $i \in [n]$, in instance $\calI_k$ the values $v^k_i((t-1)k+1) = v^k_i((t-1)k+2) = \ldots = v^k_i(tk) = v_i(t)/k$. 

The following claim establishes useful connections between the instances $\calI_k$ and $\calD$.

   \begin{claim}\label{claim:surjectiveMapping}
      Let $\calF \coloneqq \{\calX \in [0,1]^{n \times m} \mid \sum_{i=1}^n X_{i,t} = 1 \text{ for all } t \in [m]\}$ be the set of fractional allocations in the given instance $\calD$ and $\calF_k = \{\calX \in [0,1]^{n \times mk} \mid \sum_{i=1}^n X_{i,t} = 1 \text{ for all } t \in [mk]\}$ be the set of fractional allocations in instance $\calI_k$. Then, there exists a surjective (onto) map $g : \calF_k \mapsto \calF$ with the property that, for each fractional allocation $\calY \in \calF_k$, every agent receives the same value in $g(\calY) \in \calF$ as in $\calY$, i.e., $v_i \left(g(\calY)_i \right) = v^{k}_i \left(Y_i \right)$ for each $i \in [n]$.  
   \end{claim}
   \begin{proof}
We define $g : \calF_k \mapsto \calF$ as follows: for any fractional allocation $\calY \in \calF_k$, set $g(\calY)_{i,t} = \frac{1}{k} \sum_{\ell = 1}^k Y_{i,(t-1)k + \ell}$, for each $i \in [n]$ and $t \in [m]$. 

First, to establish that $g$ is surjective, fix any $\calZ \in \calF$. We will show that there exists an allocation $\calX \in \calF_k$ such that $\calZ = g(\calX)$. Specifically, set $X_{i,(t-1)k + \ell} \coloneqq Z_{i,t}$, for each $i \in [n]$, $t\in [m]$, and $\ell \in [k]$. Indeed, with this definition, we have $\calX \in \calF_k$. Furthermore, $g(\calX)_{i,t} = \frac{1}{k} \sum_{\ell = 1}^k X_{i,(t-1)k + \ell} = \frac{1}{k} \sum_{\ell = 1}^k Z_{i,t} = Z_{i,t}$ for each $i \in [n]$ and $t \in [m]$. Hence, we have $\calZ = g(\calX)$, i.e., $g$ is surjective.

We next complete the proof of the claim by proving that $g$ preserves agents' values. For each  $\calY \in \calF_k$ and agent $i \in [n]$ it holds that  
   \begin{align*}
      v_i(g(\calY)_i) & = \sum_{t=1}^m v_i(t)  \  g(\calY)_{i,t} = \sum_{t=1}^m \sum_{\ell = 1}^{k} \frac{v_i(t)}{k}   Y_{i,(t-1)k + \ell} \\ & = \sum_{t=1}^m \sum_{\ell = 1}^{k}  v^{k}_i((t-1)k+\ell)  \  Y_{i,(t-1)k + \ell} \tag{$v^{k}_i((t-1)k+\ell) = \frac{v_i(t)}{k}$}\\
      & = v_i^{k}(Y_i).
   \end{align*}
   The claim stands proved. 
   \end{proof}
   For each instance $\calI_k$ in the infinite sequence $\calI_1, \calI_2, \calI_3, \ldots$, we can apply~\Cref{theorem:addingFtoPO} to infer the existence of an allocation $\alloc^k$ that is $\fPO$ and $\IEFo$ in $\calI_k$. For each $k \in \mathbb{N}$, write $\calQ^k \in [0,1]^{n \times mk}$ to denote the fractional allocation representation of $\alloc^k$, i.e., $Q^k_{i,t} = 1$ if $t \in A^k_i$ and, otherwise, $Q^k_{i,t} = 0$. From the sequence of allocations $\left\{ \calA^k \right\}_{k \in \mathbb{N}}$---equivalently, $\left\{ \calQ^k \right\}_{k \in \mathbb{N}}$ ---we construct the sequence of fractional allocations $\left\{ \calX^k \right\}_{k \in \mathbb{N}}$ in the given instance $\calD$ by setting $\calX^k \coloneqq g(\calQ^k)$ for each $k$; here $g$ is the map obtained in \Cref{claim:surjectiveMapping}. 
   
Since the set of fractional allocations in $\calD$ is closed and bounded, the Bolzano-Weierstrass theorem implies that the sequence $\left\{ \calX^k \right\}_{k \in \mathbb{N}}$ has a convergent subsequence $\calX^{a_1}, \calX^{a_2}, \calX^{a_3},\ldots$. Let $\calX^*$ be the limit point of this subsequence. 
     
We will show that the fractional allocation $\calX^* \in [0,1]^{n \times m}$ is $\EF$ and $\fPO$ in the instance $\calD$. The existence of such a fair and efficient allocation establishes the theorem. 

We first prove that $\calX^*$ is $\EF$ in $\calD$. Here, we upper bound the maximum envy in the fractional allocations $\calX^{a_k} \in [0,1]^{n \times m}$ that form the converging subsequence. Recall that $\calX^{a_k}$ is obtained by applying map $g(\cdot)$ to the (fractional representation of) the allocation $\alloc^{a_k}$ that is $\EFEo$ allocation in the instance $\calI_{a_k}$ (\Cref{theorem:addingFtoPO}). Hence, in instance $\calI_{a_k}$, for each agent $i \in [n]$,
there exists a subset $S \subseteq [m \ a_k]$ of size at most one ($|S| \leq 1$) such that $v_i^{a_k}(A^{a_k}_i \triangle S) \geq \max_{j\in [n]} v_i^{a_k}(A^{a_k}_j)$. 
That is, $v_i^{a_k}(A^{a_k}_i ) + \max_{s \in [m a_k]} \ |v_i^{a_k}(s)|  \geq \max_{j\in [n]} v_i^{a_k}(A^{a_k}_j)$. Therefore, the maximum envy under allocation $\alloc^{a_k}$ in instance $\calI_{a_k}$ is upper bounded as follows 
\begin{align}
\max_{i, j \in [n]} \left( v_i^{a_k}(A^{a_k}_j) - v_i^{a_k}(A^{a_k}_i) \right) \leq \max_{i \in [n], \ s \in [m a_k]} |v_i^{a_k}(s)| = \max_{i \in [n], \ t \in [m]} \frac{1}{a_k} |v_i(t)| \leq \max_{i \in [n], \ t \in [m]} \frac{1}{k} |v_i(t)| \label{ineq:max-envy}
\end{align}
Fractional allocation $\calX^{a_k}$ is obtained by applying map $g(\cdot)$ to $\alloc^{a_k}$. Hence, using \Cref{claim:surjectiveMapping}, we extend inequality (\ref{ineq:max-envy}) to upper bound the maximum envy in $\calX^{a_k}$:
\begin{align}
\max_{i, j \in [n]} \left( v_i(X^{a_k}_j) - v_i(X^{a_k}_i) \right) \leq \frac{1}{k} \left( \max_{i \in [n], \ t \in [m]} |v_i(t)| \right)\label{ineq:max-envy-i} 
\end{align}
Note that the maximum envy is a continuous function of the fractional allocations $\calX \in [0,1]^{n \times m}$. In addition, the right-hand-side of Equation (\ref{ineq:max-envy-i}) tends to zero as $k \to \infty$. Hence, via the Order Limit Theorem, we obtain that the maximum envy in $\calX^*$ is at most zero: 
$\max_{i, j \in [n]} \left( v_i(X^*_j) - v_i(X^{*}_i) \right) \leq 0$. Therefore, as stated, $\calX^*$ is $\EF$ in the given instance $\calD$.

To prove that $\calX^*$ is $\fPO$ in the instance $\calD$, we will show that all the fractional allocations $\calX^{a_k} \in [0,1]^{n \times m}$ in the converging subsequence are $\fPO$ in $\calD$. Since the set of $\fPO$ allocations in $\calD$ is a closed set (\Cref{lemma:fPOIsClosed}), we get that the limit point of the subsequence, $\calX^*$, is also $\fPO$. 

Recall that, for each $k \in \mathbb{N}$, the fractional allocation $\calX^{a_k}$ is obtained by applying the map $g(\cdot)$ to the (fractional representation of) the allocation $\alloc^{a_k}$ that is an $\fPO$ allocation in the instance $\calI_{a_k}$ (\Cref{theorem:addingFtoPO}). Now, consider any fractional allocation $\calZ$ in $\calD$ -- we will show that $\calZ$ cannot Pareto dominate $\calX^{a_k}$. Given that the map $g$ is surjective (\Cref{claim:surjectiveMapping}), there exists a fractional allocation $\calY$ in $\calI_{a_k}$ such that $\calZ = g(\calY)$. Since $g$ preserves the valuations of all the $n$ agents (\Cref{claim:surjectiveMapping}) and $\calY$ does not Pareto dominate $\alloc^{a_k}$, we have that $\calZ$ does not Pareto dominate $\calX^{a_k}$. Hence, the fractional allocations $\calX^{a_k}$, along with the limit $\calX^*$, are $\fPO$. 

Overall, we obtain the existence of an $\EF$ and $\fPO$ allocation of the divisible mixed manna in the given instance $\calD$. This completes the proof of the theorem. 
\end{proof}

%% file: 8.appendix.tex
\section{$\XEFo$ and $\fPO$ for Indivisible Mixed Manna}\label{section:argument-works-for-goods}
This section provides a sketch of how the proof of \Cref{theorem:main_fair_efficient} can be modified to establish the existence of $\PO$ and $\EFo$ allocations of indivisible goods. In particular, we will show the existence of mixed-manna allocations that are $\PO$, and \emph{extrospectively} envy-free up to one item ($\XEFo$); this fairness notion is defined next. 

\begin{definition}\label{definition:XEFo}
In an instance $\calI = \angleb{[n],[m], \{v_i\}_{i=1}^n}$, an allocation $\alloc = (A_1, \ldots, A_n)$ is said to be extrospectively envy-free up to one item ($\XEFo$) if, for all agents $i,j \in [n]$, there exists a subset $S \subseteq [m]$ of size at most one ($|S| \leq 1$) such that $v_i(A_i) \geq v_i(A_j \triangle S)$.
\end{definition}

Our goal will be to prove the following theorem.

\begin{restatable}{theorem}{extrospecitveEFo}
\label{theorem:OutPO}
Every fair division instance $\calI = \angleb{[n], [m], \{v_i\}_{i=1}^n}$, with indivisible mixed manna and additive valuations, admits an allocation that is both extrospectively envy-free up to one item ($\XEFo$) and Pareto efficient ($\PO$).
\end{restatable}

To avoid repeating the arguments used in the proof of \Cref{theorem:main_fair_efficient}, we will primarily focus on the modifications required to establish the existence of $\PO$ and $\XEFo$ allocations.

To prove this existential result, some changes are required in the construction of the non-degenerate instance $\widehat{\calI} = \angleb{[n], [\mbar], \{\vbar_i\}_{i=1}^m}$ (see \Cref{subsection:nonDegInstanceDef} for comparison) and in the primal and dual programs (see \Cref{subsection:instace-reduction}). We begin by describing the changes in $\barI$; with a slight abuse of notation we continue to denote the non-degenerate instance by $\barI$. We set $\mbar \coloneqq m + n$, i.e., we add $n$ auxiliary items to the given instance $\calI$. For each $i \in [n]$, the valuations of the new item, $(m + i)$, is set as follows: $\vbar_i(m+i) = \lambda /2$ and $\vbar_j(m+i) = 0$, for all $j \in [n] \setminus \{i\}$. For the given items $t \in [m]$ and any agent $i\in [n]$, if $v_i(t) = 0$, we set $\vbar_i(t) = 0$, i.e., leave the zero values unchanged. Otherwise, if $v_i(t) \neq 0$, then we set $\vbar_i(t) = v_i(t) - \overline{\varepsilon}_{i,t}$. The perturbations $\overline{\varepsilon}_{i,t}$ are  drawn independently from the uniform distribution $\mathrm{Uniform}[0, \overline{\varepsilon}]$, where $\overline{\varepsilon} >0$ is fixed to satisfy

\begin{equation}\label{definition:varepsilon-updated}
\overline{\varepsilon} < \frac{\lambda \omega}{16 m n^2(m+n) \ \max\limits_{i \in [n], t \in [m]} |v_i(t)|}
\end{equation}
Here, $\lambda$ and $\omega$ are defined as in \Cref{subsection:nonDegInstanceDef}. We will now describe the modified linear programs, denoted by $\overline{\textbf{\textrm{LP}}}$ and $\overline{\textbf{\textrm{Dual-LP}}}$. Their construction will utilize the following constant
\begin{equation}\label{definition:value-of-eta-updated}
\overline{\eta} \coloneqq \frac{\lambda}{2n(m+n) \max\limits_{i\in [n], \ t \in [\mbar]}|\vbar_i(t)|}
\end{equation}
The constants $\overline{\varepsilon}$ and $\overline{\eta}$ satisfy the following analogue of \Cref{ineq:to-contradict}

\begin{equation}\label{equation:relationship-between-constants}
\overline{\varepsilon} < \frac{\overline{\eta} \omega}{4mn} 
\end{equation}

\[
\begin{array}{rlll}
\overline{\textbf{\textrm{LP}}}(w,\overline{\eta}): & & \qquad \overline{\textbf{\textrm{Dual-LP}}}(w,\overline{\eta}): \\  \ \ 
\max & \displaystyle \sum_{i \in [n]}\sum_{t \in [\mbar]} \frac{1}{(w_i + \overline{\eta})} \vbar_{i}(t) x_{i,t} & \ \ \ \qquad \min \displaystyle \sum_{t \in [\mbar]} p_t \\  
\text{s.t.} & \displaystyle \sum_{i \in [n]} x_{i, t} = 1 \ \text{ for all } t \in [\mbar] & \ \ \ \qquad \text{s.t.} \ \ \ p_t \geq \frac{1}{w_i + \overline{\eta} }\vbar_{i}(t) \text{ for all } i \in [n] \text{ and } t \in [\mbar]. \\ 
& x_{i,t} \ge 0 \ \ \text{for all } i \in [n] \text{ and } t \in [\mbar]. & &
\end{array}
\]

The following lemma is a counterpart of \Cref{lemma:EFE-non-degenerate}; its proof follows similar arguments and is omitted. 

\begin{restatable}{lemma}{nonDegEFE2}\label{lemma:EFE-non-degenerate-V2}
If an allocation $\alloc = (A_1,\ldots, A_n)$ is $\XEFo$ in the non-degenerate instance $\barI = \angleb{[n],[\mbar], \{\vbar_i\}_{i=1}^n}$, then the allocation $(A_1 \setminus \{(m+i)\}_{i=1}^n, A_2 \setminus \{(m+i)\}_{i=1}^n, \ldots, A_n \setminus \{(m+i)\}_{i=1}^n)$ is $\XEFo$ in the given instance $\calI = \angleb{[n],[m], \{v_i\}_{i=1}^n}$.
\end{restatable}

We will now state an analogue of \Cref{lemma:PO-non-degenerate}. While the proof here follows an approach similar to that of \Cref{lemma:PO-non-degenerate}, key technical aspects differ and, hence, we include the proof below.

\begin{restatable}{lemma}{nonDegPO2}\label{lemma:PO-non-degenerate-V2}
If an allocation $\alloc = (A_1,\ldots, A_n) \in \OPT(\overline{\textbf{\textrm{LP}}}(w,\overline{\eta}))$, for any $w \in \Delta_{n-1}$, then the allocation $(A_1 \setminus \{(m+i)\}_{i=1}^n, A_2 \setminus \{(m+i)\}_{i=1}^n, \ldots, A_n \setminus \{(m+i)\}_{i=1}^n)$ is $\PO$ in the given instance $\calI = \angleb{[n],[m], \{v_i\}_{i=1}^n}$.
\end{restatable}
\begin{proof}
Define an intermediate instance $\calI' = \angleb{[n], [m+n], \{v'_i\}_{i=1}^n}$ wherein $v'_i(t) = v_i(t)$, for all agents $i \in [n]$ and given items $t \in [m]$. In addition, set $v'_i(m+t) = \vbar_i(m+t)$ for each $i, t \in [n]$. Note that instance $\calI'$ is constructed by adding the auxiliary items $m+1, \ldots, m+n$ to $\calI$ and keeping the valuations of the original $m$ items unchanged. Instance $\calI'$ is defined for the purposes of analysis: We will first show that $\alloc$ is $\PO$ in $\calI'$. Then, we will conclude the proof by showing that $(A_1 \setminus \{(m+i)\}_{i=1}^n,\ldots, A_n \setminus \{(m+i)\}_{i=1}^n)$ is $\PO$ in the given instance $\calI$.

Towards a contradiction, we assume that there exists an allocation $\allocb = (B_1,\ldots, B_n)$ that Pareto dominates $\alloc$ in $\calI'$. That is, $v'_i(B_i) \geq v'_i(A_i)$, for all $i \in [n]$, and $v'_a(B_a) > v'_a(A_a)$ for some $a \in [n]$. Without loss of generality, we also assume that $\allocb$ is Pareto efficient in $\calI'$; this ensures that, for each $i \in [n]$, the auxiliary item $(m + i)$ is allocated to agent $i$ in $\allocb$. Additionally, note that each auxiliary item $(m+i)$ also gets allocated to agent $i$ in $\alloc$ since $\alloc \in \OPT(\overline{\textbf{\textrm{LP}}}(w,\overline{\eta}))$. These observations imply that

\begin{align*}
0 & < \sum_{i=1}^n v'_i(B_i) - \sum_{i=1}^n v'_i(A_i)\\
& =  n  \frac{\lambda}{2} + \sum_{i=1}^n v'_i(B_i \setminus \{(m+i)\}_{i=1}^n) -  n  \frac{\lambda}{2} - \sum_{i=1}^n v'_i(A_i\setminus \{(m+i)\}_{i=1}^n) \tag{$v'_i(m+i) = \lambda/2$ for all $i$}\\
& = \sum_{i=1}^n v_i(B_i \setminus \{(m+i)\}_{i=1}^n) - \sum_{i=1}^n v_i(A_i\setminus \{(m+i)\}_{i=1}^n) \label{equation:nonDegPOProof1-updated}\numberthis
\end{align*}
That is, $\sum_{i=1}^n v_i(B_i \setminus \{(m+i)\}_{i=1}^n) - \sum_{i=1}^n v_i(A_i\setminus \{(m+i)\}_{i=1}^n) > 0$. Using this strict inequality and the definition of $\omega$, we obtain 
\[\omega \leq \sum_{i=1}^n v_i(B_i \setminus \{m+i\}_{i=1}^n) - \sum_{i=1}^n v_i(A_i\setminus \{m+i\}_{i=1}^n) = \sum_{i=1}^n v'_i(B_i) - \sum_{i=1}^n v'_i(A_i) \label{equation:nonDegPOProof2-updated} \numberthis \]
 
 Next, we use the welfare optimality of $\alloc$ in $\barI$ (and under the weight vector $w$) to arrive at a contradiction. In particular, $\alloc \in \OPT(\overline{\textbf{\textrm{LP}}}(w,\overline{\eta}))$ implies that $\sum_{i=1}^n (w_i + \overline{\eta})^{-1} \vbar_i(B_i) \leq \sum_{i=1}^n (w_i + \overline{\eta})^{-1} \vbar_i(A_i)$. For each subset $S \subseteq [\mbar]$ and each agent $i$, it holds that $v'_i(S) - \overline{\varepsilon} |S| \leq \vbar_i(S) \leq v'_i(S)$. Therefore, we have $\sum_{i=1}^n (w_i + \overline{\eta})^{-1} (v'_i(B_i) - \overline{\varepsilon} |B_i|) \leq \sum_{i=1}^n (w_i + \overline{\eta})^{-1} v'_i(A_i)$. This inequality reduces to 
 \begin{align}
\sum_{i=1}^n \frac{1}{(w_i + \overline{\eta})} (v'_i(B_i) - v'_i(A_i)) & \leq \overline{\varepsilon} \sum_{i=1}^n \frac{1}{(w_i + \overline{\eta})}  |B_i| \nonumber \\
& \leq \overline{\varepsilon} \sum_{i=1}^n \frac{1}{\overline{\eta}}|B_i|   \tag{$w_i \geq 0$} \\ 
& = \frac{\overline{\varepsilon} (m+n)}{\overline{\eta}} \label{ineq:extra-two-updated}
\end{align}
We can lower bound the left-hand-side of inequality (\ref{ineq:extra-two-updated}) as follows 
\begin{align*}
\sum_{i=1}^n \frac{1}{(w_i + \overline{\eta})} (v'_i(B_i) - v'_i(A_i)) & \geq \frac{1}{2} \sum_{i=1}^n  (v'_i(B_i) - v'_i(A_i)) \tag{$v'_i(B_i) - v'_i(A_i) \geq 0$ and $w_i + \overline{\eta} \leq 2$}\\
& \geq \frac{\omega}{2} \tag{via (\ref{equation:nonDegPOProof2-updated})}
\end{align*}
The last inequality and Equation (\ref{ineq:extra-two-updated}) imply $\frac{\omega}{2} \leq \frac{\overline{\varepsilon} (m+n)}{\overline{\eta}}$, or equivalently, $\frac{\overline{\eta} \omega}{2(m+n)} \leq \overline{\varepsilon}$. This bound, however, contradicts the strict inequality (\ref{equation:relationship-between-constants}), which states that $\overline{\varepsilon} < \frac{\overline{\eta} \omega}{4mn}$.\footnote{We have $n,m \geq 2$.} Therefore, by way of contradiction, we obtain that no allocation $\calB$ Pareto dominates $\alloc \in \OPT(\overline{\textbf{\textrm{LP}}}(w,\overline{\eta}))$ in the instance $\calI'$, i.e., $\alloc$ is $\PO$ in $\calI'$. 

Finally, we will use the Pareto optimality of $\alloc$ in $\calI'$ to show that the allocation $(A_1 \setminus \{(m+i)\}_{i=1}^n,\ldots, A_n \setminus \{(m+i)\}_{i=1}^n)$ is $\PO$ in the given instance $\calI = \angleb{[n], [m], \{v_i\}_{i=1}^n}$. Assume, towards a contradiction, that allocation $(D_1, \ldots, D_n)$ Pareto dominates $(A_1 \setminus \{(m+i)\}_{i=1}^n,\ldots, A_n \setminus \{(m+i)\}_{i=1}^n)$ in $\calI$. As observed above, for each $i \in [n]$, the item $(m+i)$ is allocated to agent $i$ in allocation $\alloc$. Hence, our assumption would imply that $(D_1 \cup \{(m+1)\}, \ldots, D_i \cup \{(m+i)\}, \ldots, D_n \cup \{(m+n)\})$ Pareto dominates $(A_1,\ldots, A_i, \ldots, A_n)$ in $\calI'$. Since this contradicts the already established Pareto optimality of $\alloc$ in $\calI'$, we obtain that $(A_1 \setminus \{(m+i)\}_{i=1}^n, \ldots, A_n \setminus \{(m+i)\}_{i=1}^n)$ is $\PO$ in the instance $\calI$. The lemma stands proved.  
\end{proof}

Another key change is how we apply the KKM theorem. Here, we define the sets $\overline{C}_1, \overline{C}_2,\ldots, \overline{C}_n \subseteq \Delta_{n-1}$ considering $\argmin$ over the optimal prices (instead of $\argmax$ as in \Cref{definition:KKM-sets-a}).

\begin{equation}\label{definition:KKM-sets}
  \overline{C}_i \coloneqq \left\{w \in \Delta_{n-1} \mid \text{ there exists } \alloc \in \OPT(\overline{\textbf{\textrm{LP}}}(w,\overline{\eta})) \text{ such that } i \in \argmin_{j \in [n]}{p(A_j)} \text{ for the optimal prices } p \right\}
\end{equation}
Note that each $\overline{C}_i$ is the set of weights $w \in \Delta_{n-1}$ for which there is an optimal allocation $\alloc \in \OPT(\overline{\textbf{\textrm{LP}}}(w,\overline{\eta}))$ wherein $i$'s bundle has the lowest price, under the optimal $p$. This defining condition, that agent $i$ is worst off in terms of optimal bundle prices, is in complete contrast to the one imposed in \Cref{section:setup-KKM} (\Cref{definition:KKM-sets-a}). Nonetheless, we show below that this somewhat counterintuitive condition leads to the desired $\XEFo$ guarantee. 

The following lemma asserts that the sets $\overline{C}_1,\ldots, \overline{C}_n \subseteq \Delta_{n-1}$ are closed; the proof of this lemma is similar to that of \Cref{lemma:KKM-closedness}.

\begin{lemma}[Closedness]\label{lemma:KKM-closedness-V2}
For each $i \in [n]$, the set $\overline{C}_i \subseteq \Delta_{n-1}$ is closed.
\end{lemma}

We will now prove that the sets $\overline{C}_1,\ldots, \overline{C}_n $ satisfy the KKM covering condition; the proof of this lemma is notably different from its counterpart, \Cref{lemma:KKM-covering}. 

\begin{lemma}[KKM Covering]\label{lemma:KKM-covering-V2}
For any $w \in \Delta_{n-1}$ there exists an $i \in \supp(w)$ such that $w \in \overline{C}_i$.
\end{lemma}
\begin{proof}
Fix a weight vector $w \in \Delta_{n-1}$. We consider two cases based on whether $\supp(w) = [n]$ or $\supp(w) \subsetneq [n]$.\\

\noindent
{\it Case {\rm I}: $\supp(w) = [n]$.} Consider any allocation $\alloc \in \OPT(\overline{\textbf{\textrm{LP}}}(w,\overline{\eta}))$ and optimal prices $p = \OPT(\overline{\textbf{\textrm{Dual-LP}}}(w,\overline{\eta}))$. Let $i' \in \argmin_{j \in [n]} p(A_j)$. In the current case $\supp(w) = [n]$ and, hence, $i' \in \supp(w)$. Additionally, we have $w \in \overline{C}_{i'}$, since $i' \in \argmin_{j \in [n]} p(A_j)$. Hence, for vectors $w$ with full support, the KKM covering condition holds. \\

\noindent
{\it Case {\rm II}: $\supp(w) \subsetneq [n]$.} In this case $[n] \setminus \supp(w) \neq \emptyset$. Write optimal prices $p = \OPT(\overline{\textbf{\textrm{Dual-LP}}}(w,\overline{\eta}))$ and consider any allocation $\alloc \in \OPT(\overline{\textbf{\textrm{LP}}}(w,\overline{\eta}))$. We will show that there exists an $i' \in \supp(w)$ with the property that $i' \in \argmin_{j \in [n]} \ p(A_j)$. For such an agent $i'$, we have $w \in C_{i'}$, and, hence, we will obtain KKM covering in the current case as well.

Define $H \coloneqq \{i \in [n] \mid w_i \geq 1/n\}$; $H \neq \emptyset$ since $\sum_{i=1}^n w_i = 1$. Also, the set of agents $[n]$ partitions into three disjoint subsets: $[n] \setminus \supp(w)$, $H$, and $\supp(w) \setminus H$. 
 
 Let $h \in \argmax_{i \in H} w_i$ be an agent with the highest weight value; note that $h \in H \subseteq \supp(w)$. 

To show that there exists an $i' \in \supp(w)$ such that $w \in \overline{C}_{i'}$, we will first show that $\min\limits_{i \in [n] \setminus \supp(w)} p(A_i) \geq p(A_h)$. This will conclude the proof, since we will obtain 
$$\min\limits_{i \in [n] \setminus \supp(w)} p(A_i) \geq p(A_h) \geq \min\limits_{j \in \supp(w)} p(A_j).$$
Here, the final inequality follows from $h \in \supp(w)$. This will in turn imply the existence of an $i' \in \supp(w)$ such that $i' \in \argmin_{j \in [n]} p(A_j)$, i.e., $w \in \overline{C}_{i'}$, completing the proof.

Hence, in the remainder of the proof, we will show that $\min\limits_{i \in [n] \setminus \supp(w)} p(A_i) \geq p(A_h)$. Towards this, we will prove the following two inequalities
\begin{equation}\label{equation:XEFo-covering-1}
\min\limits_{i \in [n] \setminus \supp(w)} p(A_i) \geq \lambda/2\overline{\eta}
\end{equation}
\begin{equation}\label{equation:XEFo-covering-2}
\lambda/2\overline{\eta} \geq p(A_h)
\end{equation}

Together inequalities (\ref{equation:XEFo-covering-1}) and (\ref{equation:XEFo-covering-2}) imply the desired bound $\min\limits_{i \in [n] \setminus \supp(w)} p(A_i) \geq \lambda/2\overline{\eta} \geq p(A_h)$.

\noindent
\emph{Establishing Inequality (\ref{equation:XEFo-covering-1}).} Note that in the instance $\widehat{\calI} = \angleb{[n], [\mbar], \{{\vbar_i}\}_{i=1}^n}$ the set of items $[\mbar]$ can be partitioned into three disjoint sets $G$, $C$, and $G_0$ (\Cref{lemma:three-types}). Consider any agent $i \in [n] \setminus \supp(w)$. By construction, we have $\vbar_i(m+i) = \lambda/2$ and $\vbar_j(m+i) = 0$ for all $j \neq i$. In the allocation $\alloc$, the item $(m+i)$ will be allocated to the agent $i$ since $\alloc$ lies in $\OPT(\overline{\textbf{\textrm{LP}}}(w,\overline{\eta}))$. Also, item $(m+i)$ has price $p_{m+i} = \frac{\lambda}{2} \  \frac{1}{w_i+\overline{\eta}} = \frac{\lambda}{2\overline{\eta}}$. Additionally, we will show that any agent $i\in [n] \setminus \supp(w)$ will not be allocated any chore from $C$ under $\alloc$. Since only chores have negative prices (\Cref{proposition:price-non-negativity}), this will imply that $p(A_i) \geq p_{m+i} = \lambda/2\overline{\eta}$, establishing $(\ref{equation:XEFo-covering-1})$. 

We will now show that no chore $c \in C$ is allocated to the agent $i \in [n] \setminus \supp(w)$. To this end, it suffices to show that $\vbar_h(c)/(w_h + \overline{\eta}) > \vbar_i(c)/\overline{\eta}$, or equivalently $\overline{\eta} \vbar_h(c) > (w_h + \overline{\eta}) \vbar_i(c)$. To establish this, note that  

\begin{align*}
(w_h + \overline{\eta}) |\vbar_i(c)| & >  |\vbar_i(c)|/n \tag{$w_h \geq 1/n$ and $\overline{\eta} > 0$}\\
 & \geq \lambda/n \tag{$|\vbar_i(c)| = |v_i(c) - \varepsilon_{i,c}| \geq \lambda$}\\
 & >  \overline{\eta} \ \max_{a,b} |\widehat{v}_a(b)|  \tag{\Cref{definition:value-of-eta-updated}}\\
 & \geq  \overline{\eta} \ |\widehat{v}_h(c)|.
\end{align*}
Since $\widehat{v}_h(c), \vbar_i(c) < 0$, we get the desired inequality $\overline{\eta} \vbar_h(c) > (w_h + \overline{\eta}) \vbar_i(c)$.\\

\noindent
\emph{Establishing Inequality (\ref{equation:XEFo-covering-2}).} If an item $t \in [\mbar]$ is allocated to agent $h$, then we have $p_t =  \vbar_h(t)/ (w_h + \overline{\eta}) \leq \max_{a,b} |{\vbar}_a(b)|/(w_h + \overline{\eta}) < \max_{a,b} |{\vbar}_a(b)|/w_h \leq n \max_{a,b} |{\vbar}_a(b)|$, where the final inequality follows from $w_h \geq 1/n$. Therefore, $p(A_h) \leq \sum_{t \in A_h} p_t \leq (m+n) \  n\max_{a,b} |{\vbar}_a(b)|$. Finally, using the definition of $\overline{\eta}$ (\Cref{definition:value-of-eta-updated}), we obtain the inequality (\ref{equation:XEFo-covering-2}): $p(A_h) \leq n(m+n)\max_{a,b} |{\vbar}_a(b)| = \lambda/2\overline{\eta}$.

This completes the proof of the lemma. 
\end{proof}

\Cref{lemma:KKM-closedness-V2,lemma:KKM-covering-V2} together give us the following analogue of \Cref{theorem:KKM-application}.

\begin{theorem}\label{theorem:KKM-application-V2}
   There exists a point $\overline{w}^* \in \Delta_{n-1}$ with the property that, for each $i \in [n]$, we have an allocation $\alloc^{(i)} \in \OPT(\overline{\textbf{\textrm{LP}}}(\overline{w}^*,\overline{\eta}))$ wherein $i \in \argmin_{j \in [n]} p^*(A^{(i)}_j)$; here $p^* = \OPT(\overline{\textbf{\textrm{Dual-LP}}}(\overline{w}^*,\overline{\eta}))$.
\end{theorem}

The subsequent argument to prove the existence of $\PO$ and $\XEFo$ allocation also closely follows \Cref{section:phaseII}. Let  $\overline{w}^* \in \Delta_{n-1}$ be the weight vector whose existence is guaranteed in \Cref{theorem:KKM-application-V2}, and let $p^* = \OPT(\overline{\textbf{\textrm{Dual-LP}}}(\overline{w}^*,\overline{\eta}))$ be the associated optimal prices. 

For each agent $i \in [n]$, write $\overline{L}_i$ to denote the set of items that must be allocated to agent $i$ in any allocation in $\OPT(\overline{\textbf{\textrm{LP}}}(\overline{w}^*,\overline{\eta}))$, i.e., $\overline{L}_i \coloneqq \left\{ s \in [\mbar] \mid \frac{\vbar_i(s)}{(\overline{w}^*_i + \overline{\eta})} >  \frac{\vbar_j(s)}{(\overline{w}^*_j + \overline{\eta})} \text { for all } j \neq i \right\}$. Denote by $\overline{\calL} = (\overline{L}_1, \ldots, \overline{L}_n)$ the partial allocation formed by these bundles and let $\overline{U} \coloneqq [\mbar] \setminus \left(\cup_{j=1}^n \overline{L}_j\right)$.

Next, we consider a bipartite graph $\overline{H} = ([n] \cup \overline{U}, \overline{E})$, with $[n]$ and $\overline{U}$ as the left and right part, respectively. Edge $(i,t) \in \overline{E} \subseteq [n] \times \overline{U}$ is included in the bipartite graph if and only if $p^*_t = \frac{\vbar_i(t)}{\overline{w}^*_i + \overline{\eta}}$. The vertices adjacent to any vertex $x$ in $\overline{H}$ will be denoted by $\overline{\Gamma}(x)$. Arguments similar to the ones used in \Cref{section:phaseII} establish the following lemma for $\overline{H}$.

\begin{lemma}\label{lemma:H-is-forest-V2}
The bipartite graph $\overline{H} = ([n] \cup \overline{U}, \overline{E})$ is a forest.
\end{lemma}

Define collection of complete allocations $\overline{O}^* \coloneqq \left\{ (B_1, \ldots, B_n) \in \Pi_n([\mbar]) \mid \overline{L}_i \subseteq B_i \subseteq \overline{L}_i \cup \overline{\Gamma}(i) \text{ for each } i \in [n] \right\}$. Analogous to \Cref{lemma:O-star-is-opt}, we have the following result for $\overline{O}^*$. 

\begin{proposition}\label{proposition:O-star-is-opt-V2}
   Any allocation $\calB \in \overline{O}^*$ if and only if $\calB \in \OPT(\overline{\textbf{\textrm{LP}}}(\overline{w}^*,\overline{\eta}))$.
\end{proposition}

\paragraph{Max-Min Allocations.} We now define some key constructs for the rest of the analysis. Note that while the high-level structure of the proof here is similar to \Cref{subsection:min-max-allocations}, multiple constructs differ from the previous ones. In particular, the price threshold $\overline{\tau}$ (defined below) has a max-min formulation in contrast to $\tau$ defined in \Cref{eqn:defn-tau}. 
$$\overline{\tau} \coloneqq \max_{(B_1,\ldots, B_n)\in \overline{O}^*} \ \  \min_{i \in [n]} p^*(B_i).$$

Define set $\overline{\calM}^* \coloneqq \{(B_1,\ldots, B_n) \in \overline{O}^* \mid  \min_{j \in [n]} p^*(B_j) = \overline{\tau}\}$ as the collection of all max-min allocations. Additionally, for each agent $i \in [n]$ and subsets $S$, we define 
$$p^-_i(S) \coloneqq \min\limits_{\substack{T \subseteq \overline{\Gamma}(i), \\ |T| \leq 1}} p^*(S \triangle T).$$

We also define, for each allocation $\allocb = (B_1, \ldots, B_n)$, the potential $\varphi^-(\allocb)$, which counts the number of agents whose bundle prices can be toggled to $\overline{\tau}$ or below. 
$$\varphi^-(\allocb)  \coloneqq \left|\{i \in [n] \mid p_i^-(B_i) \leq \overline{\tau}\} \right|.$$

The following theorem is a key result of this section, and it shows that there exists a  max-min allocation $\calA^*$ wherein the bundle prices of all the agents can be toggled to $\overline{\tau}$ or below.   

\begin{theorem}\label{theorem:leveled-allocation-is-fair-n-efficient-V2}
There exists an allocation $\calA^* \in \overline{\calM}^*$ such that $\varphi^-(\calA^*) = n$. 
\end{theorem}

 Before proving \Cref{theorem:leveled-allocation-is-fair-n-efficient-V2} (in \Cref{appendix:subsection-level} below), we use this result to establish the main result of this section, the existence of $\XEFo$ and $\PO$ allocations.

\extrospecitveEFo*
\begin{proof}
    Consider the weight vector $\overline{w}^* \in \Delta_{n-1}$ whose existence is implied by \Cref{theorem:KKM-application-V2} and let $p^* = \allowbreak \OPT(\overline{\textbf{\textrm{Dual-LP}}}(\overline{w}^*,\overline{\eta}))$. Let $\alloc^* \in \overline{\calM}^*$ be the allocation whose existence is implied by \Cref{theorem:leveled-allocation-is-fair-n-efficient-V2}. 

 First, note that $\alloc^* \in \overline{\calM}^* \subseteq \overline{O}^*$, therefore by \Cref{proposition:O-star-is-opt-V2}, we get $\alloc^* \in \OPT(\overline{\textbf{\textrm{LP}}}(\overline{w}^*,\overline{\eta}))$. 
 
 Next, we will show that $\alloc^*$ is $\XEFo$ for the instance $\overline{\calI}$. By \Cref{theorem:leveled-allocation-is-fair-n-efficient-V2}, we have $p^-_j(A^*_j) \leq \overline{\tau}$ for all $j \in [n]$. Since, $\alloc^* \in \overline{\calM}^*$ we also have $\overline{\tau} = \min_{i \in [n]} \  p^*(A^*_i)$. Hence, for any pair of agents $i,j \in [n]$, we have

    \begin{equation}\label{equation:pEF-to-EF-1-V2}
    \min\limits_{\substack{T \subseteq \overline{\Gamma}(j), \\ |T| \leq 1}} p^*(A^*_j \triangle T) = p^-_j(A^*_j) \leq \overline{\tau} \leq p^*(A^*_i).
    \end{equation}
Recall that $\alloc^* \in \overline{O}^*$ and, hence, $\alloc^* \in \OPT(\overline{\textbf{\textrm{LP}}}(\overline{w}^*,\overline{\eta}))$ (\Cref{proposition:O-star-is-opt-V2}). Therefore, for each item $t \in A^*_i \subseteq \overline{L}_i \cup \overline{\Gamma}(i)$ it holds that $p^*_t = \vbar_i(t)/(\overline{w}^*_i + \overline{\eta})$. Furthermore, for each item $s \in A^*_j \Delta T$, we have $p^*_s = \max_{a \in [n]} \vbar_a(s)/(\overline{w}^*_a + \overline{\eta}) \geq \vbar_i(s)/(\overline{w}^*_i + \overline{\eta})$. Using these bounds with \Cref{equation:pEF-to-EF-1-V2}, we obtain  
    \begin{equation}\label{equation:pEF-to-EF-2-V2}
        \frac{1}{\overline{w}^*_i + \overline{\eta}} \ \min\limits_{\substack{T \subseteq \overline{\Gamma}(j), \\ |T| \leq 1}} \vbar_i(A^*_j \triangle T) \leq \min\limits_{\substack{T \subseteq \overline{\Gamma}(j), \\ |T| \leq 1}} p^*(A^*_j \triangle T) \leq p^*(A^*_i) = \frac{1}{\overline{w}^*_i + \overline{\eta}} \ \vbar_i(A^*_i).
    \end{equation}
Since $(\overline{w}^*_i + \overline{\eta}) > 0$, we get that $\alloc^*$ is $\XEFo$ in the instance $\widehat{\calI}$: For each pair of agents $i,j \in [n]$, there exists a subset $T \subseteq \overline{\Gamma}(j) \subseteq [\mbar]$ of size at most one ($|T| \leq 1$) such that $\vbar_i(A^*_i) \geq \vbar_i(A^*_j \triangle T)$.

Therefore, we have that the allocation $\alloc^*$ is $\XEFo$ in $\overline{\calI}$ and $\alloc^* \in \OPT(\overline{\textbf{\textrm{LP}}}(\overline{w}^*,\overline{\eta}))$. Finally, \Cref{lemma:EFE-non-degenerate-V2,lemma:PO-non-degenerate-V2} imply that the allocation $(A^*_1 \setminus \{m+1\}, A^*_2 \setminus \{m+1\},\ldots, A^*_n \setminus \{m+1\})$ is $\XEFo$ and $\PO$ in the given instance $\calI$. This completes the proof of the theorem. 
\end{proof}

\subsection{Proof of \Cref{theorem:leveled-allocation-is-fair-n-efficient-V2}}
\label{appendix:subsection-level}
   \begin{algorithm}[h]
   \caption{\textsc{$\XEFo$AugmentTree Algorithm}}\label{algo:augmenting-forest-2}
   
   \SetKwInOut{Input}{Input}
   \SetKwInOut{Output}{Output}
   \SetKwFor{While}{while}{do}{end}
   \SetKwFor{ForAll}{for all}{do}{end}
   \SetKwIF{If}{ElseIf}{Else}{if}{then}{else if}{else}{end}
   
   \SetAlgoNlRelativeSize{-1}  
   \DontPrintSemicolon          

   \Input{Allocation $\calS = (S_1, \dots, S_n) \in \overline{\calM}^*$  (i.e., $\min_i p^*(S_i)=\overline{\tau}$) with $\varphi^-(\calS) < n$.}
   \Output{Allocation $\widehat{\calS} = (\widehat{S}_1, \dots, \widehat{S}_n) \in \overline{\calM}^*$ satisfying $\varphi^-(\widehat{\calS})  > \varphi^-(\calS)$.}
  
Let $r \in [n]$ be any agent with $p^-_r(S_r) > \overline{\tau}$. Initialize set $\calQ = \{r\}$.\label{line:root-node-in-queue-V2}\;

Write $R$ to denote the tree in bipartite graph (forest) $\overline{H}$ rooted at $r$. For each node $x$ in the tree $R$, let $\pr(x)$ denote the parent of $x$. In addition, $\ch(x) = \overline{\Gamma}(x) \setminus \{ \pr(x) \}$ are the children of $x$ in the tree. Also, for each agent $i$ in the tree, $\cho(i) = \cup_{t \in \ch(i)} \ \ch(t)$ denotes the agents that are immediate descendants of $i$.\label{line:rooting-the-tree-V2}\; 
Initialize $\widehat{S}_j = S_j$ for all agents $j \in [n]$.\;
   \BlankLine
   \While{$\calQ \neq \emptyset$}{
       Let $i \in \calQ$ be any agent in $\calQ$. Update $\calQ \gets \calQ \setminus \{i\}$. \label{loopInvariant:i-needs-a-fix}\; 
       \BlankLine
       
       Set $X \subseteq \ch(i)$ to be an item subset that satisfies: (i) $p^-_i( \widehat{S}_i \triangle X) \leq \overline{\tau}$,  (ii) $p^*(\widehat{S}_i \triangle X) \geq \overline{\tau}$, and (iii) $p^*(\widehat{S}_i \triangle \{t\}) < p^*( \widehat{S}_i)$ for each $t \in X$. 
       {\color{teal}\{Existence of $X$ is guaranteed via \Cref{lemma:forest-augmentation-existence-V2}.\}}\label{line:definition-of-X-V2}\;
       \BlankLine
       \ForAll{items $t \in X$}{
        If item $t \in \widehat{S}_i$, then set $a$ to be any agent in $\ch(t) \neq \emptyset$. Otherwise, if $t \notin \widehat{S}_i$, then let $a \in \ch(t)$ be the agent such that $t \in \widehat{S}_a$. 
      {\color{teal}{\{Such an agent $a \in \cho(i)$ always exists.\}}}\label{line:agent-a-V2}\;
      Update $\widehat{S}_{a} \gets \widehat{S}_{a} \triangle \{t\}$.\label{line:update-a-V2}\; 
           \BlankLine
      If $p^-_{a}(\widehat{S}_{a}) > \overline{\tau}$, then include $a$ for processing: $\calQ \gets \calQ \cup \{a\}$. {\color{teal}{\{Otherwise, if $p^-_{a}(\widehat{S}_{a}) \leq \overline{\tau}$, then we do not include $a$ in $\calQ$.\}}}\label{line:add-to-Q-V2}\;  
       } 
       \BlankLine
       Update $\widehat{S}_i \gets \widehat{S}_i \triangle X$\label{line:i-bulk-update-V2}\;
   }
   \Return{Allocation $\widehat{\calS} = (\widehat{S}_1,\ldots, \widehat{S}_n)$.} \;  
   \end{algorithm}

The proofs of the following three lemmas (Lemmas \ref{lemma:forest-augmentation-existence-V2} to \ref{lemma:each-agent-once-V2}) follow along the lines of the proofs of Lemmas \ref{lemma:forest-augmentation-existence} to \ref{lemma:each-agent-once}, and hence are omitted.

 \begin{lemma}\label{lemma:forest-augmentation-existence-V2}
 For any agent $i \in [n]$ in the rooted tree, let $S$ be an item subset such that $\overline{L}_i \subseteq S \subseteq \overline{L}_i \cup \overline{\Gamma}(i)$ and $p^-_i(S) > \overline{\tau}$. Then, there exists subset $X \subseteq \ch(i)$ satisfying the following properties
    \begin{enumerate}
       \item[(i)] $p^-_i(S \triangle X) \leq \overline{\tau}$,
       \item[(ii)] $p^*(S \triangle X) \geq \overline{\tau}$, and
       \item[(iii)] $p^*(S \triangle \{t\}) < p^*(S)$ for each $t \in X$.
    \end{enumerate}
 \end{lemma}

 \begin{lemma}
\label{lemma:loop-succeeds-V2}
Each iteration of the while-loop in the \textsc{$\XEFo$AugmentTree Algorithm} executes successfully and maintains an allocation $\widehat{\calS} \in \overline{\calM}^*$. Furthermore, at the end of the iteration (Line~\ref{line:i-bulk-update-V2}), for the agent $i$ selected in that iteration, it holds that $p^-_i(\widehat{S}_i) \leq \overline{\tau}$.
\end{lemma}

\begin{lemma}
\label{lemma:each-agent-once-V2}
Any agent $i \in [n]$ is included at most once in the set $\calQ$ throughout the execution of \textsc{$\XEFo$AugmentTree Algorithm}.
\end{lemma}

The corollary below follows from Lemmas \ref{lemma:loop-succeeds-V2} and \ref{lemma:each-agent-once-V2}.
\begin{corollary}
\label{corollary:termination-V2}
The \textsc{$\XEFo$AugmentTree Algorithm} (\Cref{algo:augmenting-forest-2}) always terminates and returns a complete allocation $\widehat{\calS} = (\widehat{S}_1, \ldots, \widehat{S}_n) \in \overline{\calM}^*$.
\end{corollary}

We now establish \Cref{theorem:leveled-allocation-is-fair-n-efficient-V2}.

\begin{proof}[Proof of \Cref{theorem:leveled-allocation-is-fair-n-efficient-V2}]
To prove the theorem it suffices to show that, for any input allocation $\calS \in \overline{\calM}^*$ with $\varphi^-(\calS) < n$, \textsc{$\XEFo$AugmentTree Algorithm} returns another allocation $\widehat{\calS} \in \overline{\calM}^*$ with a strictly higher $\varphi^-(\cdot)$ value. In particular, by repeatedly applying the algorithm, we obtain the desired allocation $\calA^* \in \overline{\calM}^*$ with $\varphi^-(\calA^*) = n$, thereby establishing the theorem.

\Cref{corollary:termination-V2} ensures that the algorithm successfully finds $\widehat{\calS} \in \overline{\calM}^*$. We now complete the proof by showing that $\varphi^-(\widehat{\calS}) > \varphi^-(\calS)$.

Note that, at the end of the first iteration, the root agent $r$ receives a bundle with $p^-_r(\cdot)$ value at most $\overline{\tau}$ (\Cref{lemma:loop-succeeds-V2}). Since $r$’s bundle is not modified thereafter (\Cref{lemma:each-agent-once-V2}), we have $p^-_r(\widehat{S}_r) \leq \overline{\tau}$ in the returned allocation as well.
 
 Moreover, for any agent $j \in [n]$, if in the input allocation $\calS$ we have $p^-_j(S_j) \leq \overline{\tau}$, then $p^-_j(\widehat{S}_j) \leq \overline{\tau}$ in the returned allocation $\widehat{\calS}$ as well: If $j$’s bundle was not updated during the algorithm, the assertion holds trivially, since $\widehat{S}_j = S_j$. Otherwise, if $j$’s bundle was updated and its $p^-_j(\cdot)$ value became more than $\overline{\tau}$, then, by $j$’s inclusion in $\calQ$ and the subsequent updates, the algorithm ensures that $p^-_j(\widehat{S}_j) \leq \overline{\tau}$ at termination; recall that $\calQ = \emptyset$ when the algorithm terminates. 
 
This observation and the inclusion of $r$ (as an agent which satisfies $p^-_r(\widehat{S}_r) \leq \overline{\tau}$) imply that $\varphi^-(\widehat{\calS}) > \varphi^-(\calS)$. The theorem stands proved. 
\end{proof}

The following theorem follows from \Cref{theorem:OutPO} via a limiting argument similar to the one used in \Cref{section:addingFtoPO} to prove \Cref{theorem:addingFtoPO}.

\begin{theorem}
\label{theorem:OutfPO}
Every fair division instance $\calI = \angleb{[n], [m], \{v_i\}_{i=1}^n}$, with indivisible mixed manna and additive valuations, admits an allocation that is both extrospectively envy-free up to one item ($\XEFo$) and fractionally Pareto efficient ($\fPO$).
\end{theorem}

\section{Weighted $\XEFo$ and $\fPO$ Allocations}

Our existence guarantees continue to hold when the agents have unequal entitlements (weights). A
\emph{weighted} fair division instance is a tuple $\mathcal{I} = \langle [n], [m], \{v_i\}_{i=1}^{n},
\mathbf{e} \rangle$, where, as before, $[n]$ denotes the set of agents, $[m]$ the set of indivisible
items, and $v_i$ the additive valuation of agent $i \in [n]$. In addition,
$\mathbf{e} = (e_1, e_2, \ldots, e_n) \in \mathbb{R}^n_{>0}$ is a vector of \emph{entitlements} where for each agent $i \in [n]$, $e_i > 0$ is a cardinal measure of agent $i$'s entitlement.
In the weighted setting, agent $i$'s comparison of its own bundle with that of agent $j$ is scaled by
the two entitlements: agent $i$ is said to envy agent $j$ in an allocation
$A = (A_1, \ldots, A_n)$ if $\frac{v_i(A_i)}{e_i} < \frac{v_i(A_j)}{e_j}$. Next we define the weighted versions of $\IEFo$ and $\XEFo$.
 
\begin{definition}[wIEF1]
\label{def:wief1}
In a weighted instance $\mathcal{I} = \langle [n], [m], \{v_i\}_{i=1}^{n}, \mathbf{e} \rangle$, an
allocation $A = (A_1, \ldots, A_n) \in \Pi_n([m])$ is said to be \emph{weighted introspectively
envy-free up to one item} (wIEF1) if for each agent $i \in [n]$ there exists a subset $S \subseteq [m]$
of size at most one ($|S| \leq 1$) such that
\[
  \frac{v_i(A_i \triangle S)}{e_i} \;\geq\; \max_{1 \leq j \leq n} \frac{v_i(A_j)}{e_j}.
\]
\end{definition}
 
\begin{definition}[wXEF1]
\label{def:wxef1}
In a weighted instance $\mathcal{I} = \langle [n], [m], \{v_i\}_{i=1}^{n}, \mathbf{e} \rangle$, an
allocation $A = (A_1, \ldots, A_n) \in \Pi_n([m])$ is said to be \emph{weighted extrospectively
envy-free up to one item} (wXEF1) if for all agents $i, j \in [n]$ there exists a subset
$S \subseteq [m]$ of size at most one ($|S| \leq 1$) such that
\[
  \frac{v_i(A_i)}{e_i} \;\geq\; \frac{v_i(A_j \triangle S)}{e_j}.
\]
\end{definition}
 
Under equal entitlements---i.e., with $e_i = 1/n$ for all $i \in [n]$---Definitions~\ref{def:wief1}
and~\ref{def:wxef1} reduce to Definitions~\ref{definition:iefo} and~\ref{definition:XEFo}, respectively.
 
The following two theorems extend Theorems~\ref{theorem:addingFtoPO} and~\ref{theorem:OutfPO} to agents with
arbitrary entitlements.
 
\begin{theorem}
\label{thm:wief1-fpo}
Every weighted fair division instance $\mathcal{I} = \langle [n], [m], \{v_i\}_{i=1}^{n}, \mathbf{e}
\rangle$, with indivisible mixed manna, additive valuations, and positive entitlements, admits an
allocation that is both weighted introspectively envy-free up to one item (wIEF1) and fractionally
Pareto efficient (fPO).
\end{theorem}
 
\begin{theorem}
\label{thm:wxef1-fpo}
Every weighted fair division instance $\mathcal{I} = \langle [n], [m], \{v_i\}_{i=1}^{n}, \mathbf{e}
\rangle$, with indivisible mixed manna, additive valuations, and positive entitlements, admits an
allocation that is both weighted extrospectively envy-free up to one item (wXEF1) and fractionally
Pareto efficient (fPO).
\end{theorem}
 
We will omit the proofs given that they closely mimic the proofs of the equal entitlement setting
presented in the paper.